\documentclass[10pt,a4paper,oneside]{article}
\usepackage[utf8]{inputenc}
\usepackage[english]{babel}
\pdfoutput=1

\usepackage{amsmath,amsfonts,amssymb,amsopn,amscd,amsthm}
\allowdisplaybreaks

\usepackage[yyyymmdd,hhmmss]{datetime}

\usepackage{fancyhdr}
\usepackage{subfig}

\newcommand{\Abstract}{
\begin{abstract}
We present several new sufficient conditions for uniform boundedness of the reduced correlations and free energy of an abstract polymer system in a complex multidisc around zero fugacity. They resolve a discrepancy between two incomparable and previously known extensions of Dobrushin's  classic condition. All conditions arise from an extension of the tree-operator approach introduced by Fernández and Procacci combined with a novel family of partition schemes of the spanning subgraph complex of a cluster. The key technique is the increased transfer of structural information from the partition scheme to a tree-operator on an enhanced space.
\end{abstract}
}

\newcommand{\TitleFull}{Sufficient conditions for uniform bounds in abstract polymer systems and explorative partition schemes}

\newcommand{\TitleShort}{SCUBs and partition schemes}
\newcommand{\TitlePDF}{SCUBs\ and\ partition\ schemes} 

\newcommand{\AuthorsFull}{Temmel Christoph (math@temmel.me)\thanks{The author acknowledges the support of the VIDI project ``Phase transitions, Euclidean fields and random fractals'', NWO 639.032.916}}

\newcommand{\AuthorsShort}{Temmel}

\newcommand{\Keywords}{Keywords:
cluster expansion,
abstract polymer system,
partition scheme,
tree-operator,
hardcore gas}

\newcommand{\Head}{
 \maketitle
 \Abstract{}
 \Keywords{}\\\MSC{}\\
 
 \noindent This is an extended version of~\cite{Temmel__SufficientConditionsForUniformBoundsInAbstractPolymerSystemsAndExplorativePartitionSchemes__JSP_2014}.
 \tableofcontents
 \listoffigures
}
\newtheorem{Thm}{Theorem}
\newtheorem{Lem}[Thm]{Lemma}
\newtheorem{Prop}[Thm]{Proposition}

\newtheorem{Alg}[Thm]{Algorithm}
\newtheorem{Que}[Thm]{Question}
\newtheorem{Exam}[Thm]{Example}

\theoremstyle{remark}
\newtheorem*{Rem}{Remark}

\newcommand{\Then}{\,\Rightarrow\,}
\newcommand{\Iff}{\,\Leftrightarrow\,}
\newcommand{\ForAll}{\forall\,}
\newcommand{\Exists}{\exists\,}
\newcommand{\NExists}{\nexists\,}
\newcommand{\Formally}[1]{\stackrel{F}{#1}}
\newcommand{\NatNum}{\mathbb{N}}
\newcommand{\NatNumZero}{\NatNum_0} 
\newcommand{\ComplexNum}{\mathbb{C}}
\newcommand{\Set}[1]{{\{#1\}}}
\newcommand{\BigSet}[1]{{\left\{#1\right\}}}
\newcommand{\IntRange}[1]{{[#1]}}
\newcommand{\Cardinality}[1]{{|#1|}}
\newcommand{\Modulus}[1]{|#1|}
\newcommand{\Iverson}[1]{{[#1]}}
\newcommand{\Identity}[1]{{\text{id}_{#1}}}
\newcommand{\Factorial}[1]{#1!}
\newcommand{\FiniteSubsetOf}{\Subset}
\DeclareMathOperator{\Support}{supp}
\DeclareMathOperator{\InteriorOf}{Int}

\newcommand{\Tree}{\mathbb{T}}
\newcommand{\TreePath}[2]{{P(#1,#2)}}
\newcommand{\Parent}[1]{{\mathfrak{p}(#1)}}
\newcommand{\Confluent}{{\curlywedge}}

\newcommand{\Edges}[1]{E(#1)}
\newcommand{\Vertices}[1]{V(#1)}
\newcommand{\GraphMetricSymbol}{d}
\newcommand{\GraphMetricOf}[3]{\GraphMetricSymbol_{#1}(#2,#3)}
\newcommand{\ConnectedTo}{\leftrightarrow}
\newcommand{\AdjacentTo}{\backsim}
\newcommand{\SubGraph}{\le}

\newcommand{\ActionRemove}{REMOVE}
\newcommand{\ActionSelect}{SELECT}

\newcommand{\AdmissibleEdgesTree}[1]{\mathcal{A}_{#1}(\Tree)}
\newcommand{\Cluster}[1]{G(#1)}
\newcommand{\CompatibleTo}{{\,\not\approx\,}}

\newcommand{\ConflictingEdgesTree}[1]{\mathcal{C}_{#1}(\Tree)}

\newcommand{\EscapePairSet}{\PolymerSet_{\star}}
\newcommand{\EscapeInjection}{\mathfrak{i}}
\newcommand{\EscapePair}{{(\Polymer,\EscapePolymer)}}
\newcommand{\EscapePolymer}{\varepsilon}
\newcommand{\Exhausts}{\nearrow}
\newcommand{\ExplorationAlgorithm}[1]{\mathcal{E}_{#1}}
\newcommand{\FiniteVolume}{\Lambda}
\newcommand{\FreeEnergy}[1]{F_{#1}}
\newcommand{\IncompatibleTo}{{\,\approx\,}}
\newcommand{\IncompatiblePolymers}[1]{\mathcal{I}(#1)}
\newcommand{\IncompatibleOtherPolymers}[1]{\mathcal{I}^\star(#1)}
\newcommand{\LabelSpace}{\mathcal{L}}

\newcommand{\Leop}{\phi} 
\newcommand{\LeopGeneric}[1]{\Leop_{#1}^{\text{gen}}}

\newcommand{\LeopKP}[1]{\Leop_{#1}^{\text{KP}}}
\newcommand{\LeopDob}[1]{\Leop_{#1}^{\text{Dob}}}
\newcommand{\LeopFP}[1]{\Leop_{#1}^{\text{FP}}}

\newcommand{\LeopRed}[1]{\Leop_{#1}^{\text{red}}}
\newcommand{\LeopRet}[1]{\Leop_{#1}^{\text{ret}}}
\newcommand{\LeopRetMulti}[1]{\Leop_{#1}^{\text{m-ret}}}

\newcommand{\LeopSup}[1]{\Leop_{#1}^{\text{sup}}}
\newcommand{\LeopMulti}[1]{\Leop_{#1}^{\text{mul}}}

\newcommand{\LeopSynthetic}[2]{\Leop_{#2}^{\text{syn}(#1)}}
\newcommand{\LeopMixing}[2]{\Leop_{#2}^{\text{mix}(#1)}}

\newcommand{\OPPR}[2]{\varphi_{#1}^{#2}} 
\newcommand{\ParameterDisc}[1]{\mathcal{R}_{#1}}
\newcommand{\PartitionFunction}[1]{\Xi_{#1}}
\newcommand{\PDPS}{\ParameterDisc{\PolymerSet}}
\newcommand{\PinnedConnectedExpression}[2]{
 \frac%
  {\partial\log\PartitionFunction{#1}}
  {\partial z_{#2}}
}
\newcommand{\PinnedSeries}[1]{{\Psi_{#1}}}
\newcommand{\PinnedSeriesEscaping}[1]{{\Psi_{#1}^{\star}}}

\newcommand{\Polymer}{\gamma}
\newcommand{\PolymerOther}{\xi}
\newcommand{\PolymerSet}{\mathcal{P}}
\newcommand{\PolymersMinusEscapeSet}{\PolymerSet\setminus\Set{\EscapePolymer}}
\newcommand{\PowerFracDual}[2]{{\PowerFracTriple{#1}{\phantom{}}{#2}}}

\newcommand{\PowerFracDualDMinusOne}{\PowerFracDual{(D-1)}{D}}
\newcommand{\PowerFracTriple}[3]{{\frac{#1^{#1}#2^{#2}}{#3^{#3}}}}
\newcommand{\ReducedCorrelation}[2]{\Phi_{#1}^{#2}}

\newcommand{\ReturningClass}[2]{[#1]_{(#2)}}
\newcommand{\ReturningEquivalent}[1]{\sim_{(#1)}}
\newcommand{\ReturningEquivalentKMinusOne}[1]{\ReturningEquivalent{k-1}}
\newcommand{\ReturningEquivalentK}{\ReturningEquivalent{k}}
\newcommand{\ReturningEquivalentKPlusOne}{\ReturningEquivalent{k+1}}

\newcommand{\ReturningStatus}{s}
\newcommand{\ReturningSame}{\mathbf{S}}
\newcommand{\ReturningDifferent}{\mathbf{D}}

\newcommand{\RTVarDepth}{t}
\newcommand{\RTVarLabel}{\tau}
\newcommand{\RTVarSize}{\tau}
\newcommand{\RTLevel}[2]{L_{#2}(#1)}
\newcommand{\RTNonRootVertices}[1]{W(#1)}
\newcommand{\RTPermutationCount}[1]{\mathfrak{p}(#1)}
\newcommand{\RTKSubtreeAt}[2]{{#1}_k^{#2}}
\newcommand{\RTsDepth}[1]{\mathfrak{T}^{=}_{#1}}
\newcommand{\RTsDepthLessThan}[1]{\mathfrak{T}^{\le}_{#1}}
\newcommand{\RTsDepthFinite}{\mathfrak{T}^{<}_{\infty}}

\newcommand{\RTsSize}[1]{\mathcal{T}_{#1}}
\newcommand{\RTsSizeFinite}{\RTsSize{\infty}}

\newcommand{\Stars}{\mathcal{S}}
\newcommand{\StarAt}[2]{S_{#1}^{#2}}
\newcommand{\StarLeaves}[1]{L(#1)}
\newcommand{\StarPermutations}[1]{\Cardinality{\StarLeaves{#1}}!}

\newcommand{\SchemeAnonymous}{S}
\newcommand{\SchemeGeneric}{Gen}
\newcommand{\SchemePenrose}{Pen}
\newcommand{\SchemeReturning}{Ret}
\newcommand{\SchemeSynthetic}{Syn}
\newcommand{\SingletonTrees}[1]{\mathcal{T}_{#1}}
\newcommand{\SpanningSubgraphs}[1]{\mathcal{C}_{#1}}
\newcommand{\SpanningTrees}[1]{\mathcal{T}_{#1}}

\newcommand{\StepStyle}[1]{\textbf{(#1)}}

\newcommand{\StepPenCousins}{\StepStyle{pen-cousins}}
\newcommand{\StepPenElders}{\StepStyle{pen-elders}}
\newcommand{\StepPenParent}{\StepStyle{pen-parent}}
\newcommand{\StepPenUncles}{\StepStyle{pen-uncles}}

\newcommand{\StepGenBoundary}{\StepStyle{gen-boundary}}
\newcommand{\StepGenIgnored}{\StepStyle{gen-ignored}}
\newcommand{\StepGenParent}{\StepStyle{gen-parent}}
\newcommand{\StepGenUncles}{\StepStyle{gen-uncles}}
\newcommand{\StepGenCousins}{\StepStyle{gen-cousins}}

\newcommand{\StepGreedyBoundary}{\StepStyle{greedy-boundary}}
\newcommand{\StepGreedyIgnored}{\StepStyle{greed-ignored}}
\newcommand{\StepGreedyParent}{\StepStyle{greedy-parent}}
\newcommand{\StepGreedyUncles}{\StepStyle{greedy-uncles}}
\newcommand{\StepGreedyCousins}{\StepStyle{greedy-cousins}}

\newcommand{\StepRetSameBoundary}{\StepStyle{ret-same-boundary}}
\newcommand{\StepRetSameIgnored}{\StepStyle{ret-same-ignored}}
\newcommand{\StepRetSameParent}{\StepStyle{ret-same-parent}}
\newcommand{\StepRetSameUncles}{\StepStyle{ret-same-uncles}}
\newcommand{\StepRetSameCousins}{\StepStyle{ret-same-cousins}}

\newcommand{\StepRetDiffBoundary}{\StepStyle{ret-diff-boundary}}
\newcommand{\StepRetDiffIgnored}{\StepStyle{ret-diff-ignored}}
\newcommand{\StepRetDiffParent}{\StepStyle{ret-diff-parent}}
\newcommand{\StepRetDiffUnclesDifferent}{\StepStyle{ret-diff-uncles-diff}}
\newcommand{\StepRetDiffUnclesSame}{\StepStyle{ret-diff-uncles-same}}
\newcommand{\StepRetDiffCousins}{\StepStyle{ret-diff-cousins}}

\newcommand{\StepSynGreedyBoundary}{\StepStyle{syn-greedy-boundary}}
\newcommand{\StepSynGreedyIgnored}{\StepStyle{syn-greedy-ignored}}
\newcommand{\StepSynGreedyParent}{\StepStyle{syn-greedy-parent}}
\newcommand{\StepSynGreedyUncles}{\StepStyle{syn-greedy-uncles}}
\newcommand{\StepSynGreedyCousins}{\StepStyle{syn-greedy-cousins}}

\newcommand{\StepSynRetBoundary}{\StepStyle{syn-ret-boundary}}
\newcommand{\StepSynRetIgnored}{\StepStyle{syn-ret-ignored}}
\newcommand{\StepSynRetParent}{\StepStyle{syn-ret-parent}}
\newcommand{\StepSynRetUnclesReturning}{\StepStyle{syn-ret-uncles-ret}}
\newcommand{\StepSynRetUnclesGreedy}{\StepStyle{syn-ret-uncles-greedy}}
\newcommand{\StepSynRetCousins}{\StepStyle{syn-ret-cousins}}

\newcommand{\BehaviourGreedy}{\mathbf{G}}
\newcommand{\BehaviourReturning}{\mathbf{R}}
\newcommand{\BehaviourSetIndexedBy}[1]{\Set{\BehaviourGreedy,\BehaviourReturning}^{#1}}

\newcommand{\MixingVector}{\vec{g}}
\newcommand{\SyntheticVector}{\vec{g}}

\newcommand{\IncompatiblesBehaviourReturningIn}[2]{\operatorname{In}^{#1}_{#2}}
\newcommand{\IncompatiblesBehaviourReturningOut}[2]{\operatorname{Out}^{#1}_{#2}}

\newcommand{\Ursell}[1]{\mathfrak{u}(#1)}

\usepackage{tikz}

\usetikzlibrary{shapes.geometric} 

\tikzstyle{vertex}  = [circle, draw, minimum size = 17pt, inner sep = 0pt]
\tikzstyle{vertexlong}  = [ellipse, draw, inner sep = 2pt]
\tikzstyle{marked} = [very thick, color = red,densely dotted]

\newcommand{\selfedge}[2]{\draw[out=#2-30, in=#2+30, looseness=7] (#1) to (#1);}
\newcommand{\selfedgemark}[2]{\draw[out=#2-30, in=#2+30, looseness=7, marked] (#1) to (#1);}
\newcommand{\drawnode}[2]{\node[vertex] (#1) at (#2) {$#1$};}
\newcommand{\drawnodebullet}[2]{\node[inner sep = 1pt] (#1) at (#2) {$\bullet$};}
\newcommand{\drawnodeempty}[2]{\node (#1) at (#2) {};}
\usepackage{xy}
\xyoption{arrow}
\xyoption{graph}
\xyoption{matrix}

\usepackage{rotating}

\newcommand{\MxyLeopComparison}[2]{
\xymatrix@C=#1@R=#2{
  *+{}
  &*+{}
  &*+{}
  &*+{\LeopFP{\Polymer}(\vec{\mu})}
  &*+{}
  \\*+{\LeopKP{\Polymer}(\vec{\mu})}
     \ar@{}[r]|-{\begin{turn}{0}$\le$\end{turn}}
  &*+{\LeopDob{\Polymer}(\vec{\mu})}
     \ar@{}[urr]|-{\begin{turn}{30}$\le$\end{turn}}
  &*+{}
  &*+{}
  &*+{}
  &*+{\LeopMixing{\MixingVector}{\Polymer}}"
    \ar@{}[ull]|-{
      \begin{turn}{-30}
      $\stackrel{\MixingVector=\vec{\BehaviourGreedy}}{\le}$
      \end{turn}
    }
  \\*+{}
  &*+{}
  &*+{\LeopRed{\Polymer}(\vec{\mu})}
    \ar@{}[ul]|-{\begin{turn}{-45}$\le$\end{turn}}
  &*+{}
  &*+{\LeopRet{\Polymer}(\vec{\mu})}
    \ar@{}[ll]|-{\begin{turn}{0}$\le$\end{turn}}
    \ar@{}[ur]|-{
      \begin{turn}{45}
      $\stackrel{\MixingVector=\vec{\BehaviourReturning}}{\le}$
      \end{turn}
    }
  &*+{}
}}

\newcommand{\OnlyPrivate}[1]{}
\newcommand{\AlsoArxiv}[1]{#1}

\usepackage[pdftex%
 ,bookmarks=true%
 ,bookmarksopen=true%
 ,bookmarksopenlevel=3%
 ,colorlinks=false%
 ,pdftitle=Arxiv\ \pdfdate:\ \TitlePDF{}%
]{hyperref}

\title{\TitleFull{}}
\author{\AuthorsFull{}}
\date{}
\begin{document}

\Head{}

\section{Introduction}
\label{sec_introduction}

A recurring theme in statistical mechanics is the search for improved lower bounds of the radius of analyticity of the cluster expansion of a locally finite abstract polymer system around zero fugacity. This demands a convergence argument and a sufficient condition for uniform boundedness (short SCUB) of finite volume quantities.\\

The first SCUBs have been derived by Gruber \& Kunz~\cite{Gruber_Kunz__GeneralPropertiesOfPolymerSystems__CMP_1971}, still in the context of subset polymer systems, and Koteck\'{y} \& Preiss~\cite{Kotecky_Preiss__ClusterExpansionForAbstractPolymerModels__CMP_1986}. Both are based on cluster expansion techniques. A major improvement, in particular because of its short proof, has been Dobrushin's SCUB~\cite{Dobrushin__EstimatesOfSemiInvariantsForTheIsingModelAtLowTemperatures__AMST_1996}. It first bounds the reduced correlations, that is ratios of partition functions, which entails the boundedness of the free energy.\\

Recent work by Fern\'{a}ndez \& Procacci~\cite{Fernandez_Procacci__ClusterExpansionForAbstractPolymerModels_NewBoundsFromAnOldApproach__CMP_2007} uses an identity by Penrose~\cite{Penrose__ConvergenceOfFugacityExpansionsForClassicalSystems__SMFA_1967} to convert the problem of bounding the cluster expansion to determining a non-trivial fixpoint of a tree-operator. Their SCUB improves on Dobrushin's SCUB by not only accounting for the number of polymers incompatible with each polymer, but also their mutual incompatibilities. They show that all previously known SCUBs arise within their framework of tree-operators and are relaxations of their SCUB.\\

There is another Dobrushin-style SCUB building on work of Scott \& Sokal~\cite{Scott_Sokal__TheRepulsiveLatticeGasTheIndependentSetPolynomialAndTheLovaszLocalLemma__JSP_2005}, which reduces the degree of Dobrushin's SCUB by one. It is better than Fern\'{a}ndez \& Procacci's SCUB on triangle-free graphs and optimal on trees, but may be worse on general graphs.\\

The motivation of this paper is to resolve the above discrepancy. We want to interpret the above reduced-degree SCUB in the context of tree-operators, improve it analogously to Fern\'{a}ndez \& Procacci's work and clarify its relation with the already known SCUBs. We present several improvements of the reduced-degree SCUB of increasing strength and complexity. All new and known SCUBs arise in an extended tree-operator framework and are specialisations of a strongest SCUB. \\

The key mathematical technique is a new family of partition schemes of the spanning subgraph complex of a cluster (in the context of a cluster expansion of the partition function). These partition schemes encode more of the local structure of a cluster induced by the polymer system. This yields tighter relaxations of Penrose's theorem, a classic tree-graph identity. An extended tree-operator framework uses the additional information from the tighter relaxation to derive SCUBs with an increased domain. We consider this approach to partition schemes and tree-operators amenable to future exploitation.\\

There is a deep connection between the Lov\'{a}sz Local Lemma~\cite{Erdos_Lovasz__ProblemsAndResultsOn3ChromaticHypergraphsAndSomeRelatedQuestions__CMSJB_1975} and the partition function of the abstract polymer system at negative real fugacity discovered by Scott \& Sokal~\cite{Scott_Sokal__TheRepulsiveLatticeGasTheIndependentSetPolynomialAndTheLovaszLocalLemma__JSP_2005}. As already done for Fern\'{a}ndez \& Procacci's SCUB~\cite{Bissacot_Fernandez_Proccaci_Scoppola__AnImprovementOfTheLovaszLocalLemmaViaClusterExpansion__CPC_2011}, the new SCUBs also improve the Lov\'{a}sz Local Lemma.\\

The organisation of this paper is as follows: In the following subsections we introduce notation, key quantities and basic identities. Section~\ref{sec_scub} lists the SCUBs and discusses them. Section~\ref{sec_inductivePolymerLevelProofs} contains Dobrushin-style inductive proofs. Section~\ref{sec_treeOperatorFramework} contains a review of the tree-operator framework and the common part of the proof of the new SCUBs. Section~\ref{sec_partitionSchemes} is all about our new family of partition schemes. Finally, section~\ref{sec_proof} combines the explorative partition schemes and the tree-operator framework to derive the new SCUBs.

\subsection{Setup and notation}
\label{sec_setupAndNotation}
We consider an at most countable set of polymers $\PolymerSet$ with a symmetric and reflexive incompatibility relation $\IncompatibleTo$. We choose $\IncompatibleTo$ instead of the traditional $\CompatibleTo$ to mimic the standard graph theoretic notation for adjacency in the induced graph $(\PolymerSet,\IncompatibleTo)$. We assume that the graph $(\PolymerSet,\IncompatibleTo)$ is connected. It may be locally infinite. For each polymer $\Polymer$, we write $\IncompatiblePolymers{\Polymer}$ for the set of polymers incompatible with $\Polymer$, and $\IncompatibleOtherPolymers{\Polymer}:= \IncompatiblePolymers{\Polymer}\setminus\Set{\Polymer}$ for the set of polymers incompatible with and different than $\Polymer$.\\

For every finite subset of polymers $\FiniteVolume\FiniteSubsetOf\PolymerSet$, we define the \emph{grand canonical partition function} $\PartitionFunction{\FiniteVolume}: \ComplexNum^\FiniteVolume\to\ComplexNum$ by
\begin{equation}\label{eq_partitionFunction}
 \PartitionFunction{\FiniteVolume}(\vec{z}):=
  \sum_{n\ge 0} \frac{1}{\Factorial{n}}
  \sum_{\vec{\PolymerOther}\in\FiniteVolume^n}
   \left(
    \prod_{1\le i<j\le n} \Iverson{\PolymerOther_i\CompatibleTo\PolymerOther_j}
   \right)
  \prod_{i=1}^n z_{\PolymerOther_i}
 \,,
\end{equation}
where $\vec{z}$ are the \emph{activities} or \emph{fugacities} on $\FiniteVolume$ and $\Iverson{.}$ is an Iverson bracket. In particular, $\PartitionFunction{\emptyset}(\vec{z})=1$. The partition function is the generating function of weighted compatible subsets of $\FiniteVolume$. \AlsoArxiv{The function $\PartitionFunction{\FiniteVolume}(\vec{z})$ is affine in each parameter $z_\Polymer$.}\\

We lift scalar arithmetic operations and comparisons component-wise to vectors. The \emph{projection} of a vector $\vec{x}:=(x_\Polymer)_{\Polymer\in\PolymerSet}$ along the coordinates indexed by $\FiniteVolume\FiniteSubsetOf\PolymerSet$ is $\vec{x}_\FiniteVolume$, if needed for disambiguation, otherwise silently ignoring superfluous coordinates.
\subsection{Key quantities and identities}
\label{sec_key}
This section states the key quantities of interest and relations between them. Every definition and relation is a priori just formal, as not all the divisions are well-defined for all complex $\vec{z}$. They are well-defined, though, if one of the SCUBs is satisfied. There are two principal quantities of interest. The first is the \emph{reduced correlation}~\cite{Gruber_Kunz__GeneralPropertiesOfPolymerSystems__CMP_1971} of $m$ distinct polymers $\PolymerOther_1,\dotsc,\PolymerOther_m\in\FiniteVolume$
\begin{subequations}\label{eq_quantitiesOfInterest}
\begin{equation}\label{eq_reducedCorrelation}
 \ReducedCorrelation{\FiniteVolume}{\PolymerOther_1,\dotsc,\PolymerOther_m}(\vec{z})
 :=\frac%
  {\PartitionFunction{
   \FiniteVolume\setminus
   \bigcup_{i=1}^m\IncompatiblePolymers{\PolymerOther_i}
   }
   (\vec{z})
  }
  {\PartitionFunction{\FiniteVolume}(\vec{z})}\,.
\end{equation}
The second is the \emph{free energy}
\begin{equation}\label{eq_freeEnergy}
 \FreeEnergy{\FiniteVolume}(\vec{z})
 := -\frac%
  {\log\PartitionFunction{\FiniteVolume}(\vec{z})}
  {\Cardinality{\FiniteVolume}}\,.
\end{equation}
\end{subequations}
Our aim are bounds on these quantities independently of $\FiniteVolume$, that is in the \emph{thermodynamic limit}. From here on, we assume that $\Polymer\in\FiniteVolume$. The principal quantities of interest are the \emph{pinned connected function}~\cite[(2.24)]{Faris__CombinatoricsAndClusterExpansions__ProbSur_2010}
\begin{subequations}
\begin{equation}\label{eq_pinnedConnectedFunction}
 \PinnedConnectedExpression{\FiniteVolume}{\Polymer}(\vec{z})
\end{equation}
and the \emph{one polymer partition ratios}
\begin{equation}\label{eq_onePolymerRatio}
 \OPPR{\FiniteVolume}{\Polymer}(\vec{z})
 :=\frac%
  {\PartitionFunction{\FiniteVolume}(\vec{z})}
  {\PartitionFunction{\FiniteVolume\setminus\Set{\Polymer}}(\vec{z})}\,.
\end{equation}
\end{subequations}
We recall below how each of these two quantities suffices to express all other quantities. For $\Set{\PolymerOther_1,\dotsc,\PolymerOther_n}:=\FiniteVolume\FiniteSubsetOf\PolymerSet$ and $\FiniteVolume_i:=\Set{\PolymerOther_1,\dotsc,\PolymerOther_i}$, the \emph{telescoping identity} expresses the partition function as a product of one polymer partition ratios:
\begin{equation}\label{eq_telescoping}
 \PartitionFunction{\FiniteVolume}(\vec{z})
 =\prod_{i=1}^n \OPPR{\FiniteVolume_i}{\PolymerOther_i}(\vec{z})\,.
\end{equation}
Similarly, each reduced correlation is a product of inverses of suitably chosen one polymer partition ratios. The pinned connected function is a product of reduced correlations~\cite[(3.8)]{Scott_Sokal__TheRepulsiveLatticeGasTheIndependentSetPolynomialAndTheLovaszLocalLemma__JSP_2005}. The logarithm of the reduced correlations is an integral over the pinned connected function~\cite[(A.3)]{Bissacot_Fernandez_Proccaci__OnTheConvergenceOfClusterExpansionsForPolymerGases__JSP_2010}.
\AlsoArxiv{See section~\ref{sec_pinnedVsOPPR} for the details.}

\subsection{Cluster expansion and the worst case}
\label{sec_worstCase}
Let $I$ be a finite set. A vector $\vec{\PolymerOther}:=(\PolymerOther_i)_{i\in I}\in\PolymerSet^I$ has \emph{support}
\begin{equation}\label{eq_support}
 \Support\vec{\PolymerOther}:=
 \Set{\Polymer\in\PolymerSet:\quad\Exists i\in I:\quad\PolymerOther_i=\Polymer}\,.
\end{equation}
The vector $\vec{\PolymerOther}$ \emph{induces} the graph
\begin{equation}\label{eq_cluster}
 \Cluster{\vec{\PolymerOther}}
 :=(I,\Set{(i,j)\in I^2:\quad\PolymerOther_i\IncompatibleTo\PolymerOther_j})\,.
\end{equation}
We call $\Cluster{\vec{\PolymerOther}}$ a \emph{cluster}, if it is connected.
Define the \emph{Ursell functions}~\cite{Ursell__TheEvaluationOfGibbsPhaseIntegralForImperfectGases__MPCPS_1927} (or \emph{semi-invariants}~\cite{Dobrushin__EstimatesOfSemiInvariantsForTheIsingModelAtLowTemperatures__AMST_1996} or \emph{truncated functions}~\cite{Scott_Sokal__TheRepulsiveLatticeGasTheIndependentSetPolynomialAndTheLovaszLocalLemma__JSP_2005}) as
\begin{equation}\label{eq_ursellFunction}
 \Ursell{\vec{\PolymerOther}} :=
 \begin{cases}
  1
  &\text{if }\Cardinality{I}= 1\,,\\
  \displaystyle\sum_{H\text{ spans }\Cluster{\vec{\PolymerOther}}} (-1)^{\Cardinality{\Edges{H}}}
  &\text{if }\Cardinality{I}\ge 2\text{ and }\Cluster{\vec{\PolymerOther}}\text{ connected,}\\
  0
  &\text{else.}
 \end{cases}
\end{equation}
The \emph{cluster expansion}~\cite{MiracleSole__OnTheTheoryOfClusterExpansions__MPRF_2010}\cite[section 2.5]{Faris__CombinatoricsAndClusterExpansions__ProbSur_2010} is a formal expansion of the logarithm of the partition function
\begin{subequations}\label{eq_expansions}
\begin{equation}\label{eq_expansion_logarithm}
 \log\PartitionFunction{\FiniteVolume}(\vec{z})
 \Formally{=}
  \sum_{n\ge 1} \frac{1}{\Factorial{n}}
  \sum_{\vec{\PolymerOther}\in\FiniteVolume^n}
  \Ursell{\vec{\PolymerOther}}
  \prod_{i=1}^n z_{\PolymerOther_i}
\end{equation}
and the pinned connected function
\begin{equation}\label{eq_expansion_pinned}
 \PinnedConnectedExpression{\FiniteVolume}{\Polymer}(\vec{z})
 \Formally{=}
  \sum_{n\ge 0}\frac{1}{\Factorial{n}}
  \sum_{\vec{\PolymerOther}\in\Set{\Polymer}\times\FiniteVolume^n}
  \Ursell{\vec{\PolymerOther}}
  \prod_{i=1}^n z_{\PolymerOther_i}\,.
\end{equation}
\end{subequations}

The \emph{Ursell functions} have the \emph{alternating sign property}~\cite[proposition 2.8]{Scott_Sokal__TheRepulsiveLatticeGasTheIndependentSetPolynomialAndTheLovaszLocalLemma__JSP_2005}
\begin{equation}\label{eq_alternatingSignProperty}
 \ForAll\vec{\Polymer}\in\PolymerSet^I:
 \quad
 (-1)^{\Cardinality{I}+1}\Ursell{\vec{\PolymerOther}}\ge 0\,.
\end{equation}

Hence, the worst case is for negative real fugacities~\cite{Fernandez_Procacci__ClusterExpansionForAbstractPolymerModels_NewBoundsFromAnOldApproach__CMP_2007}:
\begin{equation}\label{eq_worstCase}
 \left|\PinnedConnectedExpression{\FiniteVolume}{\Polymer}(\vec{z})\right|
 \le
 \PinnedConnectedExpression{\FiniteVolume}{\Polymer}(-\Modulus{\vec{z}})
 \qquad
 \text{ and }
 \qquad
 \OPPR{\FiniteVolume}{\Polymer}(-\Modulus{\vec{z}})
 \le
 \Modulus{\OPPR{\FiniteVolume}{\Polymer}(\vec{z})}\,.
\end{equation}

\section{Sufficient conditions for uniform boundedness}
\label{sec_scub}
This paper is about sufficient conditions for uniform boundedness (short SCUB) of the reduced correlation and the free energy~\eqref{eq_quantitiesOfInterest} for small absolute values of the fugacity. A SCUB implies uniform (in $\FiniteVolume$) bounds of either the form
\begin{subequations}\label{eq_uniformBounds}
\begin{align}
 \label{eq_uniformBounds_pinned}
 \Exists \vec{C}\in\,]0,\infty[^\PolymerSet:
 \ForAll\Polymer\in\FiniteVolume\FiniteSubsetOf\PolymerSet:
 &\quad \PinnedConnectedExpression{\FiniteVolume}{\Polymer}(-\Modulus{\vec{z}})
  \le C_\Polymer\,,
\intertext{or the form}
 \label{eq_uniformBounds_onePolymerRatio}
 \Exists \vec{c}\in\,]0,\infty[^\PolymerSet:
 \ForAll\Polymer\in\FiniteVolume\FiniteSubsetOf\PolymerSet:
 &\quad c_\Polymer\le
 \OPPR{\FiniteVolume}{\Polymer}(-\Modulus{\vec{z}})\,.
\end{align}
\end{subequations}
As we may express the pinned connected expression in the one polymer partition ratios and vice-versa, the bounds in~\eqref{eq_uniformBounds} are equivalent. Such a uniform bound also removes the formal constraint from the expansions~\eqref{eq_expansion_logarithm} and~\eqref{eq_expansion_pinned}. If a SCUB is satisfied for some $\vec{\rho}$, then all the definitions and relations in sections~\ref{sec_key} and~\ref{sec_worstCase} are well-defined for every complex $\vec{z}$ with $\Modulus{\vec{z}}\le\vec{\rho}$~\cite[theorem 2.10]{Scott_Sokal__TheRepulsiveLatticeGasTheIndependentSetPolynomialAndTheLovaszLocalLemma__JSP_2005}.\\

The uniform boundedness~\eqref{eq_uniformBounds} is equivalent to the well-definedness and positivity of the following quantities:
\begin{equation}\label{eq_onePolymerPartitionRatioLimit}
 \ForAll\Polymer\in\PolymerSet:
 \quad
 \OPPR{\PolymerSet}{\Polymer}(-\vec{\rho})
 :=\lim_{\FiniteVolume\Exhausts\PolymerSet}
  \OPPR{\FiniteVolume}{\Polymer}(-\vec{\rho})\,,
\end{equation}
where the limit is monotone decreasing in the ultra-filter of finite subsets of $\PolymerSet$ containing $\Polymer$ (see lemma~\ref{lem_oppr_mon}).

\subsection{Admissible parameters}
\label{sec_admissibleParameters}
We introduce the \emph{multidisc of admissible parameters}
\begin{equation}\label{eq_parameterDisc}
 \PDPS
 :=\BigSet{
  \vec{\rho}\in[0,\infty[^\PolymerSet:\quad
  \ForAll\FiniteVolume\FiniteSubsetOf\PolymerSet:\quad
  \PartitionFunction{\FiniteVolume}(-\vec{\rho})>0
 }\,.
\end{equation}
The set $\PDPS$ is a subset of $[0,1[^\PolymerSet$, an intersection of open sets and almost closed~\cite[theorem 8.1]{Scott_Sokal__TheRepulsiveLatticeGasTheIndependentSetPolynomialAndTheLovaszLocalLemma__JSP_2005}. The \emph{interior} of $\PDPS$, with respect to the \emph{box-topology}, is
\begin{equation}\label{eq_parameterDiscInterior}
 \InteriorOf\PDPS
 :=\Set{
  \vec{\rho}\in\PDPS:\quad
  \Exists\vec{\nu}>\vec{\rho}:\quad
  \vec{\nu}\in\PDPS
 }\,,
\end{equation}
Both $\PDPS$ and $\InteriorOf\PDPS$ are \emph{log-convex} and a \emph{down-set}~\cite[proposition 2.5]{Scott_Sokal__TheRepulsiveLatticeGasTheIndependentSetPolynomialAndTheLovaszLocalLemma__JSP_2005}, i.e., if $\vec{\rho}\in\PDPS$, then every $\vec{\nu}$ with $\vec{0}\le\vec{\nu}\le\vec{\rho}$ is also in $\PDPS$. Likewise, the domain on which a SCUB holds is always either log-convex and a down-set or may immediately be extended to its log-convex and downward closure~\cite[proposition 2.15]{Scott_Sokal__TheRepulsiveLatticeGasTheIndependentSetPolynomialAndTheLovaszLocalLemma__JSP_2005}.\\

Lying in $\InteriorOf\PDPS$ already yields a generic SCUB:

\begin{Prop}\label{prop_scub_generic}
If $\vec{0}\le\vec{\rho}\le\vec{\nu}\in\PDPS$, then
\begin{equation}\label{eq_scub_generic}
 \ForAll\Polymer\in\PolymerSet:\qquad
 \OPPR{\PolymerSet}{\Polymer}(-\vec{\rho})
 \ge\begin{cases}
  \frac{\nu_\Polymer-\rho_\Polymer}{\nu_\Polymer}
  &\text{if }\nu_\Polymer>0\\
  1
  &\text{if }\nu_\Polymer=0\,.
 \end{cases}
\end{equation}
In particular, we have
\begin{equation}\label{eq_unifbound_generic}
 \vec{\rho}\in\InteriorOf\PDPS
 \Then
 \Bigl(
  \ForAll\Polymer\in\PolymerSet:
  \quad
  \OPPR{\PolymerSet}{\Polymer}(-\vec{\rho})>0
 \Bigr)
 \Then
 \text{\eqref{eq_onePolymerPartitionRatioLimit} holds.}
\end{equation}
\end{Prop}

The proof of proposition~\ref{prop_scub_generic} is in section~\ref{sec_inductivePolymerLevelProofs}.
\AlsoArxiv{Example~\ref{exam_unboundednessOnTheHomogeneousTree} shows that at least in one case uniform boundedness breaks down at a boundary point of $\PDPS$. This is shown by the infinite-volume limit of the one-polymer ratios at that boundary point. As we believe that this break-down holds all along the boundary, a SCUB should only ever map out subsets of the $\InteriorOf\PDPS$. Thus, we reformulate the search for SCUBs as:

\begin{Que}\label{que_pdps_shape}
What are the properties of $\PDPS$? What are sufficient conditions for $\vec{\rho}$ to be in $\PDPS$?
\end{Que}

There has been work~\cite[section 8.2]{Scott_Sokal__TheRepulsiveLatticeGasTheIndependentSetPolynomialAndTheLovaszLocalLemma__JSP_2005} on necessary conditions, which correspond to supersets of $\PDPS$.
}

\subsection{Discrepancy between the known SCUBs}
\label{sec_discrepancy}
We recall the SCUBs by Koteck\'{y} \& Preiss~\eqref{eq_scub_inhom_kp}~\cite{Kotecky_Preiss__ClusterExpansionForAbstractPolymerModels__CMP_1986}, Dobrushin~\eqref{eq_scub_inhom_dob}~\cite{Dobrushin__EstimatesOfSemiInvariantsForTheIsingModelAtLowTemperatures__AMST_1996} and Fern\'{a}ndez \& Procacci~\eqref{eq_scub_inhom_fp}~\cite{Fernandez_Procacci__ClusterExpansionForAbstractPolymerModels_NewBoundsFromAnOldApproach__CMP_2007}. All three SCUBs have the same shape: There is a function $\LeopGeneric{}: [0,\infty[^\PolymerSet\to[0,\infty[^\PolymerSet$ such that
\begin{equation}\label{eq_scubform_classic}
 \Bigl(
  \Exists\vec{\mu}:\quad
  \vec{\rho}\,\LeopGeneric{}(\vec{\mu})\le\vec{\mu}
 \Bigr)
 \Then\text{\eqref{eq_onePolymerPartitionRatioLimit} holds.}
\end{equation}
The particular shape of $\LeopGeneric{}$ is, for each SCUB in order and every $\Polymer$:
\begin{subequations}\label{eq_scub_inhom_known}
\begin{align}
 \label{eq_scub_inhom_kp}
 \LeopKP{\Polymer}(\vec{\mu})
 &:=\exp\left(
  \sum_{\PolymerOther\in\IncompatiblePolymers{\Polymer}}
   \mu_\PolymerOther
 \right)
 \\\label{eq_scub_inhom_dob}
 \LeopDob{\Polymer}(\vec{\mu})
 &:=\prod_{\PolymerOther\in\IncompatiblePolymers{\Polymer}}
   (1+\mu_\PolymerOther)
 \\\label{eq_scub_inhom_fp}
 \LeopFP{\Polymer}(\vec{\mu})
 &:=\PartitionFunction{\IncompatiblePolymers{\Polymer}}(\vec{\mu})
  = \mu + \PartitionFunction{\IncompatibleOtherPolymers{\Polymer}}(\vec{\mu})\,.
\end{align}
\end{subequations}
These SCUBs form successive improvements and arise from a common tree-operator approach, outlined in sections~\ref{sec_explanation} and~\ref{sec_treeOperatorFramework}. \AlsoArxiv{The SCUB~\eqref{eq_scub_inhom_dob} is also known as the \emph{asymmetric Lov\'{a}sz Local Lemma}~\cite{Erdos_Lovasz__ProblemsAndResultsOn3ChromaticHypergraphsAndSomeRelatedQuestions__CMSJB_1975}.}\\

The \emph{reduced SCUB} reduces the local degree by one, compared to Dobrushin's SCUB. Here $\LeopGeneric{}$ takes the shape:
\begin{equation}\label{eq_scub_inhom_red}
 \LeopRed{\Polymer}(\vec{\mu})
  :=(1+\mu_\Polymer)
  \sup_{\EscapePolymer\in\IncompatibleOtherPolymers{\Polymer}}
  \prod_{\PolymerOther\in\IncompatibleOtherPolymers{\Polymer}
                          \setminus\Set{\EscapePolymer}}
   (1+\mu_\PolymerOther)\,.
\end{equation}
Its proof is in section~\ref{sec_inductivePolymerLevelProofs}. The proof is inductive à la Dobrushin~\cite{Dobrushin__EstimatesOfSemiInvariantsForTheIsingModelAtLowTemperatures__AMST_1996} and an extension of the homogeneous version stated in~\cite[Corollary~5.7]{Scott_Sokal__TheRepulsiveLatticeGasTheIndependentSetPolynomialAndTheLovaszLocalLemma__JSP_2005}. The induction over the cardinality of finite subsets $\FiniteVolume$ of $\PolymerSet$ and an ordering of the polymers ensures a bound of $\OPPR{\Lambda}{\Polymer}$ without ever considering $\IncompatibleOtherPolymers{\Polymer}\setminus\FiniteVolume\not=\emptyset$. The reduced SCUB only holds with a strict inequality in~\eqref{eq_scubform_classic}, due the use of proposition~\ref{prop_scub_generic}.\\

Clearly, $\LeopRed{}$ beats $\LeopDob{}$, because it is uniformly smaller. If $(\PolymerSet,\IncompatibleTo)$ is \emph{triangle-free} and the fugacity is homogeneous, then $(\IncompatiblePolymers{\Polymer},\IncompatibleTo)$ is a star and $\LeopFP{}$ is worse than $\LeopRed{}$. The reduced SCUB is exact on regular infinite trees with homogeneous fugacity~\cite{Shearer__OnAProblemOfSpencer__Comb_1985}. On the other hand, if $(\PolymerSet,\IncompatibleTo)$ contains many triangles, then $\LeopFP{}$ beats $\LeopRed{}$. See section~\ref{sec_example} for numerical examples. This raises the questions:

\begin{Que}\label{que_treeoperator}
Is there a tree-operator hidden behind $\LeopRed{}$, too?
\end{Que}

\begin{Que}\label{que_improvement}
Can we improve $\LeopRed{}$ similarly to how $\LeopFP{}$ improves $\LeopDob{}$?
\end{Que}

\begin{Que}\label{que_unification}
Is there an approach unifying $\LeopFP{}$, $\LeopRed{}$ and a possible improvement from question~\ref{que_improvement}?
\end{Que}
\subsection{The new SCUBs}
\label{sec_new_scubs}

This section answers the above questions affirmatively and states the new SCUBs. Figure~\ref{fig_relation} summarises the relation between the different SCUBs. We give an informal back-of-the-envelope derivation of the new SCUBs and their names in section~\ref{sec_explanation}.\\

The new SCUBs have the shape: There is a function $\LeopGeneric{}: [0,\infty[^\PolymerSet\to[0,\infty[^\PolymerSet$ such that
\begin{equation}\label{eq_scubform_new}
 \Bigl(
  \Exists\vec{\mu}:\quad
  \vec{\rho}\,\LeopGeneric{}(\vec{\mu})<\vec{\mu}
 \Bigr)
 \Then\text{\eqref{eq_onePolymerPartitionRatioLimit} holds.}
\end{equation}
The strict inequality follows from the use of proposition~\ref{prop_scub_generic}. Our first SCUB is the \emph{reduced SCUB}~\eqref{eq_scub_inhom_red}. The answer to question~\ref{que_improvement} is the \emph{returning SCUB}, with shape
\begin{equation}\label{eq_scub_inhom_ret}
 \LeopRet{\Polymer}(\vec{\mu})
  :=(1+\mu_\Polymer)
  \sup_{\EscapePolymer\in\IncompatibleOtherPolymers{\Polymer}}
  \PartitionFunction{\IncompatibleOtherPolymers{\Polymer}
                       \setminus\Set{\EscapePolymer}}\,.
\end{equation}
Its proof is in section~\ref{sec_proof_ret}. As we derive $\LeopRet{}$ as a convergence condition of a tree-operator, we may relax it immediately to $\LeopRed{}$ in section~\ref{sec_proof_red}. This answers question~\ref{que_treeoperator}.\\

The function $\LeopRet{}$ is also not comparable with $\LeopFP{}$. The mixing SCUB below answers question~\ref{que_unification}: it generalises both $\LeopRed{}$ and $\LeopFP{}$. We interpolate at each polymer between the greedy behaviour of $\LeopFP{}$, marked by $\BehaviourGreedy$, and the returning behaviour of $\LeopRet{}$, marked by $\BehaviourReturning$. The naming is explained in section~\ref{sec_explanation}. For a fixed choice of behaviour $\MixingVector\in\BehaviourSetIndexedBy{\PolymerSet}$, we get a SCUB of shape
\begin{equation}\label{eq_scub_mixing_fixed}
 \LeopMixing{\MixingVector}{\Polymer}(\vec{\mu})
 :=
 \Iverson{\MixingVector_\Polymer=\BehaviourGreedy}
 \LeopFP{\Polymer}(\vec{\mu})
 +
 \Iverson{\MixingVector_\Polymer=\BehaviourReturning}
 \LeopRet{\Polymer}(\vec{\mu})\,.
\end{equation}
The \emph{mixing SCUB} is the supremum over all choices of $\MixingVector$, with shape:
\begin{equation}\label{eq_scub_mixing}
 \Bigl(
  \Exists\MixingVector\in\BehaviourSetIndexedBy{\PolymerSet},
  \Exists\vec{\mu}\in[0,\infty[^\PolymerSet:
  \quad
  \vec{\rho}\,\LeopMixing{\MixingVector}{}(\vec{\mu}) <\vec{\mu}
 \Bigr)
 \Then\text{\eqref{eq_onePolymerPartitionRatioLimit} holds.}
\end{equation}
Its proof is in section~\ref{sec_proof_mixing}. The SCUBs $\LeopFP{}$ and $\LeopRet{}$ are special cases:
\begin{equation}\label{eq_leop_mixing_reduction}
 \LeopFP{}=\LeopMixing{\vec{\BehaviourGreedy}}{}
 \quad\text{ and }\quad
 \LeopRet{}=\LeopMixing{\vec{\BehaviourReturning}}{}\,.
\end{equation}
In the homogeneous (both in $(\PolymerSet,\IncompatibleTo)$ and $\vec{\rho}$) setting, $\LeopMixing{\MixingVector}{\Polymer}$ reduces to the extreme cases in~\eqref{eq_leop_mixing_reduction}.

\begin{figure}
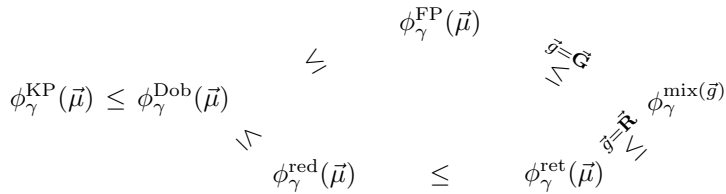

\begin{displaymath}
 \MxyLeopComparison{1em}{1em}
\end{displaymath}
\caption[Relationship between the SCUBs]{Relationship between the SCUBs}
\label{fig_relation}
\end{figure}

\subsection{Numerical example}
\label{sec_example}
This section contains example calculations of the SCUBs on two transitive lattices (see figure~\ref{fig_lattices}) with homogeneous fugacity. Hence, we may exclude the mixing SCUB. The results are contained in table~\ref{tab_example}. The hexagonal lattice is triangle-free, whence the returning SCUB becomes the reduced SCUB and gives no additional improvement over the FP SCUB. In the case of the line graph of the hexagonal lattice, the reduced SCUB is worse than the FP SCUB, but the returning SCUB improves on the FP SCUB. For high degree graphs, any advantages of the reduced and returning SCUBs either disappear or become negligible.

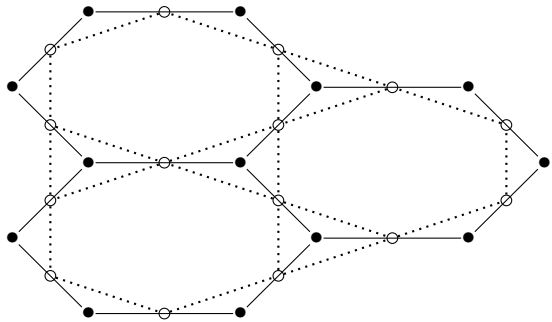
\begin{figure}
\begin{center}
\begin{tikzpicture}[scale=2]
  \tikzstyle{hexVertex} = [inner sep = 0pt]
  \tikzstyle{hexEdge} = []
  \tikzstyle{lineVertex} = [circle, draw, minimum size = 4pt, inner sep = 0pt]
  \tikzstyle{lineEdge} = [thick, dotted]

  \newcommand{\hexnode}[2]{\node[hexVertex] (#1) at (#2) {$\bullet$};}
  \newcommand{\hexedge}[2]{\draw[hexEdge] (#1) to (#2);}

  \hexnode{h11}{0.0,0.5}
  \hexnode{h12}{0.0,1.5}
  \hexnode{h21}{0.5,0.0}
  \hexnode{h22}{0.5,1}.0
  \hexnode{h23}{0.5,2.0}
  \hexnode{h31}{1.5,0.0}
  \hexnode{h32}{1.5,1.0}
  \hexnode{h33}{1.5,2.0}
  \hexnode{h41}{2.0,0.5}
  \hexnode{h42}{2.0,1.5}
  \hexnode{h51}{3.0,0.5}
  \hexnode{h52}{3.0,1.5}
  \hexnode{h61}{3.5,1.0}

  \hexedge{h11}{h21}
  \hexedge{h11}{h22}
  \hexedge{h12}{h22}
  \hexedge{h12}{h23}
  \hexedge{h21}{h31}
  \hexedge{h22}{h32}
  \hexedge{h23}{h33}
  \hexedge{h31}{h41}
  \hexedge{h32}{h41}
  \hexedge{h32}{h42}
  \hexedge{h33}{h42}
  \hexedge{h41}{h51}
  \hexedge{h42}{h52}
  \hexedge{h51}{h61}
  \hexedge{h52}{h61}

  \newcommand{\linenode}[2]{\node[lineVertex] (#1) at (#2) {};}
  \newcommand{\lineedge}[2]{\draw[lineEdge] (#1) to (#2);}

  \linenode{e11}{0.25,0.25}
  \linenode{e12}{0.25,0.75}
  \linenode{e13}{0.25,1.25}
  \linenode{e14}{0.25,1.75}
  \linenode{e21}{1.00,0.00}
  \linenode{e22}{1.00,1.00}
  \linenode{e23}{1.00,2.00}
  \linenode{e31}{1.75,0.25}
  \linenode{e32}{1.75,0.75}
  \linenode{e33}{1.75,1.25}
  \linenode{e34}{1.75,1.75}
  \linenode{e41}{2.50,0.50}
  \linenode{e42}{2.50,1.50}
  \linenode{e51}{3.25,0.75}
  \linenode{e52}{3.25,1.25}

  \lineedge{e12}{e13}
  \lineedge{e31}{e41}
  \lineedge{e34}{e42}

  \lineedge{e11}{e21}
  \lineedge{e21}{e31}
  \lineedge{e31}{e32}
  \lineedge{e32}{e22}
  \lineedge{e22}{e12}
  \lineedge{e12}{e11}

  \lineedge{e13}{e22}
  \lineedge{e22}{e33}
  \lineedge{e33}{e34}
  \lineedge{e34}{e23}
  \lineedge{e23}{e14}
  \lineedge{e14}{e13}

  \lineedge{e32}{e41}
  \lineedge{e41}{e51}
  \lineedge{e51}{e52}
  \lineedge{e52}{e42}
  \lineedge{e42}{e33}
  \lineedge{e33}{e32}

\end{tikzpicture}
\end{center}
\caption[Lattices used in examples]{The lattices used in the examples in section~\ref{sec_example}: the hexagonal lattice (solid) and its line graph (dotted).}
\label{fig_lattices}
\end{figure}

\begin{table}
\begin{center}
\begin{tabular}{cccc}
 Lattice
 & SCUB
 & Shape of $\LeopGeneric{\Polymer}$
 & Optimal $\rho$
 \\ \hline hex
 & Dob
 & $(1+\mu)^4$
 & 0.1055
 \\ hex
 & FP
 & $\mu + (1+\mu)^3$
 & 0.1290
 \\ hex
 & reduced
 & $(1+\mu)(1+\mu)^2$
 & 0.1481
 \\ hex
 & returning
 & $(1+\mu)(1+\mu)^2$
 & 0.1481
 \\ line
 & Dob
 & $(1+\mu)^5$
 & 0.0819
 \\ line
 & FP
 & $\mu + (1+2\mu)^2$
 & 0.1111
 \\ line
 & reduced
 & $(1+\mu)(1+\mu)^3$
 & 0.1055
 \\ line
 & returning
 & $(1+\mu)(1+\mu)(1+2\mu)$
 & 0.1134
\end{tabular}
\end{center}
\caption[Example calculation of the SCUBs]{Example calculations of the SCUBs on the hexagonal lattice and its line graph in the homogeneous case. See figure~\ref{fig_lattices} and section~\ref{sec_example}. The optimal $\rho$ value is the solution of $\max\Set{\frac{\mu}{\LeopGeneric{}(\mu)}:\mu\in[0,\infty[}$. Higher values of $\rho$ are better.}
\label{tab_example}
\end{table}
\subsection{Informal explanation and discussion}
\label{sec_explanation}
This section gives an informal and graphical explanation of how the different SCUBs arise from tree-operators and their interplay with partition schemes. We also discuss possible extensions of our results and methods.\\

We start with a high level outline of the tree-operator approach. See sections~\ref{sec_treeOperatorFramework} and~\ref{sec_partitionSchemes} for the formal setup and full details. After a cluster expansion, a choice of a partition scheme and an application of Penrose's identity, the problem of uniform boundedness reduces to bounding a series summing first over clusters and then over certain weighted spanning trees of a given cluster. One may resum in a different order: over the depth of the tree, irrespective of which cluster it has been in. To control the convergence of this depth-wise summation, one needs to control the weight increase, if one counts all trees of depth up to $(n+1)$ instead of $n$. A single application of a tree-operator corresponds to an increase of the permissible maximum depth by one. A SCUB for this tree-operator controls the resulting weight increase and guarantees the convergence of the iterated application of the tree-operator.\\

In more detail, the rooted trees have labels attached to their vertices. Each label induces a weight. The weight of a labelled tree is the product of the induced weights. A tree-operator creates a family of labelled depth $n+1$ trees from a labelled depth $n$ tree. The tree-operator encodes permissible attachments of children to a leaf. These attachments of a finite, possibly zero, number of children and the choice of their labels depend only on the label of the leaf. The family of depth $n+1$ trees results from all permissible attachments, independently for each depth $n$ leaf of the original tree.\\

A partition scheme selects a subset of the spanning trees of a graph, the singleton trees. The rooted, labelled and finite trees of interest in the resummed series are the singleton trees on each cluster. We relax the global constraints on these trees to local constraints between the label of a vertex and its children’s' labels. This allows to encode them by a tree-operator. The tree-operators underlying the SCUBs differ in how much the structure of a cluster constrains the singleton trees chosen by the partition scheme and how the constraints are relaxed to encode them by tree-operators. We describe these choices for all SCUBs in the following paragraphs, with sketches based on the neighbourhood in figure~\ref{fig_neighbourhood} in figures~\ref{fig_classic} and~\ref{fig_new}.\\

Koteck\'{y} \& Preiss' SCUB (see~\eqref{eq_scub_inhom_kp} and figure~\ref{fig_classic_kp}) is the simplest. Apply an arbitrary partition scheme to a cluster. The only constraint this places on the singleton trees is that each vertex is labelled by a polymer and that neighbouring vertices in a singleton tree have incompatible polymer labels. See figure~\ref{fig_anatomy}. Hence, a vertex with label $\Polymer$ may have an arbitrary finite number of children with labels in $\IncompatiblePolymers{\Polymer}$. The function $\LeopKP{\Polymer}$ is the generating function of these configurations.\\

Next are Dobrushin's SCUB (see~\eqref{eq_scub_inhom_dob} and figure~\ref{fig_classic_dob}) and Fern\'{a}ndez \& Procacci's SCUB (see~\eqref{eq_scub_inhom_fp} and figure~\ref{fig_classic_fp}). Here the partition scheme is Penrose's scheme and the labels are again the polymers. When applied to a cluster, it constrains the polymer labels of the children of a vertex in a singleton tree to form a compatible set within $\IncompatiblePolymers{\Polymer}$. Fern\'{a}ndez \& Procacci's SCUB takes this fully into account and $\PartitionFunction{\IncompatiblePolymers{\Polymer}}$ is the generating function of these compatible sets. Dobrushin's SCUB relaxes this to: each polymer may appear at most once as a child's label.\\

The reduced SCUB (see~\eqref{eq_scub_inhom_red} and figure~\ref{fig_new_red}) and the returning SCUB (see~\eqref{eq_scub_inhom_ret} and figure~\ref{fig_new_ret}) feature different choices. First, we mimic the exclusion of at least one polymer among the labels in $\IncompatibleOtherPolymers{\Polymer}$ in the Dobrushin-style proof of the reduced SCUB. The returning scheme in section~\ref{sec_schemeReturning} is a partition scheme tailored to clusters, as opposed to a partition scheme for general graphs applied to a cluster. This captures more of the structure of a cluster (see figure~\ref{fig_anatomy}) as constraints on the singleton trees. The returning scheme constrains polymer labels along a path from the root to form a lazy self-avoiding walk in the polymer system. The ``returning'' alludes to this fact. This interdicts one polymer in $\IncompatibleOtherPolymers{\Polymer}$ as a child's label in its singleton trees.\\

Second, we use incompatible pairs of polymers to label the vertices of the singleton trees. This carries the information about the last different polymer label along the path to the root and interdicts it as a child's label. See sections~\ref{sec_schemeReturning} and~\ref{sec_proof_ret} for details. This extra constraint comes at a price: while the children's labels still have to be compatible with each other, this only holds for polymers in $\IncompatibleOtherPolymers{\Polymer}$.\\

We derive tree-operators from these constraints and project the resulting SCUBs to normal single polymer labels. This is the origin of the supremum in the reduced and returning SCUB. The price of ignoring incompatibility with $\Polymer$ results in the $(1+\mu_\Polymer)$ factor. The reduced SCUB relaxes the compatibility constraints for the children's labels in $\IncompatibleOtherPolymers{\Polymer}$ analogously to Dobrushin's SCUB.\\

The mixing SCUB~\eqref{eq_scub_mixing} derives in the same way from the more elaborate synthetic scheme in section~\ref{sec_schemeSynthetic}. It interpolates between the greedy behaviour of Penrose's scheme, allowing all of $\IncompatiblePolymers{\Polymer}$ as labels for children, and the constraints imposed by the returning scheme. The discussion in section~\ref{sec_proof_mixing} suggests that the mixing SCUB is optimal within the class of SCUBs on general polymer systems and based on tree-operators indexed by incompatible polymer pairs.\\

There are several ways to improve the present SCUBs. First, we may enlarge the label spaces to make more of the known constraints on the singleton trees available to a tree-operator. Second, we may use tree-operators adding trees up to depth $k$ instead of only $1$, already outlined in~\cite[chapter 5]{Temmel__PropertiesAndApplicationsOfBernoulliRandomFieldsWithStrongDependencyGraphs__TUG_2012}. This takes constraints between the labels of some non-adjacent vertices into account. Third, we may develop new partition schemes which use more than just the incompatibility information from the polymer system. Like geometric information, if $(\PolymerSet,\IncompatibleTo)$ is a regular planar lattice.\\

Both the first and second improvement lead to SCUBs where $\LeopGeneric{\Polymer}$ depends on a larger incompatible neighbourhood of $\Polymer$ in $(\PolymerSet,\IncompatibleTo)$ than just $\IncompatiblePolymers{\Polymer}$. The even more complex shape of such SCUBs seems to be detrimental to analysis. For a concrete locally finite polymer system, one may enumerate the trees encoded by the tree-operator by computer and evaluate the resulting SCUB numerically. The utility of the third improvement depends on the integration of the additional knowledge about the polymer system. We plan to present results incorporating the above improvements in future work.\\

The derivation of improved SCUBs is straightforward for the first two improvements: the tools and theory are all present in section~\ref{sec_treeOperatorFramework} and~\cite[chapter 5]{Temmel__PropertiesAndApplicationsOfBernoulliRandomFieldsWithStrongDependencyGraphs__TUG_2012}. Our reformulation of Penrose's partition scheme and our new partition schemes lead us to the notion of explorative partition schemes in section~\ref{sec_explorativePartitionSchemes}. This is a systematic way to derive partition schemes of a cluster, such that the constraints placed on its singleton trees may be encoded with minimal loss by a tree-operator.

\begin{figure}
\begin{center}
\begin{tikzpicture}[scale=2]
  \drawnodeempty{a}{0,0}
  \drawnodebullet{a1}{75:.2}
  \drawnodebullet{a2}{195:.2}
  \drawnodebullet{a3}{315:.2}
  \drawnodebullet{b}{1,0}
  \drawnodeempty{c}{2,0}
  \drawnodebullet{c1}{2,0.15}
  \drawnodebullet{c2}{2,-0.15}
  \drawnodebullet{d}{0,-1}
  \drawnodebullet{e}{1,-1}

  \draw (a1) -- (b) (a2) -- (b) (a3) -- (b)
        (b) -- (c1) (b) -- (c2)
        (a1) -- (d) (a2) -- (d) (a3) -- (d)
        (b) -- (e)
        (c1) -- (e) (c2) -- (e);
  \draw (c1) -- (c2)
        (a1) -- (a2) -- (a3) -- (a1);
  \foreach \x in {a, b, c} {\draw[dashed] (\x) circle (.3) +(.3, .3) node {$\x$};}
  \foreach \x in {d, e}    {\draw[dashed] (\x) circle (.3) +(.3,-.3) node {$\x$};}
\end{tikzpicture}
\begin{tikzpicture}[scale=2]
  \drawnode{a}{0,1}
  \drawnode{b}{1,1}
  \drawnode{c}{2,1}
  \drawnode{d}{0,0}
  \drawnode{e}{1,0}

  \draw (a) -- (b)
        (b) -- (c)
        (a) -- (d)
        (b) -- (e)
        (c) -- (e);
  \foreach \x in {a, b, c} {\selfedge{\x}{90};}
  \foreach \x in {d, e} {\selfedge{\x}{270};}
\end{tikzpicture}
\end{center}
\caption[Anatomy of a cluster]{Anatomy of the cluster of $\Cluster{(a,a,b,c,a,d,c,e)}$ on the left, for the incompatibility graph of its support $(\Set{a,b,c,d,e},\IncompatibleTo)$ on the right. The dotted lines in the cluster mark the structural constraints it inherits from its support.}
\label{fig_anatomy}
\end{figure}
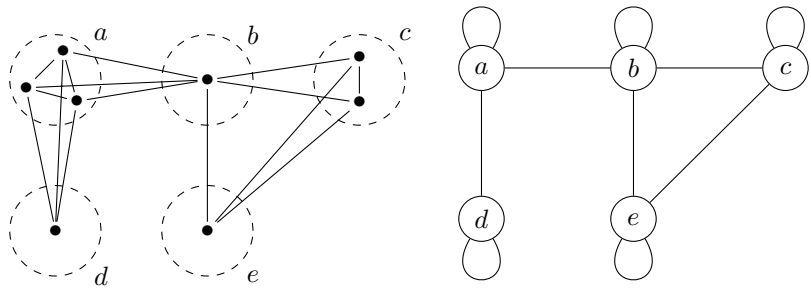

\begin{figure}
\begin{center}
\begin{tikzpicture}[scale=2]
  \node[vertex] (x1) at (0,1) {$\xi_1$};
  \node[vertex] (x2) at (1,0) {$\xi_2$};
  \node[vertex] (x3) at (0,-1) {$\xi_3$};
  \node[vertex] (eps) at (-1,0) {$\varepsilon$};
  \node[vertex] (gam) at (0,0) {$\gamma$};

  \draw (x1) -- (x2) -- (gam) -- (eps) -- (x3) -- (gam) -- (x1);

  \foreach \x in {x1, x2, x3, eps, gam} {\selfedge{\x}{45};}
\end{tikzpicture}
\end{center}
\caption[Polymer incompatibles]{An example of the incompatible neighbourhood graph $(\IncompatiblePolymers{\Polymer},\IncompatibleTo)$ of a polymer $\Polymer$, with lines representing incompatibilities.}
\label{fig_neighbourhood}
\end{figure}

\begin{figure}
\subfloat[][Constraints in $\displaystyle\LeopKP{\Polymer}(\vec{\mu})$.]{
\label{fig_classic_kp}
\begin{tikzpicture}[scale=1.8]
  \node[vertex] (gam) at (0,0) {$\gamma$};
  \node[vertex] (x1) at (1,1) {$\xi_1$};
  \node[vertex] (x2) at (1,1/3) {$\xi_2$};
  \node[vertex] (x3) at (1,-1/3) {$\xi_3$};
  \node[vertex] (eps) at (1,-1) {$\varepsilon$};

  \draw (x1) -- (x2) -- (gam) -- (eps) -- (x3) -- (gam) -- (x1);

  \selfedge{gam}{180};
  \foreach \x in {x1, x2, x3, eps} {\selfedge{\x}{0};}
\end{tikzpicture}
}
\subfloat[][Constraints in $\displaystyle\LeopDob{\Polymer}(\vec{\mu})$.]{
\label{fig_classic_dob}
\begin{tikzpicture}[scale=1.8]
  \node[vertex] (gam) at (0,0) {$\gamma$};
  \node[vertex] (x1) at (1,1) {$\xi_1$};
  \node[vertex] (x2) at (1,1/3) {$\xi_2$};
  \node[vertex] (x3) at (1,-1/3) {$\xi_3$};
  \node[vertex] (eps) at (1,-1) {$\varepsilon$};

  \draw (x1) -- (x2) -- (gam) -- (eps) -- (x3) -- (gam) -- (x1);

  \selfedgemark{gam}{180};
  \foreach \x in {x1, x2, x3, eps} {\selfedgemark{\x}{0};}
\end{tikzpicture}
}
\subfloat[][Constraints in $\displaystyle\LeopFP{\Polymer}(\vec{\mu})$.]{
\label{fig_classic_fp}
\begin{tikzpicture}[scale=1.8]
  \node[vertex] (gam) at (0,0) {$\gamma$};
  \node[vertex] (x1) at (1,1) {$\xi_1$};
  \node[vertex] (x2) at (1,1/3) {$\xi_2$};
  \node[vertex] (x3) at (1,-1/3) {$\xi_3$};
  \node[vertex] (eps) at (1,-1) {$\varepsilon$};

  \draw[marked] (x1) -- (x2) -- (gam) -- (eps) -- (x3) -- (gam) -- (x1);

  \selfedgemark{gam}{180};
  \foreach \x in {x1, x2, x3, eps} {\selfedgemark{\x}{0};}
\end{tikzpicture}
}
\caption[Constraints in the classic SCUBs]{(Colour online) Constraints in the classic SCUBs at a leaf labelled with $\Polymer$. We assume the setting of figure~\ref{fig_neighbourhood}. The incompatibilities taken into account by the SCUB are marked by dotted red lines. The SCUB sums over all weighted depth $1$ trees rooted at $\Polymer$, which do not contain two vertices labelled with polymers with a dotted red line in the corresponding diagram. More dotted red lines imply less trees, a more relaxed SCUB and bigger admissible $\vec{\rho}$.}
\label{fig_classic}
\end{figure}
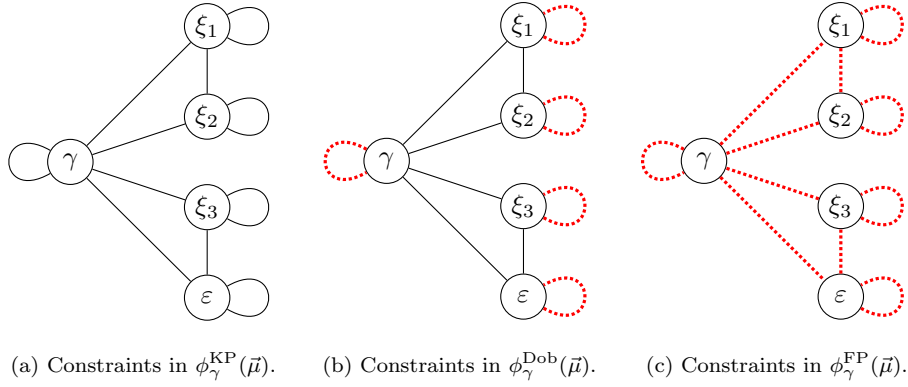

\begin{figure}
\subfloat[][Constraints in $\displaystyle\LeopRed{\Polymer}(\vec{\mu})$.]{
\label{fig_new_red}
\begin{tikzpicture}[scale=2]
  \node[vertexlong] (gam) at (0,0) {$\!\!(\gamma, \varepsilon)\!\!$};
  \node[vertexlong] (x1) at (1,1) {$\!\!(\xi_1, \gamma)\!\!$};
  \node[vertexlong] (x2) at (1,0) {$\!\!(\xi_2, \gamma)\!\!$};
  \node[vertexlong] (x3) at (1,-1) {$\!\!(\xi_3, \gamma)\!\!$};

  \draw (gam) -- (x1) -- (x2) -- (gam) -- (x3);

  \selfedge{gam}{180};
  \foreach \x in {x1, x2, x3} {\selfedgemark{\x}{0};}
\end{tikzpicture}
}
\subfloat[][Constraints in $\displaystyle\LeopRet{\Polymer}(\vec{\mu})$.]{
\label{fig_new_ret}
\begin{tikzpicture}[scale=2]
  \node[vertexlong] (gam) at (0,0) {$\!\!(\gamma, \varepsilon)\!\!$};
  \node[vertexlong] (x1) at (1,1) {$\!\!(\xi_1, \gamma)\!\!$};
  \node[vertexlong] (x2) at (1,0) {$\!\!(\xi_2, \gamma)\!\!$};
  \node[vertexlong] (x3) at (1,-1) {$\!\!(\xi_3, \gamma)\!\!$};

  \draw (gam) -- (x1);
  \draw (gam) -- (x2);
  \draw[marked] (x1) -- (x2);
  \draw (gam) -- (x3);

  \selfedge{gam}{180};
  \foreach \x in {x1, x2, x3} {\selfedgemark{\x}{0};}
\end{tikzpicture}
}
\caption[Constraints in the new SCUBs]{(Colour online) Constraints in the new SCUBs at a leaf labelled with $\Polymer$, extended to $\EscapePair$ to account for the incompatible different polymer $\EscapePolymer$ forbidden as a child label. We assume the setting of figure~\ref{fig_neighbourhood}. If a child node in a singleton tree has a polymer label in $\IncompatibleOtherPolymers{\Polymer}\setminus\Set{\EscapePolymer}$, then the forbidden incompatible different polymer for the grandchildren's labels switches from $\EscapePolymer$ to $\Polymer$. The incompatibilities taken into account by the SCUB are marked by dotted red lines. The SCUB sums over all weighted depth $1$ trees rooted at $\Polymer$, which do not contain two vertices labelled with polymers with a dotted red line in the corresponding diagram. More dotted red lines imply less trees, a more relaxed SCUB and bigger admissible $\vec{\rho}$.}
\label{fig_new}
\end{figure}
\section{Inductive polymer level proofs}
\label{sec_inductivePolymerLevelProofs}
This section contains only polymer level inductive proofs based on one polymer partition ratios instead of cluster expansion techniques. This includes the proof of proposition~\ref{prop_scub_generic} and the reduced SCUB~\eqref{eq_scub_inhom_red}.  A key part is the \emph{fundamental identity}, a \emph{deletion-contraction property} of the partition function:
\begin{subequations}\label{eq_fi}
\begin{equation}\label{eq_fi_partition}
 \ForAll\Polymer\in\FiniteVolume\FiniteSubsetOf\PolymerSet:
 \quad
 \PartitionFunction{\FiniteVolume}(\vec{z})
 =\PartitionFunction{\FiniteVolume\setminus\Set{\Polymer}}(\vec{z})
 +z_\Polymer\,
 \PartitionFunction{\FiniteVolume\setminus\IncompatiblePolymers{\Polymer}}(\vec{z})\,.
\end{equation}
In terms of the one polymer partition ratios, and if $\PartitionFunction{\FiniteVolume\setminus\Set{\Polymer}}(\vec{z})\not=0$, it becomes
\begin{equation}\label{eq_fi_oppr}
 \OPPR{\FiniteVolume}{\Polymer}(\vec{z})
 = 1 + \frac{z_\Polymer}{ \prod_{i=1}^m
  \OPPR{\FiniteVolume\setminus\Set{\Polymer,\PolymerOther_{1},\dotsc,\PolymerOther_{i-1}}}
  {\PolymerOther_i}(\vec{z})
 }\,,
\end{equation}
where $\Set{\PolymerOther_1,\dotsc,\PolymerOther_m} := \IncompatibleOtherPolymers{\Polymer}\cap\FiniteVolume$.
\end{subequations}
We recall the monotonicity properties of the one polymer partition ratios:

\begin{Lem}[{\cite[top of page 62]{Scott_Sokal__TheRepulsiveLatticeGasTheIndependentSetPolynomialAndTheLovaszLocalLemma__JSP_2005},\cite[(5.67)]{Temmel__PropertiesAndApplicationsOfBernoulliRandomFieldsWithStrongDependencyGraphs__TUG_2012}}]
\label{lem_oppr_mon}
\begin{subequations}\label{eq_oppr_mon}
For $\Polymer\in\FiniteVolume\FiniteSubsetOf\PolymerSet$,\\
if $\PartitionFunction{\FiniteVolume\setminus\Set{\Polymer}}(-\vec{\rho})>0$, then
\begin{align}
 \label{eq_oppr_mon_volume}
 \ForAll\FiniteVolume'\subseteq\FiniteVolume:
 \quad
 &\OPPR{\FiniteVolume}{\Polymer}(-\vec{\rho})
 \le
 \OPPR{\FiniteVolume'}{\Polymer}(-\vec{\rho})\,,
 \\\label{eq_oppr_mon_parameter}
 \ForAll\vec{0}\le\vec{\nu}\le\vec{\rho}:
 \quad
 &\OPPR{\FiniteVolume}{\Polymer}(-\vec{\rho})
 \le
 \OPPR{\FiniteVolume}{\Polymer}(-\vec{\nu})\,.
\end{align}
\end{subequations}
\end{Lem}

That is, the one polymer partition ratios are decreasing in both volume and negative fugacity. Equations~\eqref{eq_oppr_mon_parameter} and~\eqref{eq_telescoping} together imply that $\PDPS$ is a down-set.

\begin{proof}[Proof of proposition~\ref{prop_scub_generic}]
The proof is inspired by a coupling~\cite[model 22]{Temmel__PropertiesAndApplicationsOfBernoulliRandomFieldsWithStrongDependencyGraphs__TUG_2012}. We prove the implication~\eqref{eq_scub_generic}, for all $\Polymer\in\FiniteVolume\FiniteSubsetOf\PolymerSet$, by induction over the cardinality of $\FiniteVolume$ and then take the limit $\FiniteVolume\Exhausts\PolymerSet$. If $\nu_\Polymer=0$, then $0\le\rho_\Polymer\le\nu_\Polymer=0$ and $\OPPR{\FiniteVolume}{\Polymer}(-\vec{\rho})=1-\rho_\Polymer=1$, independent of the cardinality of $\FiniteVolume$. For the remainder of the proof, we assume that $\vec{\rho}>\vec{0}$. The induction base is
\begin{equation*}
 \OPPR{\Set{\Polymer}}{\Polymer}(-\vec{\rho})
 = 1 - \rho_\Polymer
 \ge 1 - \frac{\rho_\Polymer}{\nu_\Polymer}
 = \frac{\nu_\Polymer-\rho_\Polymer}{\nu_\Polymer}\,.
\end{equation*}
For the induction step, let $\Set{\PolymerOther_1,\dotsc,\PolymerOther_m}:=\IncompatibleOtherPolymers{\Polymer}\cap\FiniteVolume$. Then
\begin{align*}
 \OPPR{\FiniteVolume}{\Polymer}(-\vec{\nu})
 &={} 1 - \frac{\nu_\Polymer}{ \prod_{i=1}^m
  \OPPR{
   \FiniteVolume\setminus\Set{\Polymer,\PolymerOther_1,\dotsc,\PolymerOther_{i-1}}
  }{\PolymerOther_i}(-\vec{\nu})
 }
 &\text{by }\eqref{eq_fi_oppr}\\
 &={}1 - \frac{\nu_\Polymer}{\rho_\Polymer}
  \frac{\rho_\Polymer}{ \prod_{i=1}^m
  \OPPR{
   \FiniteVolume\setminus\Set{\Polymer,\PolymerOther_1,\dotsc,\PolymerOther_{i-1}}
  }{\PolymerOther_i}(-\vec{\nu})
 }
 &\text{we assume }\rho_\Polymer>0\\
 &\le{} 1 - \frac{\nu_\Polymer}{\rho_\Polymer}
  \frac{\rho_\Polymer}{ \prod_{i=1}^m
  \OPPR{
   \FiniteVolume\setminus\Set{\Polymer,\PolymerOther_1,\dotsc,\PolymerOther_{i-1}}
  }{\PolymerOther_i}(-\vec{\rho})
 }
 &\text{as }\vec{\rho}\le\vec{\nu}\text{ and by }\eqref{eq_oppr_mon_parameter}\\
 &\le{} 1 - \frac{\nu_\Polymer}{\rho_\Polymer}
  (1-\OPPR{\FiniteVolume}{\Polymer}(-\vec{\rho}))
 &\text{by }\eqref{eq_fi_oppr}\\
\end{align*}
Therefore,
\begin{equation*}
 \OPPR{\FiniteVolume}{\Polymer}(-\vec{\rho})
 \ge \frac{\nu_\Polymer-\rho_\Polymer}{\nu_\Polymer}
  + \frac{\rho_\Polymer}{\nu_\Polymer}
  \OPPR{\FiniteVolume}{\Polymer}(-\vec{\nu})
 \ge \frac{\nu_\Polymer-\rho_\Polymer}{\nu_\Polymer}\,.
\end{equation*}
\end{proof}

\begin{proof}[Proof of~\eqref{eq_scub_inhom_red}]
The first part of this proof shows that uniform bounds on a certain subclass of the one polymer partition ratios are already sufficient, while the second part gives such a bound.\\

We say that a one polymer partition ratio, identified with $(\FiniteVolume,\Polymer)$, is \emph{escaping}, if $\IncompatibleOtherPolymers{\Polymer}\setminus\FiniteVolume\not=\emptyset$. If $(\FiniteVolume,\Polymer)$ is escaping, then every polymer in $\IncompatibleOtherPolymers{\Polymer}\setminus\FiniteVolume$ is an \emph{escape} of $(\FiniteVolume,\Polymer)$. Suppose we have, for $\vec{\rho}\ge\vec{0}$, a uniform bound on escaping one polymer partition ratios. That is, the following limit is well-defined and positive:
\begin{equation}\label{eq_unifBound_escaping}
 \ForAll\Polymer\in\PolymerSet,
 \EscapePolymer\in\IncompatibleOtherPolymers{\Polymer}:
 \quad
 \OPPR{\PolymersMinusEscapeSet}{\Polymer}(-\vec{\rho})
 :=\lim_{\Polymer\in\FiniteVolume\Exhausts\PolymersMinusEscapeSet}
  \OPPR{\PolymersMinusEscapeSet}{\Polymer}(-\vec{\rho})
 >0\,.
\end{equation}
Fix $\FiniteVolume\FiniteSubsetOf\PolymerSet$. As $(\PolymerSet,\IncompatibleTo)$ is connected, there is at least one polymer $\PolymerOther_{\Cardinality{\FiniteVolume}}\in\FiniteVolume$ with $(\FiniteVolume,\PolymerOther_{\Cardinality{\FiniteVolume}})$ escaping. Iterating this argument, for $\FiniteVolume\setminus\Set{\PolymerOther_{\Cardinality{\FiniteVolume}}}$, we obtain
\begin{equation}\label{eq_telescoping_escaping}
 \PartitionFunction{\FiniteVolume}(\vec{z})
 =\prod_{i=1}^n \OPPR{\FiniteVolume_i}{\PolymerOther_i}(\vec{z})\,,
\end{equation}
where $\FiniteVolume_i:=\Set{\PolymerOther_1,\dotsc,\PolymerOther_i}$ and all factors are escaping. Therefore~\eqref{eq_unifBound_escaping} implies $\vec{\rho}\in\PDPS$.\\

The second part shows that $\vec{\rho}\,\LeopRed{\Polymer}(\vec{\mu})\le\vec{\mu}$~\eqref{eq_scub_inhom_red} implies
\begin{equation}\label{eq_unifBound_inhom_red}
 \ForAll\Polymer\in\PolymerSet
 ,\EscapePolymer\in\IncompatibleOtherPolymers{\Polymer}
 ,\Polymer\in\FiniteVolume\FiniteSubsetOf\PolymersMinusEscapeSet:
 \quad
 \OPPR{\FiniteVolume}{\Polymer}(-\vec{\rho})
 \ge\frac{1}{1+\mu_\Polymer}\,.
\end{equation}
The claim follows by induction over $\Cardinality{\FiniteVolume}$, simultaneously over all $\Polymer\in\PolymerSet$ and $\EscapePolymer\in\IncompatibleOtherPolymers{\Polymer}$. The induction base is
\begin{equation*}
 \OPPR{\Set{\Polymer}}{\Polymer}(-\vec{\rho})
 = 1 - \rho_\Polymer
 \ge 1 - \frac{\mu_\Polymer}{\LeopRed{\Polymer}(\vec{\mu})}
 \ge 1- \frac{\mu_\Polymer}{1+\mu_\Polymer}
 = \frac{1}{1+\mu_\Polymer}\,.
\end{equation*}
For the induction step let $(\FiniteVolume,\Polymer)$ be escaping and let $\Set{\PolymerOther_1,\dotsc,\PolymerOther_m}:=\FiniteVolume\cap\IncompatibleOtherPolymers{\Polymer}$. Then $m\le D_\Polymer-1$ and every $(\FiniteVolume\setminus\Set{\Polymer,\PolymerOther_1,\dotsc,\PolymerOther_{i-1}},\PolymerOther_i)$ on the rhs of~\eqref{eq_fi_oppr} is escaping, too. Hence,
\begin{align*}
 \OPPR{\FiniteVolume}{\Polymer}(-\vec{\rho})
 &={} 1 - \frac{\rho_\Polymer}{ \prod_{i=1}^m
  \OPPR{
   \FiniteVolume\setminus\Set{\Polymer,\PolymerOther_1,\dotsc,\PolymerOther_{i-1}}
  }{\PolymerOther_i}(-\vec{\rho})
 }
 &\text{by }\eqref{eq_fi_oppr}\\
 &\ge{} 1 - \frac{\rho_\Polymer}{ \prod_{i=1}^m (1-\mu_{\PolymerOther_i})^{-1}}
 &\text{by induction}\\
 &\ge{} 1 - \frac{\mu_\Polymer}{1+\mu_\Polymer}
  \times\frac{
   {\displaystyle
   \sup_{\EscapePolymer\in\IncompatibleOtherPolymers{\Polymer}}
   \prod_{\PolymerOther'\in\IncompatibleOtherPolymers{\Polymer}\setminus\Set{\EscapePolymer}}
   (1+\mu_{\PolymerOther'})^{-1}
   }
  }
  {\prod_{i=1}^m (1-\mu_{\PolymerOther_i})^{-1}}
 &\text{by }\eqref{eq_scub_inhom_red}\\
 &\ge{} 1 - \frac{\mu_\Polymer}{1+\mu_\Polymer}\times 1
  =\frac{1}{1+\mu_\Polymer}
 &\text{cancelling}\,.
\end{align*}
\end{proof}
\section{The tree-operator framework}
\label{sec_treeOperatorFramework}
This section discusses tree-operators and how to derive SCUBs from them. Section~\ref{sec_uniformBounds} introduces the escaping pinned series, a variant of the classic pinned series, and shows how bounds on the escaping pinned series translate into uniform bounds. Section~\ref{sec_depthOneTreeOperators} discusses depth one tree-operators: a depth-wise summation of trees converges, if one controls the one-step extensions of the leaves by trees of depth 1. Section~\ref{sec_penroseIdentity} recalls Penrose's identity, which expresses the Ursell function of a cluster as a sum over a subset of the spanning trees of the cluster.\\

Section~\ref{sec_boundingTheEscapedPinnedSeries} combines these tools: a good local control of the trees present in the rewritten series yields a bound of the escaping pinned series. Combined with section~\ref{sec_uniformBounds}, the local control condition becomes a template for a SCUB. To carry the extra information needed for the local control over between successive applications of the derived tree-operator, we work on a space indexed by the set of incompatible polymer pairs instead of polymers. This is the multiplexing step in section~\ref{sec_boundingTheEscapedPinnedSeries}.
\subsection{Uniform bounds}
\label{sec_uniformBounds}
The classic approach to uniform bounds~\eqref{eq_uniformBounds} is to bound the \emph{pinned series}, the limit case of~\eqref{eq_pinnedConnectedFunction} in~\eqref{eq_uniformBounds_pinned} at $\vec{\rho}\in[0,\infty[^\PolymerSet$:
\begin{equation}\label{eq_pinnedSeries}
 \rho_\Polymer\PinnedSeries{\Polymer}(\vec{\rho})
 := \rho_\Polymer\lim_{\FiniteVolume\Exhausts\PolymerSet}
  \PinnedConnectedExpression{\FiniteVolume}{\Polymer}(-\vec{\rho})
 = \rho_\Polymer\sum_{n\ge 0}
   \frac{1}{\Factorial{n}}
  \sum_{\vec{\PolymerOther}\in\PolymerSet^n}
   \Modulus{\Ursell{\Polymer,\vec{\PolymerOther}}}
  \prod_{i=1}^n \rho_{\PolymerOther_i}\,.
\end{equation}

Recall that a pair $(\FiniteVolume,\Polymer)$ is \emph{escaping}, if $\IncompatibleOtherPolymers{\Polymer}\setminus\FiniteVolume\not=\emptyset$, with every $\EscapePolymer\in\IncompatibleOtherPolymers{\Polymer}\setminus\FiniteVolume$ called an \emph{escape} of $(\FiniteVolume,\Polymer)$. The \emph{set of escape pairs} is
\begin{equation}\label{eq_incompatiblePolymerPairs}
 \EscapePairSet
 :=\Set{
  \EscapePair:
  \Polymer\in\PolymerSet,
  \EscapePolymer\in\IncompatibleOtherPolymers{\Polymer}
 }\,.
\end{equation}
The \emph{escaping pinned series} is the limit of the escaping pinned series over finite volumes and at negative real fugacity. It is indexed by the escape pairs $\EscapePair\in\EscapePairSet$ from~\eqref{eq_incompatiblePolymerPairs}:

\begin{equation}\label{eq_pinnedSeriesEscaping}
\begin{aligned}
 \rho_\Polymer\PinnedSeriesEscaping{\EscapePair}(\vec{\rho})
 &:= \rho_\Polymer\lim_{\FiniteVolume\Exhausts\PolymersMinusEscapeSet}
  \PinnedConnectedExpression{\FiniteVolume}{\Polymer}(-\vec{\rho})
 \\&= \sum_{n\ge 0}
   \frac{1}{\Factorial{n}}
  \sum_{\vec{\PolymerOther}\in(\PolymersMinusEscapeSet)^n}
   \Ursell{\Polymer,\PolymerOther_1,\dotsc,\PolymerOther_n}
  \prod_{i=1}^n (-\rho_{\PolymerOther_i})\,.
\end{aligned}
\end{equation}

Bounding $\OPPR{\PolymersMinusEscapeSet}{\Polymer}(-\vec{\rho})$ from below or $\PinnedSeriesEscaping{\EscapePair}(\vec{\rho})$ from above are equivalent by the same reasoning as outlined after~\eqref{eq_uniformBounds}. \AlsoArxiv{The details are in~\eqref{eq_product_escaping}. }The link between bounding the escaping pinned series and uniform bounds is:

\begin{Prop}
Let $\vec{0}\le\vec{\nu}<\vec{\rho}\in[0,\infty[^\PolymerSet$. If $\PinnedSeriesEscaping{\EscapePair}(\vec{\rho})<\infty$, for every $\EscapePair\in\EscapePairSet$, then we have uniform bounds as in $\eqref{eq_uniformBounds}$ for $\vec{\nu}$.
\end{Prop}

\begin{proof}
Let $\FiniteVolume\FiniteSubsetOf\PolymerSet$. By telescoping~\eqref{eq_telescoping} and because $(\PolymerSet,\IncompatibleTo)$ is connected, we may write $\PartitionFunction{\FiniteVolume}(-\vec{\rho})$ as a product of escaping one polymer ratios. This is done by induction over $\Cardinality{\FiniteVolume}$. The upper bound on $\PinnedSeriesEscaping{\EscapePair}(\vec{\rho})$ translates into lower bounds for $\OPPR{\PolymersMinusEscapeSet}{\Polymer}(-\vec{\rho})$ and consequently for every escaping one polymer ratio, uniform in the volume. Therefore, $\vec{\rho}\in\PDPS$ and the uniform bounds for $\vec{\nu}$ follow from an application of proposition~\ref{prop_scub_generic} to the bounds on $\PinnedSeriesEscaping{\EscapePair}(\vec{\rho})$.
\end{proof}
\subsection{Depth one tree-operators}
\label{sec_depthOneTreeOperators}

Let $\RTsSize{n}$ be the set of trees with vertex set $\IntRange{0,n}$ and root $0$. Let $\RTsSizeFinite := \biguplus_{n\ge 0} \RTsSize{n}$. Let $\Stars$ be the set of \emph{stars}. A star $s\in\Stars$ has leaves $\StarLeaves{s}$ and $\StarPermutations{s}$ rooted automorphisms, equivalent to permutations of its leaves. Let $\LabelSpace$ be a countable set of labels. A function $c:\LabelSpace^{s}\to [0,\infty[$ is \emph{$s$-invariant}, if it is invariant under all rooted automorphism of the star $s$. Finally, $\StarAt{\RTVarLabel}{v}$ is the star with root $v$ and $v$'s children in the tree $\RTVarLabel\in\RTsSizeFinite$ as its leaves.

\begin{Prop}[Extension of~{\cite[Proposition 7]{Fernandez_Procacci__ClusterExpansionForAbstractPolymerModels_NewBoundsFromAnOldApproach__CMP_2007}}]\label{prop_depthOneTreeOperator}
Let $(c_s)_{s\in\Stars}$ be a collection of $s$-invariant, non-negative functions. Let $X:=[0,\infty]^{\LabelSpace}$ and $\vec{\rho}\in X$. Define the tree-operator $T: X\to X$, for each $l\in\LabelSpace$, by
\begin{subequations}\label{eq_depthOneTreeOperator}
\begin{equation}\label{eq_depthOneTreeOperator_Definition}
 \vec{\mu}\mapsto
 [T(\vec{\mu})]_l
 :=
  \rho_{l}
  \sum_{s\in\Stars}
   \frac{1}{\StarPermutations{s}}
  \sum_{\vec{\lambda}\in \Set{l}\times\LabelSpace^{\Cardinality{\StarLeaves{s}}}}
   c_s(\vec{\lambda})
  \prod_{v\in\StarLeaves{s}} \mu_{\lambda_v}\,.
\end{equation}
If there exists $\vec{\mu}\in X$ such that
\begin{equation}\label{eq_depthOneTreeOperator_VectorCondition}
 T(\vec{\mu})\le\vec{\mu}\,,
\end{equation}
then the family of series, indexed by $\LabelSpace$,
\begin{equation}\label{eq_depthOneTreeOperator_SeriesDefinition}
 R_l(\vec{\rho}):=
 \sum_{n\ge 0}
  \frac{1}{\Factorial{n}}
 \sum_{\vec{\lambda}\in\Set{l}\times\LabelSpace^n,\RTVarSize\in\RTsSize{n}}
 \prod_{i\in\IntRange{0,n}}
  c_{\StarAt{\RTVarSize}{i}}
   (\vec{\lambda}_{\StarAt{\RTVarSize}{i}})
  \rho_{\lambda_i}
\end{equation}
converges uniformly, as
\begin{equation}\label{eq_depthOneTreeOperator_ConvergenceAndMajoration}
 R(\vec{\rho})
 =\lim_{n\to\infty} T^n(\vec{\rho})
 =T(R(\vec{\rho}))
 \le \vec{\mu}\,.
\end{equation}
\end{subequations}
\end{Prop}
\subsection{Penrose's identity}
\label{sec_penroseIdentity}
Let $G:=(V,E)$ be a finite, simple graph. The \emph{set of spanning subgraphs} of $G$ is $\SpanningSubgraphs{G}$. It is a \emph{poset} under the \emph{partial order} given by
\begin{equation*}
 H\SubGraph H' \Iff \Edges{H}\subseteq\Edges{H'}\,.
\end{equation*}
For $L,U\in\SpanningSubgraphs{G}$ with $L\SubGraph U$, the \emph{interval} from $L$ to $U$ is
\begin{equation*}
 [L,U]:=\Set{H\in\SpanningSubgraphs{G}: L\SubGraph H\SubGraph U}\,.
\end{equation*}
Let $\SpanningTrees{G}\subseteq\SpanningSubgraphs{G}$ be the \emph{set of spanning trees} of $G$. A function $\SchemeAnonymous:\SpanningTrees{G}\to\SpanningSubgraphs{G}$ is a \emph{partition scheme} of $\SpanningSubgraphs{G}$~\cite[before (5)]{Penrose__ConvergenceOfFugacityExpansionsForClassicalSystems__SMFA_1967}, if
\begin{equation}\label{eq_partitionScheme}
 \SpanningSubgraphs{G} :=
 \biguplus_{\Tree\in\SpanningTrees{G}} [\Tree,\SchemeAnonymous(\Tree)]\,.
\end{equation}
The \emph{set of singleton trees with respect to $S$} (short \emph{$S$-trees}) is
\begin{equation}\label{eq_singletonTrees}
 \SingletonTrees{\SchemeAnonymous}(G)
 :=\Set{\Tree\in\SpanningTrees{G}: \Tree = \SchemeAnonymous(\Tree)}\,.
\end{equation}

\begin{Thm}[Penrose {\cite[(5)]{Penrose__ConvergenceOfFugacityExpansionsForClassicalSystems__SMFA_1967}}]\label{thm_penroseIdentity}
If $\SchemeAnonymous$ is a partition scheme of $\SpanningSubgraphs{G}$, then
\begin{equation}\label{eq_penroseIdentity}
 \sum_{H\in\SpanningSubgraphs{G}} (-1)^{\Cardinality{\Edges{H}}}
 = (-1)^{\Cardinality{V}-1}
  \Cardinality{\SingletonTrees{\SchemeAnonymous}(G)}\,.
\end{equation}
\end{Thm}

\begin{Rem}
The number of $\SchemeAnonymous$-trees is independent of $\SchemeAnonymous$.
\end{Rem}

\AlsoArxiv{The proof of theorem~\ref{thm_penroseIdentity} is in section~\ref{sec_penroseProof}.}

\subsection{Bounding the escaped pinned series}
\label{sec_boundingTheEscapedPinnedSeries}
This section combines tree-operators and Penrose's identity to bound the escaping pinned series. Let $\EscapePair\in\EscapePairSet$ and $\vec{\PolymerOther}\in(\PolymersMinusEscapeSet)^n$. If $\FiniteVolume:=\Set{\Polymer,\PolymerOther_1,\dotsc,\PolymerOther_n}$, then $(\FiniteVolume,\Polymer)$ is escaping with escape $\EscapePolymer$. We assume that $G:=\Cluster{\Polymer,\vec{\PolymerOther}}$ is a cluster. Choose a partition scheme $\SchemeAnonymous_G$ of $G$ and apply Penrose's identity~\eqref{eq_penroseIdentity}:

\begin{equation}\label{eq_penroseApplication}
 \Ursell{\Polymer,\vec{\PolymerOther}}
 = \sum_{H\in\SpanningSubgraphs{G}} (-1)^{\Cardinality{\Edges{H}}}
 = (-1)^n \Cardinality{\SingletonTrees{\SchemeAnonymous_G}(G)}\,.
\end{equation}
Every $\RTVarLabel\in\SingletonTrees{\SchemeAnonymous_G}(G)$ is rooted at $0$ with label $\Polymer$ and its non-root vertices are labelled by $\vec{\PolymerOther}$. We extend the labels to $\EscapePairSet$:
\begin{itemize}
 \item the root label becomes $\EscapePair$
 \item the non-root vertex $i$ with label $\PolymerOther_i$ gets label $(\PolymerOther_i,\EscapePolymer_i)$, where $\EscapePolymer_i$ is the label of the parent of $i$ in $\RTVarLabel$. As $G$ is a cluster, we know that $\EscapePolymer_i\in\IncompatibleOtherPolymers{\PolymerOther_i}$.
\end{itemize}

\begin{Prop}\label{prop_multiplexing}
\begin{subequations}\label{eq_multiplexing}
Let $X:=[0,\infty]^\PolymerSet$ and $Y:=[0,\infty]^{\EscapePairSet}$. Let $\EscapeInjection$ be the injection from $X$ into $Y$, \emph{multiplexing} values by ignoring the escape coordinate. Let $(c_n)_{n\in\NatNumZero}$ be a family of functions with $c_n: \EscapePairSet^{n+1}\to\Set{0,1}$. For all $\EscapePair\in\EscapePairSet$ and $\vec{\PolymerOther}\in(\PolymersMinusEscapeSet)^n$, let $\SchemeAnonymous_G$ be a partition scheme on the cluster $G:=\Cluster{\Polymer,\vec{\PolymerOther}}$ with
\begin{equation}\label{eq_multiplexing_localControl}
 \ForAll\RTVarSize\in\RTsSize{n}:\quad
 \Iverson{\RTVarSize\in\SingletonTrees{\SchemeAnonymous_G}}
 \le \prod_{i=0}^n
  c_{\Cardinality{\StarLeaves{\StarAt{\RTVarSize}{i}}}}
   ((\PolymerOther_j,\EscapePolymer_j)_{j\in\Vertices{\StarAt{\RTVarSize}{i}}})\,,
\end{equation}
Define the tree-operator $\LeopMulti{}: Y\to Y$, for each $\EscapePair\in\EscapePairSet$, by
\begin{equation}\label{eq_multiplexing_operator}
 \vec{u}\mapsto \LeopMulti{\EscapePair}(\vec{u})
  := \underbrace{\EscapeInjection(\vec{\rho})_{\EscapePair}}_{=\rho_\Polymer}
   \sum_{s\in\Stars} \frac{1}{\StarPermutations{s}}
   \sum_{\vec{\lambda}\in \Set{\EscapePair}\times\EscapePairSet^{\Cardinality{\StarLeaves{s}}}}
   \prod_{v\in\StarLeaves{s}} u_{\lambda_v}
 \,.
\end{equation}
Let $\LeopSup{}: X\to X$ be the local supremum of $\LeopMulti{}$:
\begin{equation}\label{eq_multiplexing_supremum}
 \ForAll\Polymer\in\PolymerSet:\qquad
 \vec{\mu}\mapsto \LeopSup{\Polymer}(\vec{\mu})
 := \sup_{\EscapePolymer\in\IncompatibleOtherPolymers{\Polymer}}
  \LeopMulti{\EscapePair}(\EscapeInjection(\vec{\mu}))
\end{equation}
If there is $\vec{\mu}\in[0,\infty[^\PolymerSet$ such that
\begin{equation}\label{eq_multiplexing_condition}
 \vec{\rho}\, \LeopSup{}(\vec{\mu})\le \vec{\mu}\,,
\end{equation}
then
\begin{equation}\label{eq_multiplexing_escapingPinnedSeriesBound}
 \ForAll\EscapePair\in\EscapePairSet:
 \qquad
 \rho_\Polymer\PinnedSeriesEscaping{\EscapePair}(\vec{\rho})
 \le\mu_\Polymer<\infty\,.
\end{equation}
\end{subequations}
\end{Prop}

\begin{Rem}
One may improve~\eqref{eq_multiplexing_escapingPinnedSeriesBound} by adapting the calculations in~\cite[appendix]{Bissacot_Fernandez_Proccaci__OnTheConvergenceOfClusterExpansionsForPolymerGases__JSP_2010}.
\end{Rem}

\begin{proof}
If~\eqref{eq_multiplexing_localControl} holds, then we have
\begin{equation*}
 \Modulus{\Ursell{\PolymerOther_0,\dotsc,\PolymerOther_n}}
 =
 \Cardinality{\SingletonTrees{\SchemeAnonymous_G}(G)}
 \le
 \sum_{\RTVarSize\in\RTsSize{n}}
 \prod_{i=0}^n
  c_{\Cardinality{\StarLeaves{\StarAt{\RTVarSize}{i}}}}
   ((\PolymerOther_j,\EscapePolymer_j)_{
    \Set{j\in\StarAt{\RTVarSize}{i}}
   })\,.
\end{equation*}
We apply the above bound term-wise to the escaping pinned series. The resulting expression is the limit of iterated applications of the tree-operator $\LeopMulti{}$. Condition~\eqref{eq_multiplexing_condition} relaxes to a condition on $\LeopMulti{}$ appropriate for~\eqref{eq_depthOneTreeOperator_VectorCondition}. Hence, we may apply~\eqref{eq_depthOneTreeOperator_ConvergenceAndMajoration} and obtain
\begin{equation*}
 \PinnedSeriesEscaping{\EscapePair}(\vec{\rho})
 \le\sum_{n\ge 0} \frac{1}{\Factorial{n}}
  \sum_{\vec{\PolymerOther}\in(\PolymersMinusEscapeSet)^n}
   \sum_{\RTVarSize\in\RTsSize{n}}
   \prod_{i=0}^n
    \rho_{\PolymerOther_i}
   c_{\Cardinality{\StarLeaves{\StarAt{\RTVarSize}{i}}}}
   ((\PolymerOther_j,\EscapePolymer_j)_{\Set{j\in\StarAt{\RTVarSize}{i}}})
 \le\mu_\Polymer\,.
\end{equation*}
\end{proof}

\section{Partition schemes}
\label{sec_partitionSchemes}
We recall the classic formulation of Penrose's scheme in section~\ref{sec_schemePenroseStatic}. The Penrose scheme has two drawbacks. First, the distance from the root is always the same in a connected subgraph and the corresponding spanning tree chosen by Penrose's scheme. Second, the scheme works on a general simple graph with ordered vertices. This disregards information from the additional structure present in clusters induced by the polymer labels of the vertices (see figure~\ref{fig_anatomy}).\\

The tree-operator framework exerts only local control over family of trees it admits~\eqref{eq_multiplexing_localControl}. A tree-operator can not encode the Penrose scheme's global property of distance-preservation. We reformulate the original, static description of Penrose's scheme as a dynamic process: by exploring a graph from the root at unit speed into all directions, flood-filling the graph, the distance to the root turns from an a priori invariant into an emergent property of the partition scheme. We contrast those viewpoints in sections~\ref{sec_schemePenroseStatic} and~\ref{sec_schemeGreedy}.\\

This reformulation allows trade-offs. We may, based on local information, choose not to advance the exploration boundary everywhere (cf. figure~\ref{fig_exploration}). This departure from the greedy behaviour of the dynamic reformulation of Penrose's scheme loses the distance preservation, but allows to encode more local information about the graph in the singleton trees of the partition scheme. In short, we loose a useless global constraint and gain local constraints on the singleton trees encodable by a tree-operator. This leads to explorative partition schemes in section~\ref{sec_explorativePartitionSchemes}.\\

Section~\ref{sec_schemeReturning} presents the returning scheme, which is an explorative partition scheme adapted to the structure of a cluster. It alleviates both drawbacks of the Penrose scheme. The synthetic scheme underlying the mixing SCUB interpolates between the greedy and returning schemes. We discuss it in section~\ref{sec_schemeSynthetic}.

\subsection{The Penrose scheme}
\label{sec_schemePenroseStatic}
We recapitulate the original formulation of the partition scheme by Penrose~\cite[(6)]{Penrose__ConvergenceOfFugacityExpansionsForClassicalSystems__SMFA_1967}. Let $I$ be a finite and totally ordered set. Let $G:=(I,E)$ be a connected graph with root $o$. We only present the part selecting the spanning tree in a subgraph $H$ of $G$. Section~\ref{sec_schemeGreedy} describes going from a spanning tree to the maximal subgraph containing that tree in algorithm~\ref{alg_treeEdgeComplementPartitionPenrose} and the relevant properties of the singleton trees on a cluster in proposition~\ref{prop_singletonsPenroseOnCluster}.

\begin{Alg}\label{alg_penrose}
Let $H\in\SpanningSubgraphs{G}$. Let $d$ be the distance function on $H$ from $o$ to the vertices $I$. For each edge $(i,j)$ in $H$ with $d(i)\ge d(j)$, the following procedure decides between the selection or removal of $(i,j)$, independently of the decisions for the other edges of $H$:
\begin{description}
 \item[\StepPenCousins{}] If $d(i)=d(j)$, then \ActionRemove{} $(i,j)$.
 \item[\StepPenElders{}] For $i\in I\setminus\Set{o}$, let $j_i:=\min\Set{j': (i,j')\in\Edges{H}, d(j')=d(i)-1}$.
 \item[\StepPenParent{}] If $d(i)=d(j)+1$ and $j=j_i$, then \ActionSelect{} $(i,j)$.
 \item[\StepPenUncles{}] If $d(i)=d(j)+1$ and $j\not=j_i$, then \ActionRemove{} $(i,j)$.
\end{description}
\end{Alg}
The result of these parallel and independent decisions is the Penrose tree of $H$. Its edges are $\Set{(i,j_i): i\in I\setminus\Set{o}}$. The distance between a vertex $i$ and the root $o$ is the same in both $H$ and the Penrose tree.
\subsection{Explorative partition schemes}
\label{sec_explorativePartitionSchemes}
Section~\ref{sec_boundingTheEscapedPinnedSeries} shows how to bound the pinned series with the help of a tree-operator. The key is the local control of a tree~\eqref{eq_multiplexing_localControl}: more constraints on the tree allow larger values of $\vec{\rho}$ in~\eqref{eq_multiplexing_condition}. This is easier, if the partition scheme places mainly local constraints facing away from the root on its singleton trees. This section describes a family of such schemes on a high level.\\

Let $G:=(V,E)$ be a simple graph. A partition scheme is \emph{explorative}, if it is described by an exploration algorithm coupled with a compatibility condition. The algorithm selects a spanning tree from a spanning subgraph of $G$. It does so by an iterative exploration of the subgraph. The compatibility condition is a sufficient condition to uniquely reconstruct the maximal subgraph reducing to a given spanning tree from said spanning tree.\\

Given $H\in\SpanningSubgraphs{G}$ and a root $o$, the generic \emph{exploration algorithm} $\ExplorationAlgorithm{\SchemeGeneric}$~\ref{alg_explorationGeneric} selects a spanning tree of $H$ by starting at a root $o$ and growing the tree iteratively: during each iteration it considers all the nodes neighbouring its furthest level, selects at least one of those neighbours, adds each selected vertex and one of the edges leading to it to the tree and removes all superfluous edges, including those to ignored vertices. This prescription guarantees the selection of a spanning tree of $H$ (see proposition~\ref{prop_explorationGeneric}).\\

We discriminate between two types of information: the \emph{static information} is our knowledge of the structure of $G$, the choice of the root $o$ and all other information we may have about $G$, such as labellings or orderings of elements. The \emph{dynamic information} is the information we gather about $H$ during its exploration, including properties as connectedness of yet unexplored parts and which elements have already been selected or removed.\\

The complementary procedure takes a tree $\Tree\in\SpanningTrees{G}$ and partitions $E\setminus\Edges{\Tree}$ into the \emph{admissible edges} $\AdmissibleEdgesTree{\SchemeGeneric}$ and the \emph{conflicting edges} $\ConflictingEdgesTree{\SchemeGeneric}$. If this partition is compatible with the exploration algorithm~\ref{alg_explorationGeneric}, then the preimage of the $\ExplorationAlgorithm{\SchemeGeneric}$ defines a partition scheme $\SchemeGeneric$ (see proposition~\ref{prop_schemeGeneric}). The compatibility is satisfied, if all the dynamic information used in $\ExplorationAlgorithm{\SchemeGeneric}$ to select at vertex on level $n$ of the $\Tree$ from $H$ is a function of the first $n$ levels of $\Tree$ (and the static information).\\

In the rest of this section, we formalise and prove the correctness of the outlined idea, resulting in the generic explorative partition scheme $\SchemeGeneric$. The generic scheme $\SchemeGeneric$ contains some gaps, where it abstracts over a concrete choice of the particular step of the algorithm. We fill these gaps in $\SchemeGeneric$'s specification in later sections, resulting in different partition schemes.

\begin{Alg}[$\SchemeGeneric$ exploration]\label{alg_explorationGeneric}
Let $H\in\SpanningSubgraphs{G}$. We construct a sequence $(H_k)_{k\in\NatNumZero}$ of subgraphs of $H$ starting with $H_0:=H$ and a sequence $(T_k)_{k\in\NatNumZero}$ of subsets of $V$ starting with $T_0:=\Set{o}$. We think of $T_k$ as the \emph{explored tree} part of $H_k$.\\

We construct $H_{k+1}$ and $T_{k+1}$ from $H_k$ and $T_k$ as follows: Let the \emph{unexplored part} be $U_k:=V\setminus T_k$. Let the \emph{potential nodes} $P_k$ be the set of neighbours of $T_k$ in $U_k$ and the \emph{boundary} $B_k$ be the set of neighbours of $U_k$ in $T_k$. Apply the following steps in order:

\begin{description}
 \item[\StepGenBoundary{}] For every connected component $C$ of $H_k|_{U_k}$, \ActionSelect{} a subset $S_k$ of $P_k$ containing at least one vertex from $C\cap P_k$. Call these the \emph{selected nodes}. Set $T_{k+1}:=T_k\uplus S_k$.
 \item[\StepGenIgnored{}] Let the \emph{ignored nodes} be $I_k:=P_k\setminus S_k$. \ActionRemove{} all the edges in $\Edges{B_k,I_k}\cap\Edges{H_k}$.
 \item[\StepGenParent{}] For each $v\in S_k$, \ActionSelect{} $(v,w_v)\in\Edges{B_k,S_k}\cap\Edges{H_k}$.
 \item[\StepGenUncles{}] For each $v\in S_k$, \ActionRemove{} all $(v,w)\in\Edges{H_k}$ with $w_v\not=w\in B_k$.
 \item[\StepGenCousins{}] \ActionRemove{} all of $\Edges{S_k}\cap\Edges{H_k}$.
\end{description}
\end{Alg}

\begin{Rem}
The only gaps in algorithm~\ref{alg_explorationGeneric} are the selection of $S_k$ in \StepGenBoundary{} and the parent in \StepGenParent{}. The selections induce the removals.
\end{Rem}

\begin{Rem}
The steps in algorithm~\ref{alg_explorationGeneric} factorize over connected components of $H_k|_{U_k}$. The algorithm may explore each connected component independently of the progress of the exploration of the other connected components.
\end{Rem}

\begin{Rem}
Figure~\ref{fig_exploration} visualises a step of algorithm~\ref{alg_explorationGeneric}.
\end{Rem}

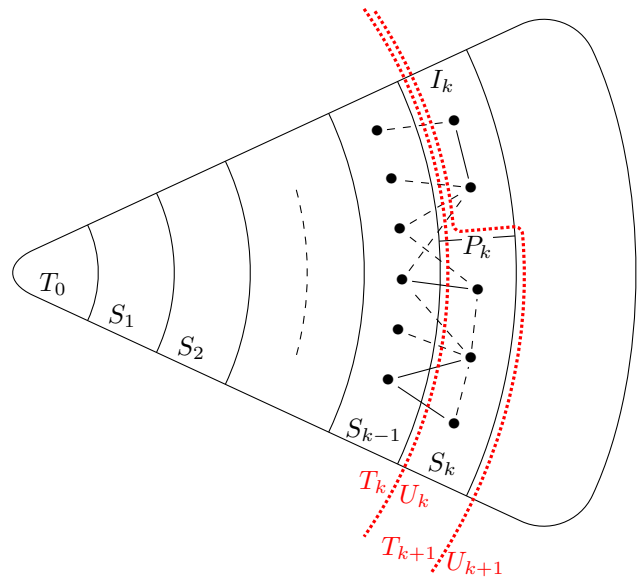
\begin{figure}
\begin{center}
\begin{tikzpicture}
  \newcommand{\arc}{25}
  \draw[rounded corners=20pt] (0,0) -- (\arc:8.5) arc (\arc:-\arc:8.5) -- cycle;
  \draw (\arc:1.5) arc (\arc:-\arc:1.5);
  \draw (\arc:2.5) arc (\arc:-\arc:2.5);
  \draw (\arc:3.5) arc (\arc:-\arc:3.5);
  \draw[dashed] (\arc-10:4.25) arc (\arc-10:-\arc+10:4.25);
  \draw (\arc:5) arc (\arc:-\arc:5);
  \draw (\arc:6) arc (\arc:-\arc:6);
  \draw (\arc:7) arc (\arc:-\arc:7);
  \draw (4:6) -- (4:7);

  \node[above] at (-\arc:1) {$T_0$};
  \node[above] at (-\arc:2) {$S_1$};
  \node[above] at (-\arc:3) {$S_2$};
  \node[above] at (-\arc:5.65) {$S_{k-1}$};
  \node[above] at (-\arc:6.65) {$S_k$};
  \node[below] at (\arc:6.65) {$I_k$};
  \node[fill=white, inner sep=0pt] at (3:6.5) {$P_k$};

  \drawnodebullet{t1}{-15:5.5}
  \drawnodebullet{t2}{-8:5.5}
  \drawnodebullet{t3}{-1:5.5}
  \drawnodebullet{t4}{6:5.5}
  \drawnodebullet{t5}{13:5.5}
  \drawnodebullet{t6}{20:5.5}

  \drawnodebullet{u1}{-18:6.5}
  \drawnodebullet{u2}{-10:6.5}
  \drawnodebullet{u3}{-2:6.5}
  \drawnodebullet{u4}{10:6.5}
  \drawnodebullet{u5}{18:6.5}

  \draw[dashed] (u1) -- (u2) -- (u3)
                (u2) -- (t3)
                (u2) -- (t2)
                (t3) -- (u4)
                (t4) -- (u3)
                (t4) -- (u4) -- (t5)
                (t6) -- (u5);
  \draw (t1) -- (u1)
        (t1) -- (u2)
        (t3) -- (u3)
        (u4) -- (u5);

  \draw[marked] (\arc+10:6.1)
                arc (\arc+10:-\arc-10:6.1);
  \draw[marked, rounded corners=3pt] (\arc+9:6.2)
                                     arc (\arc+9:5:6.2)
                                     -- (5:7.1)
                                     arc (5:-\arc-9:7.1);
  \draw[color=red] (-\arc-2:6.1) node[left] {$T_k$}
                   (-\arc-4:6.1) node[right] {$U_k$}
                   (-\arc-6:7.1) node[left] {$T_{k+1}$}
                   (-\arc-8:7.1) node[right] {$U_{k+1}$};
\end{tikzpicture}
\end{center}
\caption[The exploration algorithm]{The generic exploration algorithm~\ref{alg_explorationGeneric} constructing $H_{k+1}$ from $H_k$. The notation is from algorithm~\ref{alg_explorationGeneric}. The dotted lines show the exploration boundary between explored and unexplored parts of the graph at two consecutive steps of the exploration algorithm. The exploration boundary advances in at least one vertex of every connected component of $H_k|_{U_k}=G|_{U_k}$. At step $k$, edges between $T_k$ and $I_k$ are ignored, while the boundary advances to the vertices in $S_k$.}
\label{fig_exploration}
\end{figure}

\begin{Prop}
The following invariants hold in algorithm~\ref{alg_explorationGeneric}, $\ForAll k\in\NatNumZero$:
\begin{subequations}\label{eq_genericInvariants}
\begin{gather}
 \label{eq_geninv_subgraph}
 H_{k+1}\SubGraph H_k\\
 \label{eq_geninv_Tree}
 H_k|_{T_k}\text{ is a tree}\\
 \label{eq_geninv_Spanning}
 H_k\in\SpanningSubgraphs{G}\\
 \label{eq_geninv_stableTree}
 H_{k+1}|_{T_k}=H_k|_{T_k}\\
 \label{eq_geninv_stableRest}
 H_k|_{U_k}=H|_{U_k}\\
 \label{eq_geninv_affectedEdges}
 \left(\Edges{U_k}\uplus\Edges{U_k,B_k}\right)\cap\Edges{H_k}
 =\left(\Edges{U_k}\uplus\Edges{U_k,B_k}\right)\cap\Edges{H}\\
 \label{eq_geninv_rootDistance}
 \ForAll l\ge k:\quad
 v\in S_k\Iff\GraphMetricOf{H_l}{o}{v}=k\\
 \label{eq_geninv_BSRelation}
 B_{k+1}\subseteq S_k
\end{gather}
\end{subequations}
\end{Prop}

\begin{Rem}
We may replace every occurrence of $\Edges{H_k}$ in algorithm~\ref{alg_explorationGeneric} by $\Edges{H}$. This follows from invariant~\eqref{eq_geninv_affectedEdges}, which asserts that, for each edge, there is exactly one iteration of the exploration algorithm during which it is either selected or removed.
\end{Rem}

\begin{proof}
Invariant~\eqref{eq_geninv_subgraph} is clear, as we obtain $H_{k+1}$ from $H_k$ by removing edges. Invariant~\eqref{eq_geninv_stableRest} follows from~\eqref{eq_geninv_affectedEdges}. Invariant~\eqref{eq_geninv_stableTree} follows from~\eqref{eq_geninv_Tree} and the subgraph relationship~\eqref{eq_geninv_subgraph}.\\

\eqref{eq_geninv_BSRelation}: \StepGenIgnored{} ensures that all connections in $H_{k+1}$ between $B_k$ and $U_{k+1}$ pass through $S_k$. By induction over $k$, these are all connections in $H_{k+1}$ between $T_k$ and $U_{k+1}$.\\

We prove the remaining invariants by simultaneous induction over $k$. They hold trivially for $k=0$.\\

\eqref{eq_geninv_Tree}: The subgraph $H_{k+1}|_{T_{k+1}}$ consists of $H_{k+1}|_{T_k}$, the vertices $S_k$ and the edge sets $\Edges{T_k,S_k}\cap\Edges{H_{k+1}}$ and $\Edges{S_k}\cap\Edges{H_{k+1}}$. By~\eqref{eq_geninv_stableTree} and~\eqref{eq_geninv_Tree},  $H_{k+1}|_{T_k}=H_k|_{T_k}$ is a tree. \StepGenCousins{} ensures that $\Edges{S_k}\cap\Edges{H_{k+1}}=\emptyset$. \StepGenParent{} and \StepGenUncles{} ensure that each $v\in S_k$ has a unique neighbour in $T_k$ in $H_{k+1}|_{T_{k+1}}$. Thus, $H_{k+1}|_{T_{k+1}}$ is a tree.\\

\eqref{eq_geninv_Spanning}: The tree $H_{k+1}|_{T_{k+1}}$~\eqref{eq_geninv_Tree} is connected. If $v\in U_{k+1}$, then \StepGenBoundary{} asserts that there exists a connected component $C$ of $H_k|_{U_k}$ such that $v\in C$. \StepGenBoundary{} also asserts that there exists a $w\in C\cap S_k$ such that: \StepGenParent{} asserts that there exists a $z\in T_k$ with $z\AdjacentTo w$ in $H_k$ and from~\eqref{eq_geninv_stableRest} it follows that $H_{k+1}|_{U_{k+1}}=H_k|_{U_{k+1}}$. Thus, $v\ConnectedTo w\ConnectedTo z\ConnectedTo o$ and $H_{k+1}$ is connected.\\

\eqref{eq_geninv_affectedEdges}: During iteration $k$, \StepGenIgnored{}, \StepGenUncles{} and \StepGenCousins{} only remove edges in $\Edges{B_k,P_k}$ or $\Edges{S_k}$. Thus, $\Edges{U_{k+1}}\cap\Edges{H_k}$ and, by~\eqref{eq_geninv_rootDistance}, $\Edges{B_{k+1},U_{k+1}}\cap\Edges{H_k}=\Edges{S_k,U_{k+1}}\cap\Edges{H_k}$ are subsets of $\Edges{H_{k+1}}$.\\

\eqref{eq_geninv_rootDistance}: If $v\in T_k$, then $\GraphMetricOf{H_{k+1}}{o}{v}=\GraphMetricOf{H_k}{o}{v}=k$, as $H_{k+1}|_{T_k}=H_k|_{T_k}$ by~\eqref{eq_geninv_stableTree}. It remains to show that
\begin{equation*}
 v\in S_k
 \Iff \GraphMetricOf{H_{k+1}}{o}{v}=k+1\,.
\end{equation*}
If $\GraphMetricOf{H_{k+1}}{o}{v}=k+1$, then $v\not\in T_k:=\biguplus_{l=0}^k B_k$ and there exists a $w$ with $(v,w)\in\Edges{H_{k+1}}$ and $\GraphMetricOf{H_{k+1}}{o}{w}=k$. Thus, $w\in B_k$. By \StepGenIgnored{} and \StepGenUncles{}, the only neighbours of $B_k$ in $U_k$ in the graph $H_{k+1}$ are those in $S_k$.\\

On the other hand, if $v\in S_k$, then $H_{k+1}|_{T_{k+1}}$ is a tree~\eqref{eq_geninv_Tree} and $v$ has a unique parent $w\in B_k\subseteq S_{k-1}$~\eqref{eq_geninv_BSRelation}. As $w$ is $v$'s only neighbour at distance $k$ from the root in $H_{k+1}$~\eqref{eq_geninv_rootDistance}, we have $\GraphMetricOf{H_{k+1}}{o}{v}=\GraphMetricOf{H_{k+1}}{o}{w}+1=k+1$.
\end{proof}

\begin{Prop}\label{prop_explorationGeneric}
Algorithm~\ref{alg_explorationGeneric} describes a function $\ExplorationAlgorithm{\SchemeGeneric}:\SpanningSubgraphs{G}\to\SpanningTrees{G}$. It is monotone decreasing: $\ExplorationAlgorithm{\SchemeGeneric}(H)\SubGraph H$, for each $H\in\SpanningSubgraphs{G}$.
\end{Prop}

\begin{proof}
Let $H\in\SpanningSubgraphs{G}$ and apply algorithm~\ref{alg_explorationGeneric} to it. If $T_k=V$, then $H_k=H_{k+1}$. As long as $T_k\not=V$, $H_k\in\SpanningSubgraphs{G}$ by~\eqref{eq_geninv_subgraph} and we have $S_k\not=\emptyset$. Therefore, $(T_k)_{k\in\NatNumZero}$ grows strictly monotone to and then stabilises in $V$ in at most $\Cardinality{V}$ steps. Thus, the monotone decreasing~\eqref{eq_geninv_subgraph} sequence $(H_k)_{k\in\NatNumZero}$ of subgraphs of $H$ stabilises in $H_{\Cardinality{V}}$. Finally,~\eqref{eq_geninv_Tree} asserts that $H_{\Cardinality{V}}=H_{\Cardinality{V}}|_{T_\Cardinality{V}}$ is a tree and~\eqref{eq_geninv_Spanning} that $H_{\Cardinality{V}}\in\SpanningSubgraphs{H}$. Hence, $H_{\Cardinality{V}}\SubGraph H$.
\end{proof}

\begin{Prop}\label{prop_schemeGeneric}
Let $\ExplorationAlgorithm{\SchemeGeneric}$ be as in proposition~\ref{prop_explorationGeneric}. For each $\Tree\in\SpanningTrees{G}$, partition $\Edges{\Tree}$ into two sets: the admissible edges $\AdmissibleEdgesTree{\SchemeGeneric}$ and the conflicting edges $\ConflictingEdgesTree{\SchemeGeneric}$. If, for each $H\in\ExplorationAlgorithm{\SchemeGeneric}^{-1}(\Tree)$, we have
\begin{subequations}\label{eq_treeEdgeComplementPartitionCompatibility}
\begin{equation}\label{eq_adm_gen}
 \ForAll e\in\AdmissibleEdgesTree{\SchemeGeneric}\setminus\Edges{H}:
 \quad
 \ExplorationAlgorithm{\SchemeGeneric}(V,\Edges{H}\uplus\Set{e})=\Tree
\end{equation}
and
\begin{equation}\label{eq_conf_gen}
 \ForAll \emptyset\not=C\subseteq\ConflictingEdgesTree{\SchemeGeneric}\setminus\Edges{H}:
 \quad
 \ExplorationAlgorithm{\SchemeGeneric}(V,\Edges{H}\uplus C)\not=\Tree\,,
\end{equation}
\end{subequations}
then the map
\begin{equation}\label{eq_schemeGeneric}
 \SchemeGeneric:
 \quad
 \SpanningTrees{G}\to\SpanningSubgraphs{G}
 \quad
 \Tree\mapsto\SchemeGeneric(\Tree)
 :=(V,\Edges{\Tree}\uplus\AdmissibleEdgesTree{\SchemeGeneric})
\end{equation}
is a partition scheme of $G$ with $[\Tree,\SchemeGeneric(\Tree)]=\ExplorationAlgorithm{\SchemeGeneric}^{-1}(\Tree)$.
\end{Prop}

\begin{Rem}
In our specific partition schemes, the exploration is so local, that also~\eqref{eq_conf_gen} only needs to be checked for single edges.
\end{Rem}

\begin{proof}
Fix $\Tree\in\SpanningTrees{G}$. For $A\subseteq\AdmissibleEdgesTree{\SchemeReturning}$ and $C\subseteq\ConflictingEdgesTree{\SchemeReturning}$, let $H_{(A,C)}:=(V,\Edges{\Tree}\uplus A\uplus C)$. We show that $H_{(A,C)}\in[\Tree,\SchemeGeneric(\Tree)]$, iff $C=\emptyset$.\\

Case $C=\emptyset$: We argue by induction over the cardinality of $A$. For the induction base with $A=\emptyset$ we have $H_{(\emptyset,\emptyset)}=\Tree$ and $\ExplorationAlgorithm{\SchemeGeneric}(\Tree)=\Tree$. For the induction step consider $A:=A'\uplus\Set{e}$. By the induction hypothesis, $\ExplorationAlgorithm{\SchemeGeneric}(H_{(A',\emptyset)})=\Tree$. As $e\in\AdmissibleEdgesTree{\SchemeGeneric}$, we apply~\eqref{eq_adm_gen} to see that $\ExplorationAlgorithm{\SchemeGeneric}(H_{(A,\emptyset)})=\Tree$, too.\\

Case $C\not=\emptyset$: We know that $\ExplorationAlgorithm{\SchemeGeneric}(H_{(A,\emptyset)})=\Tree$. Therefore,~\eqref{eq_conf_gen} implies that $\ExplorationAlgorithm{\Tree}(H_{(A,C)})\not=\Tree$.
\end{proof}
\subsection{The greedy scheme}
\label{sec_schemeGreedy}
The greedy scheme is a reformulation of Penrose's scheme in algorithm~\ref{alg_penrose} as an explorative partition scheme. This reformulation yields a different viewpoint and initiated the development of explorative partition schemes. Let $I$ be a finite and totally ordered set. Let $G:=(I,E)$ be a connected graph. The static information comprises the total order on $I$, the structure of $G$ and the choice of the root $o\in I$.\\

We call this partition scheme \emph{greedy} to emphasise the selection $S_k$ in \StepGreedyBoundary{}. Algorithm~\ref{alg_explorationGreedy} \emph{flood-fills} $H$. It incorporates a parallel version of \emph{Dijkstra's single-source shortest path algorithm}~\cite{Dijkstra__ANoteOnTwoProblemsInConnexionWithGraphs__NM_1959},\cite[page 145]{Korte_Vygen__CombinatorialOptimization_3rd__Springer_2006} on a graph with unit edge weights.

\begin{Alg}[$\SchemePenrose$ exploration]\label{alg_explorationGreedy}
Let $H\in\SpanningSubgraphs{G}$. For every $k$, let $H_k$, $T_k$, $U_k$, $B_k$ and $P_k$ be as in algorithm~\ref{alg_explorationGeneric}. The steps to construct $H_{k+1}$ from $H_k$ on a connected component $C$ of $H_k|_{U_k}$ are:
\begin{description}
 \item[\StepGreedyBoundary{}] \ActionSelect{} $C\cap S_k:=C\cap P_k$.
 \item[\StepGreedyIgnored{}] As $C\cap I_k=\emptyset$, \ActionRemove{} nothing.
 \item[\StepGreedyParent{}]  For each $i\in C\cap S_k$, let $j_i:=\min\Set{j\in B_k: (i,j)\in \Edges{H_k}}$. \ActionSelect{} $(i,j_i)$.
 \item[\StepGreedyUncles{}] For each $i\in C\cap S_k$, \ActionRemove{} all $(i,j)\in\Edges{H_k}$ with $j_i\not=j\in B_k$.
 \item[\StepGreedyCousins{}] \ActionRemove{} all of $\Edges{C\cap S_k}\cap\Edges{H_k}$.
\end{description}
\end{Alg}

\begin{Prop}\label{prop_explorationGreedy}
The function $\ExplorationAlgorithm{\SchemePenrose}$ described by algorithm~\ref{alg_explorationGreedy} is $\SpanningSubgraphs{G}\to\SpanningTrees{G}$ and monotone decreasing, that is $\ForAll H\in\SpanningSubgraphs{G}:\ExplorationAlgorithm{\SchemePenrose}(H)\SubGraph H$.
\end{Prop}

\begin{proof}
Follows from proposition~\ref{prop_explorationGeneric}.
\end{proof}

\begin{Alg}[$\SchemePenrose$ tree edge complement partition]
\label{alg_treeEdgeComplementPartitionPenrose}
Let $\Tree\in\SpanningTrees{G}$. Let $L_k$ be the $k^{th}$ level of $\Tree$. We partition $E\setminus\Edges{\Tree}$ into $\AdmissibleEdgesTree{\SchemePenrose}\uplus\ConflictingEdgesTree{\SchemePenrose}$. Let $0\le k\le l$, $j\in L_k$, $i\in L_l$ and $e:=(i,j)\in E\setminus\Edges{\Tree}$. We let $e\in\ConflictingEdgesTree{\SchemePenrose}$, if
\begin{subequations}\label{eq_conf_pen}
\begin{align}
 &\label{eq_conf_pen_diffTwo}
 l\ge k+2
 \,,\\
 &\label{eq_conf_pen_uncle}
 l=k+1
  \text{ and }
 j<\Parent{i}
 \,,
\end{align}
\end{subequations}
and we let $e\in\AdmissibleEdgesTree{\SchemePenrose}$, if
\begin{subequations}\label{eq_adm_pen}
\begin{align}
 &\label{eq_adm_pen_sameLevel}
 l=k
 \,,\\
 &\label{eq_adm_pen_uncle}
 l=k+1
  \text{ and }
 j>\Parent{i}
 \,.
\end{align}
\end{subequations}
\end{Alg}

\begin{Prop}\label{prop_schemePenrose}
The map
\begin{equation}\label{eq_schemePenrose}
 \SchemePenrose:
 \quad
 \SpanningTrees{G}\to\SpanningSubgraphs{G}
 \quad
 \Tree\mapsto\SchemePenrose(\Tree)
 :=(I,\Edges{\Tree}\uplus\AdmissibleEdgesTree{\SchemePenrose})
\end{equation}
is a partition scheme of $G$ with $[\Tree,\SchemePenrose(\Tree)]=\ExplorationAlgorithm{\SchemePenrose}^{-1}(\Tree)$.
\end{Prop}

\AlsoArxiv{The proof of proposition~\ref{prop_schemePenrose} is in appendix~\ref{sec_schemeGreedyProof}.} We specialise to the case of $G:=\Cluster{\vec{\PolymerOther}}$ being the cluster induced by $\vec{\PolymerOther}\in\PolymerSet^I$. This increases the static information about $G$.

\begin{Prop}[Properties of $\SingletonTrees{\SchemePenrose}(\Cluster{\vec{\PolymerOther}})$]
\label{prop_singletonsPenroseOnCluster}
Let $\Tree\in\SingletonTrees{\SchemePenrose}(\Cluster{\vec{\PolymerOther}})$ and let $C_i$ be the set of children of $i$ in $\Tree$. Then
\begin{subequations}\label{eq_singletonsPenroseOnCluster}
\begin{gather}
 \label{eq_singletonsPenroseOnClusterCompatibilityLevel}
 \ForAll k\in\NatNumZero:\quad
 \Support\vec{\PolymerOther}_{L_k}\text{ is a compatible subset of }\PolymerSet\\
 \label{eq_singletonsPenroseOnClusterCompatibilityChildren}
 \Support\vec{\PolymerOther}_{C_i}
 \text{ is a compatible subset of }\IncompatiblePolymers{\PolymerOther_i}\\
 \label{eq_singletonsPenroseOnClusterCardinality}
 \Cardinality{C_i}
 =\Cardinality{\Support\vec{\PolymerOther}_{C_i}}\,.
\end{gather}
\end{subequations}
\end{Prop}

\begin{proof}
Fix $k$. If $i,j\in L_k$ with $\PolymerOther_i\IncompatibleTo\PolymerOther_j$, then $e:=(i,j)\in E$, as $G=\Cluster{\vec{\PolymerOther}}$ is the cluster of $\vec{\PolymerOther}$. Hence, by~\eqref{eq_adm_pen_sameLevel} $e\in\Edges{\SchemePenrose(\Tree)}$ and $\SchemePenrose(\Tree)\not=\Tree$. This shows~\eqref{eq_singletonsPenroseOnClusterCompatibilityLevel}, which implies~\eqref{eq_singletonsPenroseOnClusterCompatibilityChildren}, which in turn implies~\eqref{eq_singletonsPenroseOnClusterCardinality}.
\end{proof}
\subsection{The returning scheme}
\label{sec_schemeReturning}
The returning scheme $\SchemeReturning$ is an explorative partition scheme  adapted to clusters. It transfers ideas from the Dobrushin-style inductive polymer level proof of the reduced SCUB to the cluster expansion setting. Let $I$ be a finite and totally ordered set. Let $\vec{\PolymerOther}\in\PolymerSet^I$ and $G:=\Cluster{\vec{\PolymerOther}}$. Assume that $G$ is connected. The static information comprises the total order on $I$, the cluster structure of $G$ given by $\vec{\PolymerOther}$ and the choice of the root $o\in I$.\\

The key idea behind the returning scheme is the following insight: The exploration algorithm should select those edges, which we want the singleton trees to contain. Here, we do not want to reuse the last visited different label, which is a similar effect as the \emph{escaping pairs} in the section~\ref{sec_boundingTheEscapedPinnedSeries}. More general, when one looks at the inductive proof of the reduced SCUB in section~\ref{sec_inductivePolymerLevelProofs}, then one wants to transfer the following property: labels are at most used once along each path in the inductive derivation tree (see also~\cite[Algorithm T]{Scott_Sokal__TheRepulsiveLatticeGasTheIndependentSetPolynomialAndTheLovaszLocalLemma__JSP_2005}). If we translate this into the cluster expansion setting, then we want the polymer labels of a path from the root in a singleton tree to form a lazy self-avoiding walk in the polymer system. In other words, a polymer may only label a continuous part of each path from the root in a singleton tree.\\

Let $C_\Polymer$ be the set of vertices labelled with $\Polymer$ in the cluster. The exploration algorithm prefers to select potential nodes with a different polymer label than their parent node in the boundary of the already explored part. This preference creates long paths returning back to $C_\Polymer$, if possible, and gives the partition scheme its name. The structure of a cluster ensures that there is a supergraph of $H$ in $\SpanningSubgraphs{\Cluster{\vec{\PolymerOther}}}$ containing the missing edge to close the path in $C_\Polymer$ (see figure~\ref{fig_anatomy}). Therefore, every tree which contains a path with $\Polymer$ as a label and returns to $C_\Polymer$ with a different polymer label in between is not a singleton tree. The remaining singleton trees contain only paths with labels forming a lazy self-avoiding walk in the polymer system. Proposition~\ref{prop_singletonsReturning} contains the exact properties of the singleton trees of $\SchemeReturning$.

\begin{Alg}[$\SchemeReturning$ exploration]\label{alg_explorationReturning}
Let $H\in\SpanningSubgraphs{G}$. For every $k$, let $H_k$, $T_k$, $U_k$, $B_k$ and $P_k$ be as in algorithm~\ref{alg_explorationGeneric}. The steps to construct $H_{k+1}$ from $H_k$ on a connected component $C$ of $H_k|_{U_k}$ are:\\

Call an edge $(i,j)\in\Edges{C\cap P_k,B_k}\cap\Edges{H}$ \emph{same} (or $\ReturningSame$), if $\PolymerOther_i=\PolymerOther_j$ and \emph{different} (or $\ReturningDifferent$), if $\PolymerOther_i\not=\PolymerOther_j$. Likewise call a vertex $i\in P_k$ same, if all such $(i,j)$ are same and different, if there exists such a non-same $(i,j)$. Finally, we say that a connected component $C$ of $H_k|_{U_k}$ is same, if all vertices in $C\cap P_k$ are same, and different, if $C\cap P_k$ contains at least one different vertex.\\

If $C$ is an $\ReturningSame$ connected component of $H_k|_{U_k}$:
\begin{description}
 \item[\StepRetSameBoundary{}] \ActionSelect{} $C\cap S_k:= C\cap P_k$.
 \item[\StepRetSameIgnored{}] As $C\cap I_k=\emptyset$, \ActionRemove{} nothing.
 \item[\StepRetSameParent{}] For each $i\in C\cap S_k$, let $j_i:=\min\Set{j\in B_k: (i,j)\in \Edges{H_k}}$. \ActionSelect{} $(i,j_i)$.
 \item[\StepRetSameUncles{}] For each $i\in C\cap S_k$, \ActionRemove{} all $(i,j)\in\Edges{H_k}$ with $j_i\not=j\in B_k$.
 \item[\StepRetSameCousins{}] \ActionRemove{} all of $\Edges{C\cap S_k}\cap\Edges{H_k}$.
\end{description}

If $C$ is a $\ReturningDifferent$ connected component of $H_k|_{U_k}$:
\begin{description}
 \item[\StepRetDiffBoundary{}] \ActionSelect{} $C\cap S_k:= \Set{i\in C\cap P_k: i\text{ is }\ReturningDifferent}$.
 \item[\StepRetDiffIgnored{}] We have $C\cap I_k=\Set{i\in C\cap P_k: i\text{ is }\ReturningSame}$. \ActionRemove{} all of $\Edges{B_k,(C\cap I_k)}\cap\Edges{H_k}$.
 \item[\StepRetDiffParent{}] For each $i\in C\cap S_k$, let
 $j_i:=\min\Set{j\in B_k: (i,j)\in \Edges{H_k}\text{ is }\ReturningDifferent}$.
  \ActionSelect{} $(i,j_i)$.
 \item[\StepRetDiffUnclesDifferent{}] For each $i\in C\cap S_k$, \ActionRemove{} every $\ReturningDifferent$ $(i,j)\in\Edges{H_k}$ with $j_i\not=j\in B_k$.
 \item[\StepRetDiffUnclesSame{}] For each $i\in C\cap S_k$, \ActionRemove{} every $\ReturningSame$ $(i,j)\in\Edges{C\cap S_k,B_k}\cap\Edges{H_k}$.
 \item[\StepRetDiffCousins{}] \ActionRemove{} all of $\Edges{C\cap S_k}\cap\Edges{H_k}$.
\end{description}
\end{Alg}

\begin{Prop}\label{prop_explorationReturning}
The function $\ExplorationAlgorithm{\SchemeReturning}$ described by algorithm~\ref{alg_explorationReturning} is $\SpanningSubgraphs{G}\to\SpanningTrees{G}$ and $\ExplorationAlgorithm{\SchemeReturning}|_{\SpanningTrees{G}}=\Identity{\SpanningTrees{G}}$.
\end{Prop}

\begin{proof}
Follows from proposition~\ref{prop_explorationGeneric}.
\end{proof}

\begin{Alg}[$\SchemeReturning$ tree edge complement partition]
\label{alg_treeEdgeComplementPartitionReturning}
Let $\Tree\in\SpanningTrees{G}$. Let $L_k$ be the $k^{th}$ level of $\Tree$. First, we determine if an edge $(\Parent{i},i)$ is a same or non-same edge:
\begin{equation}\label{eq_sameStatus}
 \ReturningStatus:
 \quad I\setminus\Set{o}\to\Set{\ReturningSame,\ReturningDifferent}
 \quad i\mapsto\begin{cases}
   \ReturningSame&\text{if }\PolymerOther_i=\PolymerOther_{\Parent{i}}\\
   \ReturningDifferent&\text{if }\PolymerOther_i\not=\PolymerOther_{\Parent{i}}\,.
  \end{cases}
\end{equation}
For $k\ge 1$, define the \emph{equivalence relation} $\ReturningEquivalentK$ on $L_k$ by
\begin{equation}\label{eq_equivalenceRelationReturning}
 i\ReturningEquivalentK j
 \Iff
 \ReturningStatus(\TreePath{o}{i}\setminus\Set{o})
  =\ReturningStatus(\TreePath{o}{j}\setminus\Set{o})\,,
\end{equation}
where the equality on the rhs is taken in $\Set{\ReturningSame,\ReturningDifferent}^k$ between the labels of the paths $\TreePath{o}{.}$ to the root. This implies that an equivalence class consists of either only same or only non-same nodes and whence we can extend $\ReturningStatus$ to them. For completeness, let $\ReturningEquivalent{0}$ be the trivial equivalence relation on $L_0=\Set{o}$. The equivalence classes possess a \emph{tree structure} consistent with $\Tree$:
\begin{equation}\label{eq_equivalenceClassTreeStructureReturning}
 i\ReturningEquivalentKPlusOne j
 \Then
 \Parent{i}\ReturningEquivalentK\Parent{j}\,,
\end{equation}
that is, equivalent vertices in $L_{k+1}$ have equivalent parents in $L_k$. Therefore, $\ReturningClass{i}{k+1}$ has the \emph{parent class} $\Parent{\ReturningClass{i}{k+1}}:=\ReturningClass{\Parent{i}}{k}$.\\

We partition $E\setminus\Edges{\Tree}$ into $\AdmissibleEdgesTree{\SchemeReturning}\uplus\ConflictingEdgesTree{\SchemeReturning}$. Let $0\le k\le l$, $j\in L_k$, $i\in L_l$ and $e:=(i,j)\in E\setminus\Edges{\Tree}$. We let $e\in\ConflictingEdgesTree{\SchemeReturning}$, if one of the following mutually exclusive conditions holds:
\begin{subequations}\label{eq_conf_ret}
\begin{align}
 &\label{eq_conf_ret_notClassPath}
 l\ge 2
  \text{ and }
 \ReturningClass{j}{k}\not\in\TreePath{\ReturningClass{o}{0}}{\ReturningClass{i}{l}}
 \,,\\
 &\label{eq_conf_ret_classAncestorDifferent}
 l\ge 2
  \text{ and }
 \ReturningClass{j}{k}\in\TreePath{\ReturningClass{o}{0}}{\ReturningClass{\Parent{\Parent{i}}}{l-2}}
  \text{ and }
 \PolymerOther_i\not=\PolymerOther_j
 \,,\\
 &\label{eq_conf_ret_classAncestorSame}
 l\ge 2
  \text{ and }
 \ReturningClass{j}{k}\in\TreePath{\ReturningClass{o}{0}}{\ReturningClass{\Parent{\Parent{i}}}{l-2}}
  \text{ and }
 \PolymerOther_i=\PolymerOther_j
  \text{ and }
 \ReturningStatus(C)=\ReturningSame
 \,,\\
\intertext{
 where $C\in\TreePath{\ReturningClass{j}{k}}{\ReturningClass{\Parent{i}}{l-1}}$ the unique class with $\Parent{C}=\ReturningClass{j}{k}$,
}
 &\label{eq_conf_ret_smallUncleDifferent}
 l\ge 1
  \text{ and }
 \ReturningClass{j}{k}=\ReturningClass{\Parent{i}}{l-1}
  \text{ and }
 \PolymerOther_i\not=\PolymerOther_j
  \text{ and }
 \ReturningStatus(i)=\ReturningDifferent
  \text{ and }
 j<\Parent{i}
 \,,\\
 &\label{eq_conf_ret_differentUncleSame}
 l\ge 1
  \text{ and }
 \ReturningClass{j}{k}=\ReturningClass{\Parent{i}}{l-1}
  \text{ and }
 \PolymerOther_i\not=\PolymerOther_j
  \text{ and }
 \ReturningStatus(i)=\ReturningSame
 \,,\\
 &\label{eq_conf_ret_smallUncleSame}
 l\ge 1
  \text{ and }
 \ReturningClass{j}{k}=\ReturningClass{\Parent{i}}{l-1}
  \text{ and }
 \PolymerOther_i=\PolymerOther_j
  \text{ and }
 \ReturningStatus(i)=\ReturningSame
  \text{ and }
 j<\Parent{i}
 \,.
\end{align}
\end{subequations}
We let $e\in\AdmissibleEdgesTree{\SchemeReturning}$, if one of the following mutually exclusive conditions holds:
\begin{subequations}\label{eq_adm_ret}
\begin{align}
 &\label{eq_adm_ret_equalClass}
 l\ge 2
  \text{ and }
 \ReturningClass{j}{k}=\ReturningClass{i}{l}
 \,,\\
 &\label{eq_adm_ret_classAncestor}
 l\ge 2
  \text{ and }
 \ReturningClass{j}{k}\in\TreePath{\ReturningClass{o}{0}}{\ReturningClass{\Parent{\Parent{i}}}{l-2}}
  \text{ and }
 \PolymerOther_i=\PolymerOther_j
  \text{ and }
 \ReturningStatus(C)=\ReturningDifferent
 \,,\\
\intertext{
 where $C\in\TreePath{\ReturningClass{j}{k}}{\ReturningClass{\Parent{i}}{l-1}}$ the unique class with $\Parent{C}=\ReturningClass{j}{k}$,
}
 &\label{eq_adm_ret_differentUncleDifferent}
 l\ge 1
  \text{ and }
 \ReturningClass{j}{k}=\Parent{\ReturningClass{i}{l}}
  \text{ and }
 \PolymerOther_i\not=\PolymerOther_j
  \text{ and }
 \ReturningStatus(i)=\ReturningDifferent
  \text{ and }
 j>\Parent{i}
 \,,\\
 &\label{eq_adm_ret_sameUncleDifferent}
 l\ge 1
  \text{ and }
 \ReturningClass{j}{k}=\Parent{\ReturningClass{i}{l}}
  \text{ and }
 \PolymerOther_i=\PolymerOther_j
  \text{ and }
 \ReturningStatus(i)=\ReturningDifferent
 \,,\\
 &\label{eq_adm_ret_uncleSame}
 l\ge 1
  \text{ and }
 \ReturningClass{j}{k}=\Parent{\ReturningClass{i}{l}}
  \text{ and }
 \PolymerOther_i=\PolymerOther_j
  \text{ and }
 \ReturningStatus(i)=\ReturningSame
  \text{ and }
 j>\Parent{i}
 \,.
\end{align}
\end{subequations}
\end{Alg}

\begin{Rem}
Of particular importance are the \emph{deep edges} in~\eqref{eq_adm_ret_classAncestor}. This is where we use the particular structure of $\Cluster{\vec{\PolymerOther}}$ (see figure~\ref{fig_anatomy}). Paths in the tree returning to a previously visited polymer find here always an admissible edge to add, thus excluding the tree from $\SingletonTrees{\SchemeReturning}(\Cluster{\vec{\PolymerOther}})$. This is exploited in~\eqref{eq_singletonsReturningAllNonReturning}.

\end{Rem}
\begin{Prop}\label{prop_schemeReturning}
The map
\begin{equation}\label{eq_schemeReturning}
 \SchemeReturning:
 \quad
 \SpanningTrees{G}\to\SpanningSubgraphs{G}
 \quad
 \Tree\mapsto\SchemeReturning(\Tree)
 :=(I,\Edges{\Tree}\uplus\AdmissibleEdgesTree{\SchemeReturning})
\end{equation}
is a partition scheme of $G$ with $[\Tree,\SchemeReturning(\Tree)]=\ExplorationAlgorithm{\SchemeReturning}^{-1}(\Tree)$.
\end{Prop}

\begin{proof}
If we admit that the partition in algorithm~\ref{alg_treeEdgeComplementPartitionReturning} satisfies the compatibility condition~\eqref{eq_treeEdgeComplementPartitionCompatibility}, then proposition~\ref{prop_schemeReturning} is a direct consequence of proposition~\ref{prop_schemeGeneric}. Thus, we show~\eqref{eq_treeEdgeComplementPartitionCompatibility} for $\ExplorationAlgorithm{\SchemeReturning}$ and $\AdmissibleEdgesTree{\SchemeReturning}$.\\

Fix $\Tree\in\SpanningTrees{G}$ and $H\in\ExplorationAlgorithm{\SchemeReturning}^{-1}(\Tree)$. Let $0\le k\le l$, $j\in L_k$, $i\in L_l$ and $e:=(i,j)\in E\setminus\Edges{\Tree}$. Let $(H_n)_{n\in\NatNumZero}$ and $(T_n)_{n\in\NatNumZero}$ be the sequences associated with $H$ from algorithm~\ref{alg_explorationReturning}. For $\emptyset\not=F\subseteq E\setminus\Edges{H}$, set $\tilde{H}:=(I,\Edges{H}\uplus F)$ and let $(\tilde{H}_n)_{n\in\NatNumZero}$ and $(\tilde{T}_n)_{n\in\NatNumZero}$ be its associated sequences from algorithm~\ref{alg_explorationReturning}.\\

The \emph{influence level} $m(e)$ of $e$ is the level of the confluent (last common ancestor) of the equivalence classes of its endpoints in the tree formed by the equivalence classes: $\ReturningClass{j}{k}\Confluent\ReturningClass{i}{l}\subseteq L_{m(e)}$. We claim that, for each $N\in\NatNum$ \AlsoArxiv{(compare with~\eqref{eq_penroseNoInfluence})}:
\begin{equation}\label{eq_returningNoInfluence}
 \left(\ForAll e\in F:\,m(e)\ge N\right)
 \Then
 (H_n|_{T_n})_{n=0}^N = (\tilde{H}_n|_{\tilde{T}_n})_{n=0}^N\,.
\end{equation}
To show~\eqref{eq_returningNoInfluence}, we proceed by induction over $n$, for $0\le n\le N$. By definition $H_0|_{T_0}=\tilde{H}_0|_{\tilde{T}_0}$. For the induction step from $n<N$ to $n+1$, we show that algorithm~\ref{alg_explorationReturning} is not influenced by the presence of such an $e\in F$. The addition of $e$ does not change $\tilde{P}_n$ nor the $\ReturningStatus$ classification of its vertices compared to $P_n$. Let $\mathfrak{i}$ and $\mathfrak{j}$ be the ancestors of $i$ and $j$ at level $n+1$ of $\Tree$ respectively. They are both in $P_n$ and $\tilde{P}_n$. As $\ReturningClass{\mathfrak{i}}{n+1}=\ReturningClass{\mathfrak{j}}{n+1}$ in $\Tree$, we have two possibilities in  $H_n|_{U_n}$: both $\mathfrak{i}$ and $\mathfrak{j}$ are classified $\ReturningSame$ and in an $\ReturningSame$ connected component of $H_n|_{U_n}$ or both $\mathfrak{i}$ and $\mathfrak{j}$ are classified $\ReturningDifferent$ and in a $\ReturningDifferent$ connected component of $H_n|_{U_n}$. In both cases the
presence of $e$ in $\tilde{H}$ could merge the connected components of $i$ and $j$ in $H_n|_{U_n}$ respectively into one connected component of $\tilde{H}_n|_{\tilde{U}_n}$, but only of the same classification. Therefore, all vertices in $\tilde{P}_n=P_n$ end in connected components of $H_n|_{U_n}$ and $\tilde{H}_n|_{\tilde{U}_n}$ of the same classification respectively. Thus, $S_n=\tilde{S_n}$ and $T_n=\tilde{T_n}$. Finally, the selection of the parent in \StepRetSameBoundary{} and \StepRetDiffBoundary{} is independent of $H_n|_{U_n}$ and $\tilde{H}_n|_{\tilde{U}_n}$ in all possible combinations. We conclude that $H_{n+1}|_{T_{n+1}}=\tilde{H}_{n+1}|_{\tilde{T}_{n+1}}$.\\

To show~\eqref{eq_adm_gen}, we assume that $F\subseteq\AdmissibleEdgesTree{\SchemeReturning}$. If $e\in\AdmissibleEdgesTree{\SchemeReturning}$, then by~\eqref{eq_adm_ret} $e$ has influence level $m(e)=k$. We go through all the cases of~\eqref{eq_adm_ret} and show that $e$ is always removed. If $e$ is of type~\eqref{eq_adm_ret_equalClass}, then it is removed by \StepRetSameCousins{} or \StepRetDiffCousins{} during iteration $(k-1)$. If $e$ is of type~\eqref{eq_adm_ret_classAncestor}, then it is removed by \StepRetDiffIgnored{} during iteration $k$. If $e$ is of type~\eqref{eq_adm_ret_differentUncleDifferent}, then it is removed by \StepRetDiffUnclesDifferent{} during iteration $k$. If $e$ is of type~\eqref{eq_adm_ret_uncleSame}, then it is removed by \StepRetSameUncles{} during iteration $k$. If $e$ is of type~\eqref{eq_adm_ret_sameUncleDifferent}, then it is removed by \StepRetDiffUnclesSame{} during iteration $k$.\\

To show~\eqref{eq_conf_gen}, we assume that $F':=F\cap\ConflictingEdgesTree{\SchemePenrose}\not=\emptyset$. Choose $e\in F$ with $m(e)=N:=\min\Set{m(f): f\in F'}$ minimal. We demonstrate that the presence of $e$ causes $(H_n|_{T_n})_{n\in\NatNumZero}$ to diverge from $(\tilde{H}_n|_{\tilde{T}_n})_{n\in\NatNumZero}$ exactly at level $N+1$, that is, \eqref{eq_returningNoInfluence} holds and $H_{N+1}|_{T_{N+1}}\not=\tilde{H}_{N+1}|_{\tilde{T}_{N+1}}$. We go through all the cases of~\eqref{eq_conf_ret}:\\

Case $e$ of type~\eqref{eq_conf_ret_notClassPath}: It is evident that $N < k\lor l$.  Hence, there exist ancestors $\mathfrak{i}$ and $\mathfrak{j}$ of $i$ and $j$ in $L_{N+1}$ respectively with $\ReturningClass{\Parent{\mathfrak{i}}}{N}=\ReturningClass{\Parent{\mathfrak{j}}}{N}=\ReturningClass{j}{k}\Confluent\ReturningClass{i}{l}\subseteq L_N$. If $l=N$ or $k=N$, then $i=\mathfrak{i}$ or $j=\mathfrak{j}$ respectively. It also follows from~\eqref{eq_conf_ret_notClassPath} that $\ReturningClass{\mathfrak{i}}{N+1}\not=\ReturningClass{\mathfrak{j}}{N+1}$. Hence, during iteration $N$ of algorithm~\ref{alg_explorationReturning} and without loss of generality, $\mathfrak{i}$ is classified as $\ReturningSame$ in an $\ReturningSame$ connected component of $H_N|_{U_N}$ and $\mathfrak{j}$ is classified as $\ReturningDifferent$ in a $\ReturningDifferent$ connected component of $H_N|_{U_N}$.
The addition of $e$ in $\tilde{H}$ places $\mathfrak{i}$ and $\mathfrak{j}$ in the same connected component of $\tilde{H}_N|_{\tilde{U}_N}$, via the path $\mathfrak{i}\ConnectedTo i\ConnectedTo j\ConnectedTo\mathfrak{j}$. Hence, $\mathfrak{i}$ is an $\ReturningSame$ vertex in a $\ReturningDifferent$ connected component of $\tilde{H}_N|_{\tilde{U}_N}$ and is not selected into $\tilde{S_N}$ by \StepRetDiffBoundary{}. Thus, $T_{N+1}\not=\tilde{T}_{N+1}$.\\

Case $e$ of type~\eqref{eq_conf_ret_classAncestorDifferent}: Here $N=k$. Let $\mathfrak{i}$ be the ancestor of $i$ with $\Parent{\ReturningClass{\mathfrak{i}}{N+1}}=\ReturningClass{j}{N}$. As $\PolymerOther_i\not=\PolymerOther_j$ the addition of $e$ in $\tilde{H}$ classifies $i$ as $\ReturningDifferent$ in step $N$. Therefore $i\in\tilde{S}_N$ by \StepRetDiffBoundary{}, but $i\not\in T_N$. Thus, $T_{N+1}\not=\tilde{T}_{N+1}$.\\

Case $e$ of type~\eqref{eq_conf_ret_classAncestorSame}: Here $N=k$. Let $\mathfrak{i}$ be the ancestor of $i$ with $\Parent{\ReturningClass{\mathfrak{i}}{N+1}}=\ReturningClass{j}{N}$. As $\PolymerOther_i=\PolymerOther_j$ the addition of $e$ in $\tilde{H}$ classifies $i$ as $\ReturningSame$ in step $N$. As $i\ConnectedTo\mathfrak{i}$ in $H_N|_{U_N}$ and also in $\tilde{H}_N|_{\tilde{U}_N}$, we know that $i$ is in a $\ReturningDifferent$ connected component of $\tilde{H}_N|_{\tilde{U}_N}$, namely the one of $\mathfrak{i}$. Therefore, $i\in\tilde{S}_N$ by \StepRetSameBoundary{}, but $i\not\in T_N$. Thus, $T_{N+1}\not=\tilde{T}_{N+1}$.\\

Case $e$ of type~\eqref{eq_conf_ret_smallUncleDifferent}: Here $N=k$. Let $\mathfrak{i}$ be the parent of $i$ in $\ReturningClass{j}{N}$. The addition of $e$ in $\tilde{H}$ lets \StepRetDiffParent{} select $j$ to be the parent of $i$ in $\tilde{T}_{N+1}|_{\tilde{T}_{N+1}}$. Thus, $T_{N+1}\not=\tilde{T}_{N+1}$.\\

Case $e$ of type~\eqref{eq_conf_ret_differentUncleSame}: Here $N=k$. Let $\mathfrak{i}$ be the parent of $i$ in $\ReturningClass{j}{N}$. The addition of $e$ in $\tilde{H}$ classifies $i$ as $\ReturningDifferent$ during step $N$ instead of $\ReturningSame$. Therefore, its parent in $\tilde{T}_{N+1}|_{\tilde{T}_{N+1}}$ is chosen by \StepRetDiffParent{} instead of \StepRetSameParent{} and is not $\mathfrak{i}$ any more. Thus, $T_{N+1}\not=\tilde{T}_{N+1}$.\\

Case $e$ of type~\eqref{eq_conf_ret_smallUncleSame}: Here $N=k$. Let $\mathfrak{i}$ be the parent of $i$ in $\ReturningClass{j}{N}$. The addition of $e$ in $\tilde{H}$ lets \StepRetSameParent{} select $j$ to be the parent of $i$ in $\tilde{T}_{N+1}|_{\tilde{T}_{N+1}}$. Thus, $T_{N+1}\not=\tilde{T}_{N+1}$.
\end{proof}

\begin{Prop}[Properties of $\SingletonTrees{\SchemeReturning}(\Cluster{\vec{\PolymerOther}})$]
\label{prop_singletonsReturning}
Let $\Tree\in\SingletonTrees{\SchemeReturning}(\Cluster{\vec{\PolymerOther}})$ and let $C_i$ be the set of children of $i$ in $\Tree$. Then
\begin{subequations}\label{eq_singletonsReturning}
\begin{gather}
 \label{eq_singletonsReturningCardinality}
 \Cardinality{C_i}=\Cardinality{\Support\vec{\PolymerOther}_{C_i}}\\
 \label{eq_singletonsReturningCompatibilityChildren}
 (\Support\vec{\PolymerOther}_{C_i})\setminus\Set{\PolymerOther_i}
 \text{ is a compatible subset of }\IncompatibleOtherPolymers{\PolymerOther_i}\\
 \label{eq_singletonsReturningCompatibilityClass}
 \ForAll k\in\NatNumZero, i\in L_k:\quad
 \Support\vec{\PolymerOther}_{\ReturningClass{i}{k}}\text{ is a compatible subset of }\PolymerSet\\
 \label{eq_singletonsReturningAllNonReturning}
 \ForAll i\in I\setminus\Set{o}:\quad
  i\text{ is }\ReturningDifferent
  \Then
  \PolymerOther_i\not\in\Support\vec{\PolymerOther}_{\TreePath{o}{\Parent{i}}}\,.
\end{gather}
\end{subequations}
\end{Prop}

\begin{proof}
Fix $k$ and let $i,j\in L_{k+1}$ with $i\ReturningEquivalentKPlusOne j$. If $\PolymerOther_i\IncompatibleTo\PolymerOther_j$, then $e:=(i,j)\in E$ as $G=\Cluster{\vec{\PolymerOther}}$~\eqref{eq_cluster}. By~\eqref{eq_adm_ret_equalClass}, $e\in\Edges{\SchemeReturning(\Tree)}$ and $\SchemeReturning(\Tree)\not=\Tree$. This shows~\eqref{eq_singletonsReturningCompatibilityClass}, which implies~\eqref{eq_singletonsReturningCompatibilityChildren}, which in turn implies~\eqref{eq_singletonsReturningCardinality}.\\

Suppose there exists a vertex $i$ violating~\eqref{eq_singletonsReturningAllNonReturning}. If $\Parent{i}=o$, then this contradicts the $\ReturningDifferent$ classification of $i$. If $\Parent{i}\not=o$, then there exist $j,\mathfrak{j}\in\TreePath{o}{\Parent{i}}$ with $\GraphMetricOf{\Tree}{o}{j}+1=\GraphMetricOf{\Tree}{o}{\mathfrak{j}}<\GraphMetricOf{\Tree}{o}{i}$ and $\PolymerOther_i=\PolymerOther_j\not=\PolymerOther_\mathfrak{j}$. Take such a $j$ minimal with respect to $\GraphMetricOf{\Tree}{o}{j}$. Consider the edge $e:=(j,i)$, which is in $E$ due to the fact that $G=\Cluster{\vec{\PolymerOther}}$~\eqref{eq_cluster}. It is admissible with respect to $\Tree$ of type~\eqref{eq_adm_ret_classAncestor}. Hence, $e\in\Edges{\SchemeReturning(\Tree)}$ and $\SchemeReturning(\Tree)\not=\Tree$.
\end{proof}
\subsection{The synthetic scheme}
\label{sec_schemeSynthetic}
The synthetic scheme $\SchemeSynthetic$ is an explorative partition scheme which interpolates between the behaviour of the greedy scheme $\SchemePenrose$ in algorithm~\ref{alg_explorationGreedy}, represented by $\BehaviourGreedy$, and the returning scheme $\SchemeReturning$, represented by $\BehaviourReturning$, in algorithm~\ref{alg_explorationReturning} along each edge. The static information comprises the total order on $I$, a behaviour vector $\SyntheticVector\in\BehaviourSetIndexedBy{\EscapePairSet}$, the cluster structure of $\Cluster{\vec{\PolymerOther}}$, with $\vec{\PolymerOther}\in\PolymerSet^I$ and the choice of the root $o\in I$.\\

We only state the exploration algorithm. We omit the edge partition for reconstruction from a tree and the correctness proof, as they mirror the one for the returning scheme in section~\ref{sec_schemeReturning}, with greedy and returning classes taking the role of same and different classes respectively. These greedy and returning classes are encoded by $\SyntheticVector$: the greedy class contains all children with the same polymer label or with $g_{(\PolymerOther_v,\PolymerOther_{\Parent{v}})} = \BehaviourGreedy{}$, whereas the returning class those with $g_{(\PolymerOther_v,\PolymerOther_{\Parent{v}})} = \BehaviourReturning{}$.

\begin{Alg}[$\SchemeSynthetic(\SyntheticVector)$ exploration]\label{alg_explorationSynthetic}
Let $H\in\SpanningSubgraphs{G}$. For every $k$, let $H_k$, $T_k$, $U_k$, $B_k$ and $P_k$ be as in algorithm~\ref{alg_explorationGeneric}. The steps to construct $H_{k+1}$ from $H_k$ on a connected component $C$ of $H_k|_{U_k}$ are:\\

Call an edge $(i,j)\in\Edges{C\cap P_k,B_k}\cap\Edges{H}$ \emph{greedy} (or $\BehaviourGreedy$), if $b_{\PolymerOther_i,\PolymerOther_j}=\BehaviourGreedy$ and \emph{returning} (or $\BehaviourReturning$), if $b_{\PolymerOther_i,\PolymerOther_j}=\BehaviourReturning$. Likewise call a vertex $i\in P_k$ greedy, if all such $(i,j)$ are greedy and returning, if there exists such a non-same $(i,j)$. Finally, we say that a connected component $C$ of $H_k|_{U_k}$ is greedy, if all vertices in $C\cap P_k$ are greedy, and returning, if $C\cap P_k$ contains at least one returning vertex.\\

If $C$ is an $\BehaviourGreedy$ connected component of $H_k|_{U_k}$:
\begin{description}
 \item[\StepSynGreedyBoundary{}] \ActionSelect{} $C\cap S_k:= C\cap P_k$.
 \item[\StepSynGreedyIgnored{}] As $C\cap I_k=\emptyset$ \ActionRemove{} nothing.
 \item[\StepSynGreedyParent{}] For each $i\in C\cap S_k$, let $j_i:=\min\Set{j\in B_k: (i,j)\in \Edges{H_k}}$. \ActionSelect{} $(i,j_i)$.
 \item[\StepSynGreedyUncles{}] For each $i\in C\cap S_k$, \ActionRemove{} all $(i,j)\in\Edges{H_k}$ with $j_i\not=j\in B_k$.
 \item[\StepSynGreedyCousins{}] \ActionRemove{} all of $\Edges{C\cap S_k}\cap\Edges{H_k}$.
\end{description}

If $C$ is a $\BehaviourReturning$ connected component of $H_k|_{U_k}$:
\begin{description}
 \item[\StepSynRetBoundary{}] \ActionSelect{} $C\cap S_k:= \Set{i\in C\cap P_k: i\text{ is }\BehaviourReturning}$.
 \item[\StepSynRetIgnored{}] $C\cap I_k=\Set{i\in C\cap P_k: i\text{ is }\BehaviourGreedy}$. \ActionRemove{} all of $\Edges{B_k,(C\cap I_k)}\cap\Edges{H_k}$.
 \item[\StepSynRetParent{}] For each $i\in C\cap S_k$, let
 $j_i:=\min\Set{j\in B_k: (i,j)\in \Edges{H_k}\text{ is }\BehaviourReturning}$.
  \ActionSelect{} $(i,j_i)$.
 \item[\StepSynRetUnclesReturning{}] For each $i\in C\cap S_k$, \ActionRemove{} every $\BehaviourReturning$ $(i,j)\in\Edges{H_k}$ with $j_i\not=j\in B_k$.
 \item[\StepSynRetUnclesGreedy{}] For each $i\in C\cap S_k$, \ActionRemove{} every $\BehaviourGreedy$ $(i,j)\in\Edges{C\cap S_k,B_k}\cap\Edges{H_k}$.
 \item[\StepSynRetCousins{}] \ActionRemove{} all of $\Edges{C\cap S_k}\cap\Edges{H_k}$.
\end{description}
\end{Alg}
\section{Proof of the SCUBs}
\label{sec_proof}
This section contains the proofs of the SCUBs. The proof of the returning SCUB~\eqref{eq_scub_inhom_ret} is in section~\ref{sec_proof_ret}, where we combine the tree-operator framework with the returning partition scheme. Section~\ref{sec_proof_red} proves the reduced SCUB~\eqref{eq_scub_inhom_red}, by showing that it is a relaxation of the tree-operator underlying the returning SCUB. Section~\ref{sec_proof_synthetic} introduces the synthetic SCUB, building on the synthetic schemes. It is the same as the mixing SCUB~\eqref{eq_scub_mixing}, but in a form more amenable to proof and showing the optimality of the mixing SCUB within natural constraints.
\subsection{The returning SCUB}
\label{sec_proof_ret}
In this section we prove the returning SCUB~\eqref{eq_scub_inhom_ret}. In light of proposition~\ref{prop_multiplexing}, we only have to show a local control à la~\eqref{eq_multiplexing_localControl}, so that the corresponding form~\eqref{eq_multiplexing_supremum} equals $\LeopRet{}$.\\

\begin{subequations}
The returning scheme delivers such a control. Let $\RTVarLabel\in\SingletonTrees{\SchemeReturning}(\Cluster{\vec{\PolymerOther}})$. Property~\eqref{eq_singletonsReturningAllNonReturning} is equivalent to:
\begin{equation*}
 \ForAll i\in I\setminus\Set{o}:
 \NExists j,\mathfrak{j}\in\TreePath{o}{\Parent{i}}:
 \quad
  \GraphMetricOf{\RTVarLabel}{o}{j}
  <\GraphMetricOf{\RTVarLabel}{o}{\mathfrak{j}}
  <\GraphMetricOf{\RTVarLabel}{o}{i}
 \,\land\,
 \PolymerOther_i=\PolymerOther_j\not=\PolymerOther_\mathfrak{j}\,.
\end{equation*}
Therefore the polymer labels along the path $\TreePath{o}{i}$ form a \emph{lazy self-avoiding walk} on $\Support\vec{\PolymerOther}$. By~\eqref{eq_pinnedSeriesEscaping}, we only regard $\vec{\PolymerOther}$ with  $\Support\vec{\PolymerOther}\subseteq\PolymersMinusEscapeSet$. Hence, for $i\in I$, the \emph{avoided polymers} are a $\PolymerSet$-valued sequence starting with $\EscapePolymer$ and adding $\PolymerOther_j$ if we use a $\ReturningDifferent$ edge after $j$ on the path $\TreePath{o}{i}$. The sequence is always non-empty. Denote by $\EscapePolymer_i$ the last polymer in the sequence of vertex $i$. We have $\PolymerOther_i\IncompatibleTo\EscapePolymer_i$ by construction, whence $(\PolymerOther_i,\EscapePolymer_i)_{i\in I}\in\EscapePairSet^I$.\\

We focus on a vertex $i$ and its children in the tree $\RTVarLabel$. Property~\eqref{eq_singletonsReturningCompatibilityChildren} implies that the children's polymer labels form a compatible set. If we have an $\ReturningSame$ $i$, then $(\PolymerOther_i,\EscapePolymer_i)=(\PolymerOther_{\Parent{i}},\EscapePolymer_{\Parent{i}})$, while if we have a $\ReturningDifferent$ $i$, then $\EscapePolymer_i=\PolymerOther_{\Parent{i}}$. The extended labels encode the constraints to apply. Set $I:=\IntRange{0,n}$ with $o:=0$. We drop all other constraints on $\RTVarLabel$ and get
\begin{equation}\label{eq_multiplexEstimate}
  \Iverson{
   \RTVarSize\in
   \SingletonTrees{\SchemeReturning}(\Cluster{\vec{\PolymerOther}})
  }
 \le
  \prod_{i=0}^n
   c_{s_i}((\PolymerOther_i,\EscapePolymer_i),(\PolymerOther_{i_1},\EscapePolymer_{i_1}),\dotsc,(\PolymerOther_{i_{s_i}},\EscapePolymer_{i_{s_i}}))\,,
\end{equation}
where the $(s_i)_{i=0}^n$ denote the number of children of $i$ in $\RTVarSize$ and
\begin{multline}\label{eq_multiplexTerm}
 c_n((\PolymerOther_0,\EscapePolymer_0),\dotsc,(\PolymerOther_n,\EscapePolymer_n))
 \\:=
 \sum_{A\subseteq\IntRange{n}}
 \left(
 \prod_{i\in A}
  \Iverson{(\PolymerOther_i,\EscapePolymer_i)=(\PolymerOther_0,\EscapePolymer_0)}
 \prod_{i\not=j\in A}
  \Iverson{\PolymerOther_j\CompatibleTo\PolymerOther_i}
 \right)
 \qquad\qquad\qquad
 \\\times\left(
 \prod_{i\in\IntRange{n}\setminus A}
  \Iverson{\PolymerOther_i\in\IncompatibleOtherPolymers{\PolymerOther_0}\setminus\Set{\EscapePolymer_0},\EscapePolymer_i=\PolymerOther_0}
 \prod_{i\not=j\in\IntRange{n}\setminus A}
  \Iverson{\PolymerOther_j\CompatibleTo \PolymerOther_i}
 \right)\,.
\end{multline}
The $(c_n)_{n\in\NatNumZero}$ are the star-invariant functions needed in~\eqref{eq_multiplexing_localControl}.\\

Let $Y:=[0,\infty]^{\EscapePairSet}$. Denote by $\EscapeInjection$ the injection from $X$ into $Y$, \emph{multiplexing} values by ignoring the escape coordinate in $\EscapePairSet$. Define the operator $\LeopRetMulti{}: Y\to Y$ by
\begin{equation*}
 \LeopRetMulti{\EscapePair}(\vec{u})
 :=
 \sum_{n\ge 0}
  \frac{1}{\Factorial{n}}
 \sum_{(\vec{\PolymerOther},\vec{\EscapePolymer})\in\Set{\EscapePair}\times\EscapePairSet^n}
  c_n(\vec{\PolymerOther},\vec{\EscapePolymer})
 \prod_{i=1}^n u_{(\PolymerOther_i,\EscapePolymer_i)}\,.
\end{equation*}
We see that $c_n(\vec{\PolymerOther},\vec{\EscapePolymer})$ is the coefficient of $\prod_{i=1}^n u_{(\PolymerOther_i,\EscapePolymer_i)}$ in the product of the following terms:
\begin{multline*}
 (1+u_{\EscapePair})\\
 =
 \sum_{n\ge 0}
  \frac{1}{\Factorial{n}}
 \sum_{(\vec{\PolymerOther},\vec{\EscapePolymer})\in\Set{\EscapePair}\times\EscapePairSet^n}
  \prod_{i=1}^n
   \Iverson{(\PolymerOther_i,\EscapePolymer_i)=\EscapePair}
  \prod_{i\not=j=1}^n
   \Iverson{\PolymerOther_j\CompatibleTo\PolymerOther_i}
 \prod_{i=1}^n u_{(\PolymerOther_i,\EscapePolymer_i)}
\end{multline*}
and
\begin{multline*}
 \PartitionFunction{\IncompatibleOtherPolymers{\Polymer}\setminus\Set{\EscapePolymer}}(\vec{u}_{\Set{(\PolymerOther,\Polymer):\, \PolymerOther\in\IncompatibleOtherPolymers{\Polymer}\setminus\Set{\EscapePolymer}}})\\
 =
 \sum_{n\ge 0}
  \frac{1}{\Factorial{n}}
 \sum_{(\vec{\PolymerOther},\vec{\EscapePolymer})\in\Set{\EscapePair}\times\EscapePairSet^n}
  \prod_{i=1}^n
   \Iverson{\PolymerOther_i\in\IncompatibleOtherPolymers{\Polymer}\setminus\Set{\EscapePolymer},\EscapePolymer_i=\Polymer}
  \prod_{i\not=j=1}^n
   \Iverson{\PolymerOther_j\CompatibleTo\PolymerOther_i}
 \prod_{i=1}^n u_{(\PolymerOther_i,\EscapePolymer_i)}\,.
\end{multline*}
Therefore,
\begin{equation*}
 \LeopRetMulti{\EscapePair}(\vec{u})
 =(1+u_{\EscapePair})\,
 \PartitionFunction%
  {\IncompatibleOtherPolymers{\Polymer}\setminus\Set{\EscapePolymer}}
  (\vec{u}_{\Set{(\PolymerOther,\Polymer):\, \PolymerOther\in
   \IncompatibleOtherPolymers{\Polymer}\setminus\Set{\EscapePolymer}}}
  )\,.
\end{equation*}

Taking the supremum over $\EscapePolymer\in\IncompatibleOtherPolymers{\Polymer}$ as in~\eqref{eq_multiplexing_supremum} yields $\LeopRet{}$~\eqref{eq_scub_inhom_ret}.
\end{subequations}
\subsection{The reduced SCUB}
\label{sec_proof_red}
It is evident that the returning SCUB~\eqref{eq_scub_inhom_ret} is stronger than the reduced SCUB~\eqref{eq_scub_inhom_red}, by the inequality $\LeopRet{}(\vec{\mu})\le\LeopRed{}(\vec{\mu})$. Although one can establish~\eqref{eq_scub_inhom_red} without using cluster expansion by an inductive proof based on the fundamental identity~\eqref{eq_fi}, we want to show that $\LeopRed{}(\vec{\mu})$ comes from a tree-operator, too. If we relax~\eqref{eq_multiplexTerm} to
\begin{multline*}
 c_n((\PolymerOther_0,\EscapePolymer_0),\dotsc,(\PolymerOther_n,\EscapePolymer_n))
 \\\le
 \sum_{A\subseteq\IntRange{n}}
 \left(
 \prod_{i\in A}
  \Iverson{(\PolymerOther_i,\EscapePolymer_i)=(\PolymerOther_0,\EscapePolymer_0)}
 \right)\left(
 \prod_{i\in\IntRange{n}\setminus A}
  \Iverson{\PolymerOther_i\in\IncompatibleOtherPolymers{\PolymerOther_0}\setminus\Set{\EscapePolymer_0},\EscapePolymer_i=\PolymerOther_0}
 \right)\,,
\end{multline*}
then we see that $c_n(\vec{\PolymerOther},\vec{\EscapePolymer})$ is the coefficient of $\prod_{i=1}^n u_{(\PolymerOther_i,\EscapePolymer_i)}$ in the product of the terms $(1+u_{\EscapePair})$ and $\prod_{\PolymerOther\in\IncompatibleOtherPolymers{\Polymer}\setminus\Set{\EscapePolymer}}(1+u_{(\PolymerOther,\Polymer)})$. Thus, we have
\begin{equation*}
 \LeopRetMulti{\EscapePair}(\vec{u})
 \le(1+u_{\EscapePair})
 \prod_{
  \PolymerOther\in
  \IncompatibleOtherPolymers{\Polymer}\setminus\Set{\EscapePolymer}
  }
  (1+u_{(\PolymerOther,\Polymer)})\,.
\end{equation*}

Taking the supremum over $\EscapePolymer\in\IncompatibleOtherPolymers{\Polymer}$ as in~\eqref{eq_multiplexing_supremum} yields $\LeopRed{}$~\eqref{eq_scub_inhom_red}.

\subsection{The synthetic SCUB}
\label{sec_proof_synthetic}

This section presents the synthetic SCUB. It builds upon the synthetic scheme, whence it interpolates between the FP SCUB and the returning SCUB. The proof of the synthetic SCUB follows the same reasoning as the proof of the returning SCUB~\eqref{eq_scub_inhom_ret} in section~\ref{sec_proof_ret}, whence we omit it. Section~\ref{sec_proof_mixing} shows that the synthetic SCUB equals the much simpler mixing SCUB, whence the synthetic SCUB has no use outside the proofs.\\

We interpolate along each incompatible polymer pair between the greedy behaviour of $\LeopFP{}$, marked by $\BehaviourGreedy$, and the returning behaviour of $\LeopRet{}$, marked by $\BehaviourReturning$. For $\SyntheticVector\in\BehaviourSetIndexedBy{\EscapePairSet}$ and $\Polymer\in\PolymerSet$, let the \emph{returning incoming/outgoing incompatibles} be
\begin{equation}\label{eq_behaviourIncompatibles}
 \IncompatiblesBehaviourReturningIn{\SyntheticVector}{\Polymer}
 :=\Set{
  \EscapePolymer\in\IncompatibleOtherPolymers{\Polymer}:
  \,
  g_{(\EscapePolymer,\Polymer)}=\BehaviourReturning}
 \,\,\,\text{and}\,\,
 \IncompatiblesBehaviourReturningOut{\SyntheticVector}{\Polymer}
 :=\Set{
  \EscapePolymer\in\IncompatibleOtherPolymers{\Polymer}:
  \,
  g_{(\Polymer,\EscapePolymer)}=\BehaviourReturning}\,.
\end{equation}
For a fixed choice of behaviour $\SyntheticVector\in\BehaviourSetIndexedBy{\EscapePairSet}$, we get a SCUB of shape
\begin{multline}\label{eq_leop_synthetic}
 \LeopSynthetic{\SyntheticVector}{\Polymer}(\vec{\mu})
 :=\max\Bigl\{
 \Iverson{\IncompatiblesBehaviourReturningOut{\SyntheticVector}{\Polymer}\not=\IncompatibleOtherPolymers{\Polymer}}
 \PartitionFunction{\IncompatiblesBehaviourReturningOut{\SyntheticVector}{\Polymer}}
 (\vec{\mu})
 \PartitionFunction{
  \IncompatiblePolymers{\Polymer}
  \setminus\IncompatiblesBehaviourReturningOut{\SyntheticVector}{\Polymer}
 }
 (\vec{\mu})\,,
 \\
 \Iverson{\IncompatiblesBehaviourReturningOut{\SyntheticVector}{\Polymer}\not=\emptyset}
 (1+\mu_\Polymer)\sup_{\EscapePolymer\in\IncompatiblesBehaviourReturningIn{\SyntheticVector}{\Polymer}}
  \PartitionFunction{
   \IncompatiblesBehaviourReturningOut{\SyntheticVector}{\Polymer}
   \setminus\Set{\EscapePolymer}
  }
  (\vec{\mu})
  \PartitionFunction{
   \IncompatibleOtherPolymers{\Polymer}
   \setminus(\IncompatiblesBehaviourReturningOut{\SyntheticVector}{\Polymer}
   \cup\Set{\EscapePolymer})
  }(\vec{\mu})
 \Bigr\}\,.
\end{multline}
The \emph{synthetic} SCUB is the supremum over all choices of $\SyntheticVector$, with shape:
\begin{equation}\label{eq_scub_synthetic}
 \Bigl(
  \Exists\SyntheticVector\in\BehaviourSetIndexedBy{\EscapePairSet},
  \Exists\vec{\mu}\in[0,\infty[^\EscapePairSet:
  \quad
  \vec{\rho}\,\LeopSynthetic{\SyntheticVector}{}(\vec{\mu})
  <\vec{\mu}
 \Bigr)
 \Then\text{\eqref{eq_onePolymerPartitionRatioLimit} holds.}
\end{equation}

\subsection{The mixing SCUB}
\label{sec_proof_mixing}

We show that the mixing SCUB equals the synthetic SCUB from section~\ref{sec_proof_synthetic}. The interpolation between the greedy and returning behaviours in the synthetic scheme and SCUB is too fine. The optimal case in the synthetic SCUB is when the local outgoing behaviour is either only greedy or only returning. This proves the mixing SCUB via the tree-operator interpretation of the synthetic SCUB. A direct tree-operator interpretation of the mixing SCUB follows from a specialisation of the tree-operator underlying the synthetic SCUB. There seems to be no simpler partition scheme than the synthetic scheme, though, as the interactions between greedy and returning behaviour must always be taken into account.\\

This section's arguments suggest that the mixing SCUB is optimal among SCUBs building on tree-operators indexed by $\EscapePairSet$ and using only the general properties of a polymer system. In fact, $\EscapePair$ tells us that we are at a vertex labelled with $\Polymer$, with $\EscapePolymer$ as additional information. Using this extra information to remove $\EscapePolymer$ for consideration as a child's label seems the only possible thing to do to reduce the number of admissible offspring labellings.

\begin{Prop}\label{prop_synthetic_improvement}
Let $\mu\in[0,\infty[^\PolymerSet$ and $\SyntheticVector\in\BehaviourSetIndexedBy{\EscapePairSet}$. Choose $\Polymer\in\PolymerSet$ and suppose that $\IncompatiblesBehaviourReturningOut{\SyntheticVector}{\Polymer}\not=\IncompatibleOtherPolymers{\Polymer}$. Let $\SyntheticVector'$ equal $\SyntheticVector$ except all the values indexed by $\Set{(\Polymer,\EscapePolymer): \EscapePolymer\in\IncompatibleOtherPolymers{\Polymer}}$ changed to $\BehaviourGreedy$. We have
\begin{equation}\label{eq_synthetic_improvement}
 \LeopSynthetic{\SyntheticVector}{}(\vec{\mu})
 \ge
 \LeopSynthetic{\SyntheticVector'}{}(\vec{\mu})\,.
\end{equation}
\end{Prop}

\begin{Lem}[{\cite[(8.21)]{Scott_Sokal__TheRepulsiveLatticeGasTheIndependentSetPolynomialAndTheLovaszLocalLemma__JSP_2005}}]
\label{lem_submultiplicativity}
The function $\PartitionFunction{}$ is submultiplicative in space for non-negative real fugacities:
\begin{equation}\label{eq_submultiplicativity}
 \forall\FiniteVolume_1,\FiniteVolume_2\FiniteSubsetOf\PolymerSet,
 \forall\vec{\mu}\in[0,\infty[^\EscapePairSet:
 \quad
 \PartitionFunction{\FiniteVolume_1}(\vec{\mu})
 \PartitionFunction{\FiniteVolume_2}(\vec{\mu})
 \ge\PartitionFunction{\FiniteVolume_1\uplus\FiniteVolume_2}(\vec{\mu})\,,
\end{equation}
with equality, iff $\FiniteVolume_i=\emptyset$ or $\vec{\mu}_{\FiniteVolume_i}=\vec{0}$, for $i\in\Set{1,2}$.
\end{Lem}

\begin{proof}[Proof of proposition~\ref{prop_synthetic_improvement}.]
A change between $\LeopSynthetic{\SyntheticVector}{}(\vec{\mu})$ and $\LeopSynthetic{\SyntheticVector'}{}(\vec{\mu})$ may only occur on $\IncompatiblePolymers{\Polymer}$. If $\IncompatiblesBehaviourReturningOut{\SyntheticVector}{\Polymer}=\IncompatibleOtherPolymers{\Polymer}$, then $\SyntheticVector=\SyntheticVector'$ and nothing changes. Assume that $\IncompatiblesBehaviourReturningOut{\SyntheticVector}{\Polymer}\not=\IncompatibleOtherPolymers{\Polymer}$. In the definition of the synthetic SCUB~\eqref{eq_leop_synthetic}, let $A(\SyntheticVector,\Polymer,\vec{\mu})$ be the lefthand side of the maximum and let $B(\SyntheticVector,\Polymer,\EscapePolymer,\vec{\mu})$ denote the term in the supremum. There are two cases:\\

Case $\Polymer$: We have
$\IncompatiblesBehaviourReturningOut{\SyntheticVector}{\Polymer}\supseteq\IncompatiblesBehaviourReturningOut{\SyntheticVector'}{\Polymer}=\emptyset$ and $\IncompatiblesBehaviourReturningIn{\SyntheticVector}{\Polymer}=\IncompatiblesBehaviourReturningIn{\SyntheticVector'}{\Polymer}$. Submultiplicaticity~\eqref{eq_submultiplicativity} implies that
\begin{equation*}
 A(\SyntheticVector,\Polymer,\vec{\mu})
 \ge A(\SyntheticVector',\Polymer,\vec{\mu})
 \quad\text{ and }\quad
 B(\SyntheticVector,\Polymer,\EscapePolymer,\vec{\mu})
 \ge B(\SyntheticVector',\Polymer,\EscapePolymer,\vec{\mu})\,.
\end{equation*}
Thus, $\LeopSynthetic{\SyntheticVector}{\Polymer}(\vec{\mu})\ge\LeopSynthetic{\SyntheticVector'}{\Polymer}(\vec{\mu})$.\\

Case $\PolymerOther\in\IncompatibleOtherPolymers{\Polymer}$: We have $\IncompatiblesBehaviourReturningOut{\SyntheticVector}{\Polymer}=\IncompatiblesBehaviourReturningOut{\SyntheticVector'}{\Polymer}$ and $\IncompatiblesBehaviourReturningIn{\SyntheticVector}{\PolymerOther}\supseteq\IncompatiblesBehaviourReturningIn{\SyntheticVector'}{\PolymerOther}$. Hence, only the supremum over the $B(\SyntheticVector,\PolymerOther,\EscapePolymer,\vec{\mu})$ can decrease and $\LeopSynthetic{\SyntheticVector}{\PolymerOther}(\vec{\mu})\ge\LeopSynthetic{\SyntheticVector'}{\PolymerOther}(\vec{\mu})$.
\end{proof}

\begin{Prop}\label{prop_synthetic_equals_mixing}
The synthetic and the mixing SCUB are the same.
\end{Prop}

\begin{proof}
Let $\EscapeInjection: \BehaviourSetIndexedBy{\PolymerSet}\to\BehaviourSetIndexedBy{\EscapePairSet}$, with $\EscapeInjection(\MixingVector)_{\EscapePair}:=\MixingVector_\Polymer$. If $\vec{\rho}$ fulfils the mixing SCUB, then it fulfils the synthetic SCUB: $\LeopSynthetic{\EscapeInjection(\MixingVector)}{} = \LeopMixing{\MixingVector}{}$.\\

Let $u:\BehaviourSetIndexedBy{\EscapePairSet}\to\BehaviourSetIndexedBy{\PolymerSet}$, with $u(\SyntheticVector)_\Polymer = \BehaviourGreedy$, if $\IncompatiblesBehaviourReturningOut{\SyntheticVector}{\Polymer}\not=\IncompatibleOtherPolymers{\Polymer}$ and $u(\SyntheticVector)_\Polymer = \BehaviourReturning$, if $\IncompatiblesBehaviourReturningOut{\SyntheticVector}{\Polymer}=\IncompatibleOtherPolymers{\Polymer}$. Proposition~\ref{prop_synthetic_improvement} and induction over $\PolymerSet$ shows that $\LeopSynthetic{\SyntheticVector}{} \ge \LeopMixing{u(\SyntheticVector)}{}$.\\

Thus, the optimal interpolation between greedy and returning behaviour in the synthetic scheme only happens on $\EscapeInjection(\BehaviourSetIndexedBy{\PolymerSet})$.
\end{proof}
\section*{Acknowledgements}
\label{sec_acknowledgements}
I am grateful to Roberto Fern\'{a}ndez for the time he took to explain me his work and encourage my attempts at extending it. The visits at the University of Utrecht have been supported by the RGLIS short visit grants 4076 and 4446 from the European Science Foundation (ESF). This work has been also supported by the Austrian Science Fund (FWF), project W1230-N13. I also want to thank Maria Eichlseder for help with the graphics. I am thankful of the comments of my referees, pushing me towards an improved exposition of the motivation and high-level overviews.

\section{Additional material}
\label{sec_additional}
This section contains proofs and statements not included in the official paper.

\subsection{Graphical explanations II}
\label{sec_graphical_explanations_II}

See figure~\ref{fig_explanation}.

\begin{figure}
\subfloat[][The polymer $\Polymer$ and the graph of its incompatible neighbourhood $(\IncompatiblePolymers{\Polymer},\IncompatibleTo)$.]{
\label{fig_explanation_incompatible}
\begin{tikzpicture}[scale=1.8]
  \node[vertex] (gam) at (0,0) {$\gamma$};
  \node[vertex] (x1) at (1,1) {$\xi_1$};
  \node[vertex] (x2) at (1,1/3) {$\xi_2$};
  \node[vertex] (x3) at (1,-1/3) {$\xi_3$};
  \node[vertex] (eps) at (1,-1) {$\varepsilon$};

  \draw (x1) -- (x2) -- (gam) -- (eps) -- (x3) -- (gam) -- (x1);

  \selfedge{gam}{180};
  \foreach \x in {x1, x2, x3, eps} {\selfedge{\x}{0};}
\end{tikzpicture}
}
\subfloat[][Constraints in $\displaystyle\LeopKP{\Polymer}(\vec{\mu})$.]{
\label{fig_explanation_kp}
\begin{tikzpicture}[scale=1.8]
  \node[vertex] (gam) at (0,0) {$\gamma$};
  \node[vertex] (x1) at (1,1) {$\xi_1$};
  \node[vertex] (x2) at (1,1/3) {$\xi_2$};
  \node[vertex] (x3) at (1,-1/3) {$\xi_3$};
  \node[vertex] (eps) at (1,-1) {$\varepsilon$};

  \draw (x1) -- (x2) -- (gam) -- (eps) -- (x3) -- (gam) -- (x1);

  \selfedge{gam}{180};
  \foreach \x in {x1, x2, x3, eps} {\selfedge{\x}{0};}
\end{tikzpicture}
}
\subfloat[][Constraints in $\displaystyle\LeopDob{\Polymer}(\vec{\mu})$.]{
\label{fig_explanation_dob}
\begin{tikzpicture}[scale=1.8]
  \node[vertex] (gam) at (0,0) {$\gamma$};
  \node[vertex] (x1) at (1,1) {$\xi_1$};
  \node[vertex] (x2) at (1,1/3) {$\xi_2$};
  \node[vertex] (x3) at (1,-1/3) {$\xi_3$};
  \node[vertex] (eps) at (1,-1) {$\varepsilon$};

  \draw (x1) -- (x2) -- (gam) -- (eps) -- (x3) -- (gam) -- (x1);

  \selfedgemark{gam}{180};
  \foreach \x in {x1, x2, x3, eps} {\selfedgemark{\x}{0};}
\end{tikzpicture}
}
\\
\subfloat[][Constraints in $\displaystyle\LeopFP{\Polymer}(\vec{\mu})$.]{
\label{fig_explanation_fp}
\begin{tikzpicture}[scale=1.8]
  \node[vertex] (gam) at (0,0) {$\gamma$};
  \node[vertex] (x1) at (1,1) {$\xi_1$};
  \node[vertex] (x2) at (1,1/3) {$\xi_2$};
  \node[vertex] (x3) at (1,-1/3) {$\xi_3$};
  \node[vertex] (eps) at (1,-1) {$\varepsilon$};

  \draw[marked] (x1) -- (x2) -- (gam) -- (eps) -- (x3) -- (gam) -- (x1);

  \selfedgemark{gam}{180};
  \foreach \x in {x1, x2, x3, eps} {\selfedgemark{\x}{0};}
\end{tikzpicture}
}
\subfloat[][Constraints in $\displaystyle\LeopRed{\Polymer}(\vec{\mu})$.]{
\label{fig_explanation_red}
\begin{tikzpicture}[scale=1.8]
  \draw[marked] (1,-1) +(-.15,-.15) -- +(.15,.15)
                      +(-.15,.15) -- +(.15,-.15);
  \draw (1,-1) node [circle, fill=white, inner sep=2.5pt] {};

  \node[vertex] (gam) at (0,0) {$\gamma$};
  \node[vertex] (x1) at (1,1) {$\xi_1$};
  \node[vertex] (x2) at (1,1/3) {$\xi_2$};
  \node[vertex] (x3) at (1,-1/3) {$\xi_3$};
  \node[vertex] (eps) at (1,-1) {$\varepsilon$};
  
  \draw (gam) -- (x1);
  \draw (gam) -- (x2);
  \draw (gam) -- (x3);
  
  \draw (x1) -- (x2);
  \draw[marked] (gam) -- (eps) -- (x3);

  \selfedgemark{gam}{180};
  \foreach \x in {x1, x2, x3, eps} {\selfedgemark{\x}{0};}
\end{tikzpicture}
}
\subfloat[][Constraints in $\displaystyle\LeopRet{\Polymer}(\vec{\mu})$.]{
\label{fig_explanation_ret}
\begin{tikzpicture}[scale=1.8]
  \draw[marked] (1,-1) +(-.15,-.15) -- +(.15,.15)
                      +(-.15,.15) -- +(.15,-.15);
  \draw (1,-1) node [circle, fill=white, inner sep=2.5pt] {};
  
  \node[vertex] (gam) at (0,0) {$\gamma$};
  \node[vertex] (x1) at (1,1) {$\xi_1$};
  \node[vertex] (x2) at (1,1/3) {$\xi_2$};
  \node[vertex] (x3) at (1,-1/3) {$\xi_3$};
  \node[vertex] (eps) at (1,-1) {$\varepsilon$};

  \draw (gam) -- (x1);
  \draw (gam) -- (x2);
  \draw (gam) -- (x3);

  \draw[marked] (x1) -- (x2);
  \draw[marked] (gam) -- (eps) -- (x3);

  \selfedgemark{gam}{180};
  \foreach \x in {x1, x2, x3, eps} {\selfedgemark{\x}{0};}
\end{tikzpicture}
}
\caption[Constraints in the SCUBs]{(Colour online) Constraints in the different SCUBs at a leaf labelled with $\Polymer$. We assume the incompatible neighbourhood of $\Polymer$ to look as in~\subref{fig_explanation_incompatible}. The incompatibilities taken into account by the SCUB are marked by dotted red lines. In figures~\subref{fig_explanation_red} and~\subref{fig_explanation_ret} the polymer $\EscapePolymer$ is the polymer forbidden at the current step. The SCUB sums over all weighted depth $1$ trees rooted at $\Polymer$, which do not contain two vertices labelled with polymers with a dotted red line in the corresponding diagram. More dotted red lines imply less trees, a more relaxed SCUB and bigger admissible $\vec{\rho}$ .}
\label{fig_explanation}
\end{figure}
\subsection{Proof of Penrose's theorem}
\label{sec_penroseProof}
\begin{proof}[Proof of theorem~\ref{thm_penroseIdentity}]
Let $(x_e)_{e\in E}\in\ComplexNum^E$. Then
\begin{align*}
 \sum_{H\in\SpanningSubgraphs{G}} \prod_{e\in\Edges{H}} x_e
 &={} \sum_{\Tree\in\SpanningTrees{G}} \prod_{e\in\Edges{\Tree}} x_e
  \sum_{F\subseteq\Edges{\SchemeAnonymous(\Tree)}\setminus\Edges{\Tree}}
  \prod_{f\in F} x_f\\
 &={} \sum_{\Tree\in\SpanningTrees{G}} \prod_{e\in\Edges{\Tree}} x_e
  \prod_{f\in\Edges{\SchemeAnonymous(\Tree)}\setminus\Edges{\Tree}} (1 + x_f)\,.
\end{align*}
If we choose $\vec{x}=-\vec{1}$, then in the all the contributions from trees with $\Tree\not=\SchemeAnonymous(\Tree)$ cancel and, for every $\Tree\in\SingletonTrees{\SchemeAnonymous}(G)$, the contribution is $(-1)^{\Cardinality{V}-1}$.
\end{proof}
\subsection{On the impossibility of globally excluding a fixed neighbour}
\label{sec_globallyExcludedNeighbour}
The holy grail would be a scheme excluding a globally fixed neighbour of each vertex. This is impossible, though. The counterexample is a polymer system $(\PolymerSet,\IncompatibleTo)$ isomorph (ignoring loops) to a large circle of size $N$. Take the cluster formed by $\vec{\PolymerOther}\in\PolymerSet^N$ with $\Support\vec{\PolymerOther}=\PolymerSet$, which is again isomorph to the circle of size $N$. It has $N$ spanning trees and only one of them can be expanded to the full circle. This means, that every partition scheme has $N-1$ singleton trees on $\Cluster{\vec{\PolymerOther}}$. Fix the root and the globally forbidden neighbours. Then every finite family approximation not returning to forbidden neighbours contains at most one of those $N-1$ trees. Hence no partition scheme with globally excluded neighbours exists.
\subsection{A fixpoint theorem}
\label{sec_fixPoint}
In this section we present an adaption of a well known fixpoint theorem on lattices by Tarski~\cite{Tarski__ALatticeTheoreticalFixpointTheoremAndItsApplications__PacJM_1955}. For some countable set of labels $\LabelSpace$, let $X := [0,\infty]^\LabelSpace$ be the \emph{lattice} with \emph{partial order} $\vec{x}\le\vec{y}$, that is coordinate-wise comparison of the vectors: $\ForAll l\in\LabelSpace: x_l\le y_l$.\\

The \emph{supremum} and \emph{infimum} of a subset $A$ of $X$ are defined by $\sup A:=(\sup a_l: \vec{a}\in A)_{l\in\LabelSpace}$ and $\inf A:=(\inf a_l: \vec{a}\in A)_{l\in\LabelSpace}$ respectively.

We say that a function $\phi_X\to X$ \emph{preserves the order}, if
\begin{equation}\label{eq_preserveOrder}
 \ForAll \vec{x},\vec{y}\in X:\qquad
 \vec{x}\le\vec{y}\Then \phi(\vec{x})\le\phi(\vec{y})\,.
\end{equation}
The set of \emph{decreasing points} of $\phi$ is
\begin{equation}\label{eq_phi_decreasing}
 D_\phi:=\Set{\vec{y}\in X:\phi(\vec{y})\le\vec{y})}\,.
\end{equation}
A sequence $(\vec{y}^{(n)})_{n\in\NatNum}$ of elements of $X$ is said to be \emph{non-decreasing} and \emph{non-increasing}, if $\ForAll n\in\NatNum: \vec{y}^{(n)}\le\vec{y}^{(n+1)}$ and $\vec{y}^{(n)}\ge\vec{y}^{(n+1)}$ respectively.

\begin{Prop}[after {\cite[proposition 8]{Fernandez_Procacci__ClusterExpansionForAbstractPolymerModels_NewBoundsFromAnOldApproach__CMP_2007}}]
\label{prop_treeConditionLatticeLimit}
Let $\phi: X\to X$ be order-pre\-serving. If there is a $\vec{\mu}\in D_\phi$, then $(\phi^n(\vec{0}))_{n\in\NatNumZero}$ is non-decreasing with limit $\vec{\rho}^\star$, $(\phi^n(\vec{\mu}))_{n\in\NatNumZero}$ is non-increasing with limit $\vec{\mu}^\star$ and we have
\begin{equation}\label{eq_treeConditionLatticeLimit}
 \phi(\vec{\rho}^\star)
 =\vec{\rho}^\star
 \le\vec{\mu}^\star
 = \phi(\vec{\mu}^\star)\,.
\end{equation}
In particular, $\vec{\rho}^\star\in[0,\infty[^\LabelSpace$, iff $D_\phi\cap[0,\infty[^\LabelSpace\not=\emptyset$.
\end{Prop}
\subsection{Proof of the greedy scheme}
\label{sec_schemeGreedyProof}
\begin{proof}[Proof of proposition~\ref{prop_schemePenrose}]
If we admit that the partition in algorithm~\ref{alg_treeEdgeComplementPartitionPenrose} satisfies the compatibility condition~\eqref{eq_treeEdgeComplementPartitionCompatibility}, then proposition~\ref{prop_schemePenrose} is a direct consequence of proposition~\ref{prop_schemeGeneric}. Thus, we show~\eqref{eq_treeEdgeComplementPartitionCompatibility} for $\ExplorationAlgorithm{\SchemePenrose}$ and $\AdmissibleEdgesTree{\SchemePenrose}$.\\

Fix $\Tree\in\SpanningTrees{G}$ and $H\in\ExplorationAlgorithm{\SchemeReturning}^{-1}(\Tree)$. Let $0\le k\le l$, $j\in L_k$, $i\in L_l$ and $e:=(i,j)\in E\setminus\Edges{\Tree}$. Let $(H_n)_{n\in\NatNumZero}$ and $(T_n)_{n\in\NatNumZero}$ be the  sequences associated with $H$ from algorithm~\ref{alg_explorationReturning}. For $\emptyset\not=F\subseteq E\setminus\Edges{H}$ set $\tilde{H}:=(I,\Edges{H}\uplus F)$ and let $(\tilde{H}_n)_{n\in\NatNumZero}$ and $(\tilde{T}_n)_{n\in\NatNumZero}$ be its associated sequences from algorithm~\ref{alg_explorationReturning}.\\

For $e\in F$, let $m(e)$ be the level of $i\Confluent j$. We claim that for each $N\in\NatNum$:
\begin{equation}\label{eq_penroseNoInfluence}
 \left(\ForAll e\in F:\,m(e)\ge N\right)
 \Then
 (H_n|_{T_n})_{n=0}^N = (\tilde{H}_n|_{\tilde{T}_n})_{n=0}^N\,.
\end{equation}

We show~\eqref{eq_penroseNoInfluence} by induction over $n$. By definition $H_0|_{T_0}=\tilde{H}_0|_{\tilde{T}_0}$. For the induction step from $n<N$ to $n+1$, we show that algorithm~\ref{alg_explorationGreedy} is not influenced by the presence of such an $e\in F$. The presence of $e$ does not change $\tilde{P}_n$ compared to $P_n$ in \StepGreedyBoundary{}.\\

Observe that if $e\in\AdmissibleEdgesTree{\SchemeReturning}$, then by~\eqref{eq_adm_pen} $m(e)=k$. To show~\eqref{eq_adm_gen} we assume that $F\subseteq\AdmissibleEdgesTree{\SchemeReturning}$. If $e$ is of type~\eqref{eq_adm_pen_sameLevel}, then it is removed during iteration $(k-1)$ by \StepGreedyCousins{}. If $e$ is of type~\eqref{eq_adm_pen_uncle}, then it is removed by \StepGreedyUncles{} during iteration $k$.\\

To show~\eqref{eq_conf_gen} we assume that $F':=F\cap\ConflictingEdgesTree{\SchemePenrose}\not=\emptyset$. Choose $e\in F$ with $m(e)=N:=\min\Set{m(f): f\in F'}$. We demonstrate that the presence of $e$ causes $(H_n|_{T_n})_{n\in\NatNumZero}$ to diverge from $(\tilde{H}_n|_{\tilde{T}_n})_{n\in\NatNumZero}$ exactly at level $N+1$, that is~\eqref{eq_penroseNoInfluence} holds and $H_{N+1}|_{T_{N+1}}\not=\tilde{H}_{N+1}|_{\tilde{T}_{N+1}}$. We go through all the cases of~\eqref{eq_conf_pen}:\\

Case $e$ of type~\eqref{eq_conf_pen_diffTwo}: Here $N=k$. The presence of $e$ lets \StepGreedyBoundary{} \ActionSelect{} $i\in\tilde{P_k}$.\\

Case $e$ of type~\eqref{eq_conf_pen_uncle}: Here $N=k$. The presence of $e$ lets \StepGreedyParent{} \ActionSelect{} $i$ as the parent of $i$ in $\tilde{T}_{N+1}$ instead of $\Parent{i}$.
\end{proof}
\subsection{Pinned connected function vs one polymer partition ratios}
\label{sec_pinnedVsOPPR}
The relations in this section are either to interpreted formally or if a SCUB holds. The pinned connected function is a product of one polymer partition ratios~\cite[(3.8)]{Scott_Sokal__TheRepulsiveLatticeGasTheIndependentSetPolynomialAndTheLovaszLocalLemma__JSP_2005}:
\begin{equation}\label{eq_product_finite}
 \PinnedConnectedExpression{\FiniteVolume}{\Polymer}(\vec{z})
 = \frac{\PartitionFunction{\FiniteVolume\setminus\IncompatiblePolymers{\Polymer}}(\vec{z})}
  {\PartitionFunction{\FiniteVolume}(\vec{z})}
 = \frac{1}{\OPPR{\FiniteVolume}{\Polymer}(\vec{z})}
 \prod_{i=1}^m \frac{1}{
  \OPPR{\FiniteVolume\setminus\Set{\Polymer,\PolymerOther_1,\dotsc,\PolymerOther_{i-1}}}
  {\PolymerOther_i}(\vec{z})
 }\,,
\end{equation}
where $\Set{\PolymerOther_1,\dotsc,\PolymerOther_m} := \FiniteVolume\cap\IncompatibleOtherPolymers{\Polymer}$. On the other hand, the logarithm of the reduced correlations is an integral over the pinned connected function~\cite[(A.3)]{Bissacot_Fernandez_Proccaci_Scoppola__AnImprovementOfTheLovaszLocalLemmaViaClusterExpansion__CPC_2011}:
\begin{equation}\label{eq_integration_finite}
 \log\OPPR{\FiniteVolume}{\Polymer}(\vec{z})
 = z_\Polymer \int_0^1
  \PinnedConnectedExpression{\FiniteVolume}{\Polymer}(\vec{z}(\alpha)) \,d\alpha
\end{equation}
with
\begin{equation*}
 \vec{z}(\alpha):=\begin{cases}
  z_\PolymerOther&\text{if }\PolymerOther\not=\Polymer\\
  \alpha z_{\Polymer}&\text{if }\PolymerOther=\Polymer\,.
 \end{cases}
\end{equation*}
Taking the limit $\FiniteVolume\Exhausts\PolymerSet$, we have the following relations between the pinned series and the one polymer partition ratios in the classic case:
\begin{equation}\label{eq_product_classic}
 \PinnedSeries{\Polymer}(\vec{\rho})
 \le \frac{1}{\OPPR{\PolymerSet}{\Polymer}(-\vec{\rho})}
 \prod_{\PolymerOther\in\IncompatibleOtherPolymers{\Polymer}}
  \frac{1}{\OPPR{\PolymerSet}{\PolymerOther_i}(-\vec{\rho})}\,,
\end{equation}
Thus the finiteness of $\PinnedSeries{}(\vec{\rho})$ is equivalent to the positivity of $\OPPR{\PolymerSet}{}$. The converse relation follows from~\eqref{eq_integration_finite}:
\begin{equation}\label{eq_integration_classic}
 \OPPR{\PolymerSet}{\Polymer}(-\vec{\rho})
 \ge\exp\left(-\rho_\Polymer \PinnedSeries{\Polymer}(\vec{\rho})\right)\,.
\end{equation}

In the escaping case we have
\begin{equation}\label{eq_product_escaping}
 \PinnedSeriesEscaping{\EscapePair}(\EscapeInjection(\vec{\rho}))
 \le\frac{1}{\OPPR{\PolymersMinusEscapeSet}{\Polymer}(-\vec{\rho})}
 \prod_{\PolymerOther\in\IncompatibleOtherPolymers{\Polymer}}
  \frac{1}{\OPPR{\PolymerSet\setminus\Set{\Polymer}}{\PolymerOther_i}(-\vec{\rho})}
\end{equation}
and
\begin{equation}\label{eq_integration_escaping}
  \OPPR{\PolymersMinusEscapeSet}{\Polymer}(-\vec{\rho})
 \ge\exp\left( -\rho_\Polymer
  \PinnedSeriesEscaping{\EscapePair}(\EscapeInjection(\vec{\rho})))
 \right)\,.
\end{equation}

The relationship between classic and escaping quantities is given by the following inequalities and comparisons. The pinned series is exactly

\begin{equation}\label{eq_relation_exactPinnedSeries}
 \PinnedSeries{\Polymer}(\vec{\rho})
 = \frac{1}{\OPPR{\PolymerSet}{\Polymer}(-\vec{\rho})}
 \prod_{\PolymerOther\in\IncompatibleOtherPolymers{\Polymer}}
  \frac{1}{\OPPR{\PolymerSet\setminus\Set{\Polymer}}{\PolymerOther_i}(-\vec{\rho})}\,.
\end{equation}
The one polymer partition ratios compare as
\begin{equation}\label{eq_relation_OPPR}
 \OPPR{\PolymerSet}{\Polymer}(-\vec{\rho})
 =
 1 - \frac{\rho_\Polymer}{
  \prod_{i=1}^m
   \OPPR{\PolymerSet\setminus
      \Set{\Polymer,\PolymerOther_{i+1},\dotsc,\PolymerOther_m}
   }{\PolymerOther_i}(-\vec{\rho})
 }
 \ge 1 - \frac{\rho_\Polymer}{
  \prod_{i=1}^m
   \OPPR{\PolymerSet\setminus\Set{\Polymer}}
    {\PolymerOther_i}(-\vec{\rho})
 }\,,
\end{equation}
where $\Set{\PolymerOther_1,\dotsc,\PolymerOther_m} := \IncompatibleOtherPolymers{\Polymer}$. The pinned series fulfil the identity
\begin{equation}\label{eq_relation_pinnedSeries}
 \PinnedSeries{\Polymer}(\vec{\rho})
 = \Bigl(
  1+\rho_\Polymer
  \PinnedSeries{\Polymer}(\vec{\rho})
 \Bigr)
 \PartitionFunction{\IncompatibleOtherPolymers{\Polymer}}
  \Bigl(
   \EscapeInjection(\vec{\rho})\,
   \PinnedSeriesEscaping{(.,\Polymer)}
    (\EscapeInjection(\vec{\rho}))
  \Bigr)\,.
\end{equation}
\subsection{Depth k tree-operators}
\label{sec_depthKTreeOperators}
The aim of this section is to show a generic result for depth $k$ recursive constructions of weighted, labelled, finite trees and convergence of certain series of those trees. It generalizes~\cite[proposition $8$ and parts of proposition $7$]{Fernandez_Procacci__ClusterExpansionForAbstractPolymerModels_NewBoundsFromAnOldApproach__CMP_2007}.\\

Let $\RTsDepth{n}$ be the \emph{set of rooted, finite trees of depth $n$}, $\RTsDepthLessThan{n} := \biguplus_{m\le n} \RTsDepth{m}$ the set of rooted, finite trees of depth at most $n$ and $\RTsDepthFinite := \biguplus_{m\in\NatNumZero} \RTsDepth{m}$ the set of rooted, finite trees. Denote by $\RTPermutationCount{t}$ the number of rooted automorphisms of $t$ (they fix the root). Let $\LabelSpace$ be a countable set of labels. A function $c:\LabelSpace^{t}\to [0,\infty[$ is \emph{$t$-invariant}, if it is invariant under all rooted automorphism of $t$. We denote by $\RTLevel{t}{i}$ the \emph{$i$-th level of $t$} and by $\RTNonRootVertices{t}$ the non-root vertices of $t$. Finally, $\RTKSubtreeAt{t}{v}$ is the \emph{rooted subtree} of $t$ with root $v$ and depth $k$.\\

Let $\RTsSize{n}$ be the set of trees with vertex set $\IntRange{0,n}$ and root $0$. Consequently $\RTsSizeFinite := \biguplus_{n\ge 0} \RTsSize{n}$. For $\RTVarSize\in\RTsSizeFinite$, we denote its unlabelled version by $t(\RTVarSize)\in\RTsDepthFinite$.\\

\begin{Prop}\label{prop_depthKTreeOperator}
Fix $k\in\NatNum$. Let $(c_t)_{t\in\RTsDepthLessThan{k}}$ be a collection of non-negative functions, with each $c_t$ being $t$-invariant. Let $X:=[0,\infty]^{\LabelSpace}$ and $\vec{\rho}\in X$. Consider the operator $T_{\vec{\rho}}: X\to X$, defined $\ForAll l\in\LabelSpace$ by
\begin{subequations}\label{eq_depthKTreeOperator}
\begin{equation}\label{eq_depthKTreeOperatorDefinition}
 \vec{\mu}\mapsto
 [T_{\vec{\rho}}(\vec{\mu})]_l
 :=
  \sum_{t\in\RTsDepthLessThan{k}}
   \frac{1}{\RTPermutationCount{t}}
  \sum_{\vec{\lambda}\in \Set{l}\times\LabelSpace^{\Cardinality{\RTNonRootVertices{t}}}}
   c_t(\vec{\lambda})
  \prod_{\substack{v\in\RTLevel{t}{i}\\0\le i\le k-1}}
   \rho_{\lambda_v}
  \prod_{w\in\RTLevel{t}{k}}
   \mu_{\lambda_v}\,.
\end{equation}
If there exists $\vec{\mu}\in X$ such that
\begin{equation}\label{eq_depthKTreeOperatorVectorCondition}
 T_{\vec{\rho}}(\vec{\mu})\le\vec{\mu}\,,
\end{equation}
then the family of series, indexed by $\LabelSpace$,
\begin{equation}\label{eq_depthKTreeOperatorSeriesDefinition}
 R_l(\vec{\rho}):=
 \sum_{n\ge 0}
  \frac{1}{\Factorial{n}}
 \sum_{\vec{\lambda}\in\Set{l}\times\LabelSpace^n}
 \sum_{\RTVarSize\in\RTsSize{n}}
 \prod_{\substack{i\in\IntRange{0,n}\\d_\RTVarSize(0,i)=0\operatorname{mod}k}}
  c_{\RTVarDepth(\RTKSubtreeAt{\RTVarSize}{i})}
   (\vec{\lambda}_{\RTKSubtreeAt{\RTVarSize}{i}})
 \prod_{j=0}^n \rho_{\lambda_j}
\end{equation}
converges uniformly, as
\begin{equation}\label{eq_depthKTreeOperatorConvergenceAndMajoration}
 R(\vec{\rho})
 =\lim_{n\to\infty} T_{\vec{\rho}}^n(\vec{\rho})
 =T_{\vec{\rho}}(R(\vec{\rho}))
 \le \vec{\mu}\,.
\end{equation}
\end{subequations}
\end{Prop}

\begin{proof}
Omitting some tedious rewriting of $R(\vec{\rho})$, which can be found in~\cite[section 5.4.5]{Temmel__PropertiesAndApplicationsOfBernoulliRandomFieldsWithStrongDependencyGraphs__TUG_2012}, we obtain that $\ForAll l\in\LabelSpace$:
\begin{equation*}
 [T_{\vec{\rho}}^n(\vec{0})]_l
 =
 \sum_{\RTVarDepth\in\RTsDepthLessThan{nk-1}}
 \sum_{\vec{\lambda}\in\Set{l}\times\LabelSpace^{\RTNonRootVertices{\RTVarDepth}}}
 \sum_{m=0}^{n-1}
 \sum_{v\in\RTLevel{\RTVarDepth}{mk}}
  \frac{
   c_{\RTKSubtreeAt{\RTVarDepth}{v}}
   (\vec{\lambda}_{\RTNonRootVertices{\RTKSubtreeAt{\RTVarDepth}{v}}})
   }
   {\RTPermutationCount{\RTKSubtreeAt{\RTVarDepth}{v}}}
 \prod_{v\in\RTVarDepth} \rho_{\lambda_v}\,.
\end{equation*}
Taking the monotone increasing limit, we see that
\begin{equation*}
 R(\vec{\rho})
 = \lim_{n\to\infty} [T_{\vec{\rho}}^n(\vec{0})]\,.
\end{equation*}
The operator $T_{\vec{\rho}}$ has only non-negative coefficients, hence is order preserving. Apply proposition~\ref{prop_treeConditionLatticeLimit} together with condition~\eqref{eq_depthKTreeOperatorVectorCondition} to see that $R(\vec{\rho})$ is the least fixpoint of $T_{\vec{\rho}}$ and that~\eqref{eq_depthKTreeOperatorConvergenceAndMajoration} holds.
\end{proof}

The following proposition is an extension of arguments in~\cite[appendix]{Bissacot_Fernandez_Proccaci_Scoppola__AnImprovementOfTheLovaszLocalLemmaViaClusterExpansion__CPC_2011}:

\begin{Prop}
\label{prop_depthKTreeOperatorSeriesBounds}
Suppose that $T_{\vec{\rho}}\ge\vec{\rho}\,(\vec{1}+\vec{\mu})$ and~\eqref{eq_depthKTreeOperatorVectorCondition} holds. Decompose $T_{\vec{\rho}}=:\vec{\rho}\,(\Identity{X}+S_{\vec{\rho}})$ by splitting off the root of each tree. Likewise decompose $R(\vec{\rho})=:\vec{\rho}\,Q(\vec{\rho})$. Then $S_{\vec{\rho}}\ge\vec{0}$ is order-preserving, $\vec{\rho}<\vec{1}$ and we refine~\eqref{eq_depthKTreeOperatorConvergenceAndMajoration} to
\begin{equation}\label{eq_depthKTreeOperatorSeriesBounds}
 \ForAll l\in\LabelSpace:
 \quad[Q(\vec{\rho})]_l
 \le\frac{[S_{\vec{\rho}}(\vec{\mu})]_l}{1-\rho_l}
 <\infty\,.
\end{equation}
\end{Prop}

\begin{proof}
As $R(\vec{\rho})$ and $Q(\vec{\rho})$ increase in $\vec{\rho}$, we assume without loss of generality that $\vec{\mu}\ge\vec{\rho}>0$ or fall back on a smaller polymer system omitting the $0$ entries of $\vec{\rho}$. We  have
\begin{equation*}
 \vec{\rho}\,(\vec{1}+\vec{\mu})
 \le T_{\vec{\rho}}(\vec{\mu})\le\vec{\mu}\,,
\end{equation*}
whereby
\begin{equation*}
 \ForAll l\in\LabelSpace:\quad
 \rho_l\le\frac{\mu_l}{1+\mu_l}< 1\,.
\end{equation*}
Thus
\begin{equation*}
 \vec{\rho}\,Q(\vec{\rho})
 = R(\vec{\rho})
 = T_{\vec{\rho}}(R(\vec{\rho}))
 = \vec{\rho}\,(\Identity{X}+S_{\vec{\rho}})(R(\vec{\rho}))
 = \vec{\rho}\,R(\vec{\rho})+S_{\vec{\rho}}(R(\vec{\rho}))\,,
\end{equation*}
whence, ignoring $0$ entries of $\vec{\rho}$ and as $S_{\vec{\rho}}$ is order-preserving,
\begin{equation*}
 Q(\vec{\rho})
 = \vec{\rho}\,Q(\vec{\rho})+S_{\vec{\rho}}(R(\vec{\rho}))
 \le \vec{\rho}\,Q(\vec{\rho})+S_{\vec{\rho}}(\vec{\mu})\,.
\end{equation*}
\end{proof}
\subsection{Pinned series depth k approximation}
\label{sec_depthKTreeApproximation}
This section extends~\cite[proposition 7]{Fernandez_Procacci__ClusterExpansionForAbstractPolymerModels_NewBoundsFromAnOldApproach__CMP_2007} to level $k$ approximations. By~\eqref{eq_product_classic} an upper bound on the pinned series~\eqref{eq_pinnedSeries} is equivalent to a lower bound on the one polymer partition ratios. For basic notation see section~\ref{sec_depthKTreeOperators}.\\

\begin{Prop}\label{prop_depthKTreeApproximation}
Fix $k\in\NatNum$. Suppose that we have a family of $t$-invariant functions $(c_t)_{t\in\RTsDepthLessThan{k}}$ such that
\begin{subequations}\label{eq_depthKTreeApproximation}
\begin{equation}\label{eq_depthKTreeApproximationMajoration}
 \ForAll\vec{\PolymerOther}\in\PolymerSet^{n+1}:\quad
 \Modulus{\Ursell{\PolymerOther_0,\dotsc,\PolymerOther_n}}\le
 \sum_{\RTVarSize\in\RTsSize{n}}
 \prod_{\substack{i\in\IntRange{0,n}\\d_\RTVarSize(0,i)=0\operatorname{mod}k}}
  c_{t(\RTKSubtreeAt{\RTVarSize}{i})}
   (\vec{\PolymerOther}_{\RTKSubtreeAt{\RTVarSize}{i}})\,.
\end{equation}
Let $X:=[0,\infty]^{\PolymerSet}$ and $\vec{\rho}\in X$. Consider the operator $T_{\vec{\rho}}: X\to X$ defined $\ForAll\Polymer\in\PolymerSet$ by
\begin{equation}\label{eq_depthKTreeApproximationOperator}
 \vec{\mu}\mapsto
 [T_{\vec{\rho}}(\vec{\mu})]_\Polymer
 :=
  \sum_{t\in\RTsDepthLessThan{k}}
   \frac{1}{\RTPermutationCount{t}}
  \sum_{\vec{\PolymerOther}\in \Set{\Polymer}\times\PolymerSet^{\Cardinality{\RTNonRootVertices{t}}}}
   c_t(\vec{\PolymerOther})
  \prod_{\substack{v\in\RTLevel{t}{i}\\0\le i\le k-1}}
   \rho_{\PolymerOther_v}
  \prod_{w\in\RTLevel{t}{k}}
   \mu_{\PolymerOther_v}\,.
\end{equation}
If there exists $\vec{\mu}\in X$ such that
\begin{equation}\label{eq_depthKTreeApproximationVectorCondition}
 T_{\vec{\rho}}(\vec{\mu})\le\vec{\mu}\,,
\end{equation}
then $\vec{\rho}<\vec{1}$ and
\begin{equation}\label{eq_depthKTreeApproximationConvergence}
 \vec{\rho}\,\PinnedSeries{}(\vec{\rho})\le\vec{\mu}\,.
\end{equation}
This implies that $\vec{\rho}\in\InteriorOf\PDPS$ and $\vec{\rho}<\vec{1}$.
\end{subequations}
\end{Prop}

\begin{proof}
Setting $\LabelSpace:=\PolymerSet$ in~\eqref{eq_depthKTreeOperatorSeriesDefinition}, condition~\eqref{eq_depthKTreeApproximationMajoration} implies that
\begin{multline*}
 [\vec{\rho}\,\PinnedSeries{}(\vec{\rho})]_{\Polymer}
 =\sum_{n\ge 0}
  \frac{1}{\Factorial{n}}
 \sum_{\vec{\PolymerOther}\in\Set{\Polymer}\times\PolymerSet^n}
  \Modulus{\Ursell{\vec{\PolymerOther}}}
 \prod_{i=0}^n \rho_{\PolymerOther_i}\\
 \le\sum_{n\ge 0}
  \frac{1}{\Factorial{n}}
 \sum_{\vec{\PolymerOther}\in\Set{\Polymer}\times\PolymerSet^n}
 \sum_{\RTVarSize\in\RTsSize{n}}
 \prod_{\substack{i\in\IntRange{0,n}\\d_\RTVarSize(0,i)=0\operatorname{mod}k}}
  c_{\RTVarDepth(\RTKSubtreeAt{\RTVarSize}{i})}
   (\vec{\PolymerOther}_{\RTKSubtreeAt{\RTVarSize}{i}})
 \prod_{j=0}^n \rho_{\PolymerOther_j}
 = [R(\vec{\rho})]_{\Polymer}\,.
\end{multline*}
Then~\eqref{eq_depthKTreeApproximationConvergence} follows from~\eqref{eq_depthKTreeOperatorConvergenceAndMajoration}. The relation~\eqref{eq_integration_classic} implies that $\vec{\rho}\in\InteriorOf\PDPS$. Finally, if $\rho_\Polymer\ge 1$, then~\eqref{eq_oppr_mon_volume} and the fundamental identity~\eqref{eq_fi_oppr} imply that $\OPPR{\PolymerSet}{\Polymer}(-\vec{\rho})\le 0$, a contradiction to $\vec{\rho}\in\InteriorOf\PDPS$.
\end{proof}
\subsection{One polymer partition ratios on a homogeneous tree}
\label{sec_oppr_homTree}
We give an example of a polymer system, for which the analyticity of the one polymer partition ratios breaks down in the limit at a boundary point of $\PDPS$, at a non-physical negative real fugacity.\\

\begin{Exam}\label{exam_unboundednessOnTheHomogeneousTree}
Let $(\PolymerSet,\IncompatibleTo)$ be isomorph to an infinite $D$-regular tree, modulo the loops at each vertex. Fix an end $E$, that is an equivalence class of rays, of the tree. For $\Polymer\in\PolymerSet$, let $F_\Polymer$ be the set of polymers further away from $E$ than $\Polymer$. These are the polymers, for which the ray starting at the polymer and being equivalent to $E$, passes through $\Polymer$. We have the decomposition
\begin{subequations}\label{eq_unboundednessOnTheHomogeneousTree}
\begin{equation}
 F_\Polymer :=
 \Set{\Polymer}\uplus
 \biguplus_{\PolymerOther\in\IncompatibleOtherPolymers{\Polymer}\setminus\Set{\EscapePolymer}}
 F_{\PolymerOther}\,,
\end{equation}
where $\EscapePolymer$ is the unique polymer incompatible with $\Polymer$ which is closer to $E$ than $\Polymer$. Homogeneous fugacity $\rho\vec{1}$ and the transitivity of $(\PolymerSet,\IncompatibleTo)$ imply that
\begin{equation}
 \alpha(\rho):=
 \OPPR{F_\Polymer}{\Polymer}(-\rho\vec{1})
\end{equation}
is well-defined for $\rho\le\PowerFracDualDMinusOne:=\rho^\star$~\cite{Shearer__OnAProblemOfSpencer__Comb_1985}. The fundamental identity~\eqref{eq_fi_oppr} yields
\begin{equation}
 \alpha(\rho) = 1 - \frac{\rho}{\alpha(\rho)^{D-1}}\,.
\end{equation}
By~\eqref{eq_oppr_mon_volume} we have the limit
\begin{equation}
 \lim_{\rho\to\rho^\star} \alpha(\rho) = \alpha(\rho^\star) = \frac{D-1}{D}\,.
\end{equation}
But the derivative is
\begin{equation}
 \left(\frac{\partial\alpha}{\partial z}\right)(\rho)
 = \frac{1}{\alpha(\rho)^{D-2} D[\alpha(\rho^\star)-\alpha(\rho)]}
\end{equation}
and diverges as we approach $\rho^\star$
\begin{equation}
 \lim_{\rho\to\rho^\star}
 \left(\frac{\partial\alpha}{\partial z}\right)(\rho)
 = \infty\,.
\end{equation}
\end{subequations}
\end{Exam}

\bibliographystyle{plain}
\bibliography{ref}

\end{document}